%% file: journal_SM_CS.tex
\documentclass[draftclsnofoot,onecolumn]{IEEEtran}
\usepackage{lipsum}
\usepackage{amsthm}
\usepackage{framed} 

\newtheoremstyle{mystyle}%                % Name
  {}%                                     % Space above
  {}%                                     % Space below
  {\itshape}%                             % Body font
  {}%                                     % Indent amount
  {\bfseries}%                            % Theorem head font
  { }%                                    % Punctuation after theorem head
  {.5em}%                                    % Space after theorem head, ' ', or \newline
  {}%                                     % Theorem head spec (can be left empty, meaning `normal')

\theoremstyle{mystyle}

\usepackage{ifpdf}
\usepackage{cite}
\ifCLASSINFOpdf
\else
\fi
\usepackage{fixltx2e}
\usepackage{stfloats}
\usepackage[dvipsnames]{xcolor}
\usepackage[colorlinks=False,linkcolor=red]{hyperref}
\usepackage{amsmath,dsfont,bbm,epsfig,amssymb,amsfonts,amstext,verbatim,amsopn,cite,subfigure,multirow,multicol,lipsum,xfrac}
\usepackage{amsthm}
\usepackage{thmtools}
\usepackage{mathtools}
\usepackage{perpage}
\usepackage{balance}
\usepackage{url}
\usepackage{amsfonts}
\usepackage{epsfig}
\usepackage[font={small}]{caption}
\usepackage{slashbox}
\usepackage{setspace}
\usepackage{stmaryrd}
\usepackage{psfrag}	
\usepackage{multirow}
\usepackage{slashbox}
\usepackage[process=auto]{pstool}
\usepackage{etoolbox}
\usepackage{algorithmicx}
\usepackage[Algorithm,ruled]{algorithm}
\usepackage{algpseudocode}
\usepackage{pifont}
\usepackage[nolist]{acronym}
\MakePerPage{footnote}
\usepackage{paralist}
\usepackage{enumitem}
\usepackage{textgreek}
\usepackage{upgreek}

\usepackage{tikz}
\usetikzlibrary{shapes,arrows}

\include{commands}

\begin{document}
\title{Statistical Mechanics of MAP Estimation: General Replica Ansatz}

\author{Ali~Bereyhi,
        Ralf R. M\"uller,
        and~Hermann Schulz-Baldes%\vspace{-8mm}% <-this % stops a space
\thanks{The results of this manuscript were presented in parts at 2016 IEEE Information Theory Workshop (ITW) \cite{bereyhi2016itw} and 2017 IEEE Information Theory and Applications Workshop (ITA) \cite{bereyhi2017replica}.}
\thanks{This work was supported by the German Research Foundation, Deutsche Forschungsgemeinschaft (DFG), under Grant No. MU 3735/2-1.}
\thanks{Ali Bereyhi and Ralf R. M\"uller are with the Institute for Digital Communications (IDC), Friedrich Alexander University of Erlangen-N\"urnberg (FAU), Konrad-Zuse-Stra{\ss}e 5, 91052, Erlangen, Bavaria, Germany (e-mails: ali.bereyhi@fau.de, ralf.r.mueller@fau.de).}% <-this % stops a space
\thanks{Hermann Schulz-Baldes is with the Department of Mathematics, FAU, Cauerstra{\ss}e 11, 91058, Erlangen, Bavaria, Germany (e-mail: schuba@ mi.uni-erlangen.de).}
}% <-this % stops a space
%\thanks{Manuscript received April 19, 2005; revised September 17, 2014.}}


\markboth{}%
{ }



\maketitle

\tikzstyle{block} = [draw,  rectangle, rounded corners, minimum height=2.5em, minimum width=5em]
\tikzstyle{margin} = [draw, dotted, rectangle, rounded corners, minimum height=2.1em, minimum width=2em]
\tikzstyle{sum} = [draw, circle, node distance=1cm, inner sep=0pt]
\tikzstyle{state} = [draw, circle, node distance=1cm, inner sep=1pt, minimum size=.8cm]
\tikzstyle{input} = [coordinate]
\tikzstyle{output} = [coordinate]
\tikzstyle{pinstyle} = [pin edge={to-,thick,black}]

\begin{abstract}
The large-system performance of maximum-a-posterior estimation is studied considering a general distortion function when the observation vector is received through a linear system with additive white Gaussian noise. The analysis considers the system matrix to be chosen from the large class of rotationally invariant random matrices. We take a statistical mechanical approach by introducing a spin glass corresponding to the estimator, and employing the replica method for the large-system analysis. In contrast to earlier replica based studies, our analysis evaluates the general replica ansatz of the corresponding spin glass and determines the asymptotic distortion of the estimator for any structure of the replica correlation matrix. Consequently, the replica symmetric as well as the replica symmetry breaking ansatz with $b$ steps of breaking is deduced from the given general replica ansatz. The generality of our distortion function lets us derive a more general form of the maximum-a-posterior decoupling principle. Based on the general replica ansatz, we show that for any structure of the replica correlation matrix, the vector-valued system decouples into a bank of equivalent decoupled linear systems followed by maximum-a-posterior estimators. The structure of the decoupled linear system is further studied under both the replica symmetry and the replica symmetry breaking assumptions. For $b$ steps of symmetry breaking, the decoupled system is found to be an additive system with a noise term given as the sum of an independent Gaussian random variable with $b$ correlated impairment terms. The general decoupling property of the maximum-a-posterior estimator leads to the idea of a replica simulator which represents the replica ansatz through the state evolution of a transition system described by its corresponding decoupled system. As an application of our study, we investigate large compressive sensing systems by considering the $\ell_p$ norm minimization recovery schemes. Our numerical investigations show that the replica symmetric ansatz for $\ell_0$ norm recovery fails to give an accurate approximation of the mean square error as the compression rate grows, and therefore, the replica symmetry breaking ans\"atze are needed in order to assess the performance precisely.
\end{abstract}


%\begin{IEEEkeywords}
%Replica method, Decoupling.
%\end{IEEEkeywords}

% \fi
%\IEEEpeerreviewmaketitle

%\section{Introduction}
%\IEEEPARstart{T}{his} 


 
%\hfill September 17, 2014
\section{Introduction}
\label{sec:introduction}
Consider a vector-valued \ac{awgn} system specified by
\begin{align}
\by=\mA \bx + \bz \label{eq:sys-1}
\end{align}
where the \ac{iid} source vector $\bx_{n \times 1}$, with components in a support set $\setX \subset \setR$, is measured by the random system matrix $\mA_{k \times n} \in \setA^{k \times n}$, with $\setA\subset \setR$, and corrupted by the \ac{iid} zero-mean Gaussian noise vector $\bz_{k \times 1}$, with variance $\lambda_0$, i.e., $\bz \sim \man(\boldsymbol{0},\lambda_0 \mI)$. The source vector can be estimated from the observation vector $\by_{k \times 1}$ using a \ac{map} estimator. For a given system matrix $\mA$, the estimator maps the observation vector to the estimated vector $\bhx_{n \times 1} \in \setX^n$ via the estimation function~$\bgg(\cdot | \mA)$~defined~as
\begin{align}
\bgg(\by | \mA)= \arg \min_{\bv \in \setX^n} \ \left[ \frac{1}{2\lambda} \norm{\by-\mA \bv}^2 + u(\bv) \right]  \label{eq:int-2}
\end{align}
for some ``utility function'' $u(\cdot): \setR^n \mapsto \setR^+$ and estimation parameter $\lambda \in \setR^+$. In \eqref{eq:int-2}, $\norm{\cdot}^2$ denotes the Euclidean norm, and it is assumed that the minimum is not degenerate so that $\bgg(\cdot | \mA)$ is well-defined, at least for almost all $\by$ and $\mA$. In order to analyze the performance of the system in the large-system limit, i.e., $k,n \uparrow \infty$, one considers a general distortion function $\sfd(\cdot;\cdot): \setX \times \setX \mapsto \setR$. For some choices of $\sfd(\cdot;\cdot)$, the distortion function determines the distance between the source and estimated vector, e.g. $\sfd(\hx;x)=\abs{\hx-x}^2$; however, in general, it takes different choices. The asymptotic distortion
\begin{align}
\sfD = \lim_{n \uparrow \infty} \frac{1}{n} \sum_{i=1}^n \sfd(\hx_i;x_i), \label{eq:int-2.1}
\end{align}
then, expresses the large-system performance regarding the distortion function $\sfd(\cdot;\cdot)$. The performance analysis of the estimator requires \eqref{eq:int-2} to be explicitly computed, and then, $\bhx=\bgg(\by|\mA)$ substituted in the distortion function. This task, however, is not trivial for many choices of the utility function and the source support $\setX$, and becomes unfeasible as $n$ grows large. As basic analytic tools fail, we take a statistical mechanical approach and investigate the large-system performance by studying the macroscopic parameters of a corresponding spin glass. This approach enables us to use the replica method which has been developed in the context of statistical mechanics.% for analysis of spin glasses.
\subsection{Corresponding Spin Glass}
\label{sec:spin_glasses}
Consider a thermodynamic system which consists of $n$ particles with each having a microscopic parameter $v_i \in \setV $. The vector $\bv_{n \times 1}= \left[ v_1,\ldots,v_n\right]^\trp$, collecting the microscopic parameters, presents then the microscopic state of the system and is called the ``microstate''. The main goal of statistical mechanics is to excavate the ``macroscopic parameters'' of the system, such as energy and entropy through the analysis of the microstate in the thermodynamic limit, i.e., $n \uparrow \infty$. Due to the large dimension of the system, statistical mechanics proposes a stochastic approach in which the microstate is supposed to be randomly distributed over the support $\setV^n$ due to some distribution $\rmp_{\bv}$. For this system, the Hamiltonian $\mae(\cdot): \setR^n \mapsto \setR^+$ assigns to each realization of the microstate a non-negative energy level, and $\rmH \coloneqq -\E_{\rmp_{\bv}} \log \rmp_{\bv}$ denotes the system's entropy. The ``free energy'' of the thermodynamic system at the inverse temperature $\upbeta$ is then defined as
\begin{align}
\sfF(\upbeta) \coloneqq \E_{\rmp_{\bv}} \mae(\bv)-\upbeta^{-1} \rmH. \label{eq:int-3}
\end{align}
The second law of thermodynamics states that the microstate at thermal equilibrium takes its distribution such that the free energy meets its minimum. Thus, the microstate's distribution at thermal equilibrium reads
\begin{align}
\rmp^{\upbeta}_{\bv}(\bv)= \left[ \maz(\upbeta)\right]^{-1} e^{-\upbeta \mae(\bv)} \label{eq:int-4}
\end{align}
where $\maz(\upbeta)$ is a normalization factor referred to as the ``partition function'', and the superscript $\upbeta$ indicates the distribution's dependence on the inverse temperature. The distribution in \eqref{eq:int-4} is known as the ``Boltzmann-Gibbs distribution'' and covers many distributions on $\setV^n$ by specifying $\mae(\cdot)$ and $\upbeta$ correspondingly. Substituting the Boltzmann-Gibbs distribution in \eqref{eq:int-3}, the free energy at thermal equilibrium and inverse temperature $\upbeta$ reads
\begin{align}
\sfF(\upbeta) = -\upbeta^{-1} \log \maz(\upbeta). \label{eq:int-5}
\end{align}
The average energy and entropy of the system at thermal equilibrium are then determined by taking expectation over the distribution in \eqref{eq:int-4}, i.e.,
\begin{subequations}
\begin{align}
\rmE(\upbeta) &\coloneqq  \E_{\rmp_{\bv}^\upbeta} \mae(\bv)\\
\rmH(\upbeta) &\coloneqq -\E_{\rmp_{\bv}^\upbeta} \log \rmp^{\upbeta}_{\bv}(\bv), 
\end{align}
\end{subequations}
which can be calculated in terms of the free energy via
\begin{subequations}
\begin{align}
\rmE(\upbeta) &= \frac{\dif}{\dif \upbeta} \left[ \upbeta \sfF(\upbeta) \right] \label{eq:int-5.1a} \\
\rmH(\upbeta) &= \upbeta^2 \frac{\dif}{\dif \upbeta} \left[ \sfF(\upbeta) \right]. \label{eq:int-5.1b}
\end{align}
\end{subequations}

In spin glasses \cite{edwards1975theory}, the Hamiltonian assigns the energy levels randomly using some randomizer $\sfQ$ resulting from random interaction coefficients. % in the thermodynamic system which might be caused by the physical system's nature or imperfect measurement of parameters.
 In fact, each realization of $\sfQ$ specifies a thermodynamic system represented by the deterministic Hamiltonian $\mae(\cdot|\sfQ)$. %which does not change by the variation of microstate's realizations. 
In statistical mechanics, $\sfQ$ is known to have ``quenched'' randomness while the microstate is an ``annealed'' random variable. The analysis of spin glasses takes similar steps as above considering a given realization of the randomizer, and therefore, as the system converges to its thermal equilibrium at the inverse temperature $\upbeta$, the microstate's conditional distribution given $\sfQ$, i.e., $\rmp^\upbeta_{\bv|\sfQ}$, is a Boltzmann-Gibbs distribution specified by $\mae(\cdot|\sfQ)$. Consequently, the free energy reads
\begin{align}
\sfF(\upbeta|\sfQ) = -\upbeta^{-1} \log \maz(\upbeta|\sfQ) \label{eq:int-6}.
\end{align}
where $\maz(\upbeta|\sfQ)$ is the partition function with respect to the Hamiltonian $\mae(\cdot|\sfQ)$. Here, the free energy, as well as other macroscopic parameters of the system, is random; however, the physical intuition behind the analyses suggests that these random macroscopic parameters converge to deterministic values in the thermodynamic limit. This property is known as the ``self averaging property'' and has been rigorously justified for some particular classes of Hamiltonians, e.g., \cite{pastur1991absence,guerra2002thermodynamic,guerra2002infinite,korada2010tight}. Nevertheless, in cases where a mathematical proof is still lacking, the property is supposed to hold during the analysis. According to the self averaging property, the free energy of spin glasses converges to its expectation in the thermodynamic limit.

As mentioned earlier, the \ac{map} estimator in \eqref{eq:int-2} can be investigated using a corresponding spin glass. To see that, consider a spin glass whose microstate is taken from $\setX^n$, and whose Hamiltonian is defined as
\begin{align}
\mae(\bv|\by,\mA)=\frac{1}{2\lambda} \norm{\by-\mA \bv}^2 + u(\bv). \label{eq:int-7}
\end{align}
Here, the system matrix $\mA$ and the observation vector $\by$ are considered to be the randomizers of the spin glass. In this case, given $\mA$ and $\by$, the conditional distribution of the microstate is given by
\begin{align}
\rmp^{\upbeta}_{\bv|\by,\mA}(\bv|\by,\mA)= \left[ \maz(\upbeta|\by,\mA)\right]^{-1} e^{-\upbeta \mae(\bv|\by,\mA)}. \label{eq:int-8}
\end{align}
Taking the limit when $\upbeta \uparrow \infty$ and using Laplace method of integration \cite{merhav2010statistical}, the zero temperature distribution, under this assumption that the minimizer is unique, reduces to
\begin{subequations}
\begin{align}
\rmp^{\infty}_{\bv|\by,\mA}(\bv|\by,\mA)&= \mone \{ \bv = \arg \min_{\bv \in \setX^n} \mae(\bv|\by,\mA)\}\label{eq:int-9a} \\
&=\mone \{ \bv = \bgg(\by | \mA) \}, \label{eq:int-9b}
\end{align}
\end{subequations}
where $\mone\{\cdot\}$ denotes the indicator function, and $\bgg(\cdot |\mA)$ is defined as in \eqref{eq:int-2}. \eqref{eq:int-9b} indicates that the microstate of the spin glass converges to the estimated vector of the \ac{map} estimator, i.e., $\bhx=\bgg(\by|\mA)$, in the zero temperature limit. Invoking this connection, we study the corresponding spin glass instead of the \ac{map} estimator. We represent the input-output distortion of the system regarding a general distortion function $\sfd(\cdot;\cdot)$ as a macroscopic parameter of the spin glass. Consequently, the replica method developed in statistical mechanics is employed to determine the defined macroscopic parameter of the corresponding spin glass. The replica method is a generally nonrigorous but effective method developed in the physics literature to study spin glasses. Although the method lacks rigorous mathematical proof in some particular parts, it has been widely accepted as an analysis tool and utilized to investigate a variety of problems in applied mathematics, information processing, and coding \cite{mezard1986replica,fu1986application,nishimori2001statistical,montanari2000turbo}. 
%results given in \cite{tanaka2002statistical}, in some special cases, recovered known rigorous results

The use of the replica method for studying multiuser estimators goes back to \cite{tanaka2002statistical} where Tanaka determined the asymptotic spectral efficiency of \ac{mpm} estimators by employing the replica method. The study demonstrated interesting large-system properties of multiuser estimators, and consequently, the statistical mechanical approach received more attention in the context of multiuser systems. This approach was then employed in the literature to study multiple estimation problems in large vector-valued linear systems, e.g. \cite{tanaka2001average,guo2003multiuser,guo2005randomly}. The method was also utilized to analyze the asymptotic properties of \ac{mimo} systems in \cite{muller2003channel} considering an approach similar to \cite{tanaka2002statistical}. Regarding multiuser estimators, the earlier studies mainly considered the cases in which the entries of the source vector are binary or Gaussian random variables. The results were later extended to a general source distribution in \cite{guo2005randomly}. The statistical mechanical approach was further employed to address mathematically similar problems in vector precoding, compressive sensing and analysis of superposition codes \cite{muller2008vector,guo2009single,barbier2014replica}, to name just a few examples. Despite the fact that the replica method lacks mathematical rigor, a body of work, such as \cite{montanari2006analysis,huleihel2017asymptotic,korada2010tight,reeves2016replica, barbier2016mutual,barbier2017mutual,barbier2016threshold,barbier2017universal}, has shown the validity of several replica-based results in the literature, e.g., Tanaka's formula in \cite{tanaka2002statistical}, using some alternative rigorous approaches. We later discuss these rigorous results with more details by invoking the literature of compressive sensing.% \cite{montanari2006analysis,barbier2017mutual}

\subsection{Decoupling Principle}
\label{subsec:decoupling}
Considering the \ac{map} estimator defined in \eqref{eq:int-2}, the entries of the estimated vector $\bhx$ are correlated in general, since the system matrix couples the entries of $\bx$ linearly, and $\bgg(\cdot|\mA)$ performs several nonlinear operations on $\by$. In the large-system performance analysis, the marginal joint distribution of two corresponding input-output entries $x_j$ and $\hx_j$, $1\leq j\leq n$, is of interest. To clarify our point, consider the case in which a linear estimator is employed instead of \eqref{eq:int-2}, i.e., $\bhx=\mG^{\trp} \by$. Denote the matrices $\mA$ and $\mG$ as $\mA=[\baa_1 \cdots \baa_n]$ and $\mG=[\bgg_1 \cdots \bgg_n]$, respectively, with $\baa_i$ and $\bgg_i$ being $k \times 1$ vectors for $1\leq i\leq n$. Therefore, $\hx_j$ is written as
\begin{subequations}
\begin{align}
\hx_j &= \bgg_j^\trp \by \label{eq:int-10a} \\
&=\bgg_j^\trp \left[ \sum_{i=1}^n x_i \baa_i + \bz \right] \label{eq:int-10b}\\
&=(\bgg_j^\trp \baa_j) \ x_j + \sum_{i=1, i\neq j}^n (\bgg_j^\trp \baa_i) \ x_i + \bgg_j^\trp \bz. \label{eq:int-10c}
\end{align}
\end{subequations}
Here, the right hand side of \eqref{eq:int-10c} can be interpreted as the linear estimation of a single-user system indexed by $j$ in which the symbol $x_j$ is corrupted by an additive impairment given by the last two summands in the right hand side of \eqref{eq:int-10c}. The impairment term is not necessarily independent and Gaussian. For some classes of matrix ensembles, and under a set of assumptions, it is shown that the dependency of the derived single-user systems on the index $j$ vanishes, and the distribution of the impairment terms converges to a Gaussian distribution in the large-system limit \cite{guo2002asymptotic}. As a result, one can assume the vector-valued system described by \eqref{eq:sys-1} followed by the linear estimator $\mG$ to be a set on $n$ additive scalar systems with Gaussian noise which have been employed in parallel. In other words, the vector system can be considered to ``decouple'' into a set of similar scalar systems. Each of them relates an input entry $x_j$ to its corresponding estimated one $\hx_j$. This asymptotic property of the estimator is referred to as the ``decoupling property'' and can be investigated through the large-system performance analysis.

The decoupling property was first studied for linear estimators. Tse and Hanly noticed this property while they were determining the multiuser efficiency of several linear multiuser estimators in the large-system limit \cite{tse1999linear}. They showed that for an \ac{iid} system matrix, the effect of impairment is similar to the effect of some modified Gaussian noise when the dimension tends to infinity. This asymptotic property was then investigated further by studying the asymptotics of different linear receivers and their large-system distributions \cite{tse2000linear, eldar2003asymptotic}. In an independent work, Verd\'u and Shamai also studied the linear \ac{mmse} estimator and showed that the conditional output distribution is asymptotically Gaussian \cite{verdu1999spectral}. In \cite{zhang2001output}, the authors studied the asymptotics of the impairment term when a family of linear estimators is employed and proved that it converges in distribution to a Gaussian random variable. The latter result was further extended to a larger class of linear estimators in \cite{guo2002asymptotic}.

Regarding linear estimators, the main analytical tool is random matrix theory \cite{tulino2004random,muller2013applications}. In fact,~invoking~properties of large random matrices and the central limit theorem, the decoupling property is rigorously proved, e.g. \cite{guo1999linear,shamai2001impact}. These tools, however, fail for nonlinear estimators as the source symbol and impairment term do not decouple linearly due to nonlinear operations at the estimators. In \cite{muller2004capacity}, M\"uller and Gerstacker employed the replica method and studied the capacity loss due to the separation of detection and decoding. The authors showed that the additive decoupling of the spectral efficiency, reported in \cite{shamai2001impact} for Gaussian inputs, also holds for binary inputs. As a result,~it was conjectured that regardless of input distribution and linearity, the spectral efficiency always decouples in an additive form \cite{muller2002channel}. In \cite{guo2005randomly}, Guo and Verd\'u justified this conjecture for a family of nonlinear \ac{mmse} estimators, and showed that for an \ac{iid} system matrix, the estimator decouples into a bank of single-user \ac{mmse} estimators under the \ac{rs} ansatz. In \cite{rangan2012asymptotic}, Rangan et al. studied the asymptotic performance of a class of \ac{map} estimators. Using standard large deviation arguments, the authors represented the \ac{map} estimator~as the limit of an indexed \ac{mmse} estimators' sequence. Consequently, they determined the estimator's asymptotics employing the results from \cite{guo2005randomly} and justified the decoupling property of \ac{map} estimators under the \ac{rs} ansatz for an \ac{iid} $\mA$. %

Regarding the decoupling property of \ac{map} estimators, there are still two main issues which need further~investigations:
\begin{inparaenum}
\item cases in which the system matrix $\mA$ is not \ac{iid}, and
\item the analysis of the estimator under the \ac{rsb} ans\"atze.
\end{inparaenum}
The first issue was partially addressed in \cite{tulino2013support} where, under the \ac{rs} assumption, the authors studied the asymptotics of a \ac{map} estimator employed to recover the support of a source vector from observations received through noisy sparse random measurements. They considered a model in which a sparse Gaussian source vector\footnote{It means that the entries of the source vector are of the form $x_i b_i$ where $x_i$ and $b_i$ are Gaussian and Bernoulli random variables, respectively.} is first randomly measured by a square matrix $\mV$, and then, the measurements are sparsely sampled by a diagonal matrix $\mB$ whose non-zero entries are \ac{iid} Bernoulli random variables. For this setup, the input-output information rate and support recovery error rate were investigated by considering the measuring matrix $\mV$ to belong to a larger set of matrix ensembles. These results, moreover, could address the decoupling property of the considered setting. Although the class of system matrices is broadened in \cite{tulino2013support}, it cannot be considered as a complete generalization of the property presented in \cite{guo2005randomly} and \cite{rangan2012asymptotic}, since it is restricted to cases with a sparse Gaussian source and loading factors less than one, i.e., $kn^{-1}<1$ in \eqref{eq:sys-1}. Vehkaper\"a et al. also tried to investigate the first issue for a similar formulation in compressive sensing \cite{vehkapera2014analysis}. In fact, the authors considered a linear sensing model as in \eqref{eq:sys-1} for the class of rotationally invariant random matrices\footnote{The class of rotationally invariant random matrices is precisely defined later throughout the problem formulation.} and under the \ac{rs} ansatz determined the asymptotic \ac{mse} for the least-square recovery schemes which can be equivalently represented by the formulation in \eqref{eq:int-2}. The large-system results in \cite{vehkapera2014analysis}, however, did not address the asymptotic marginal joint input-output distribution, and the emphasis was on the \ac{mse}. Regarding the second issue, the \ac{map} estimator has not yet been investigated under \ac{rsb} ans\"atze in the literature. Nevertheless, the necessity of such investigations was mentioned for various similar settings in the literature; see for example \cite{yoshida2007statistical,kabashima2009typical,zaidel2012vector}. In \cite{yoshida2007statistical}, the performances of \ac{cdma} detectors were investigated by studying both the \ac{rs} and one-step \ac{rsb} ans\"atze and the impact of symmetry breaking onto the results for low noise scenarios were discussed. The authors in \cite{zaidel2012vector} further studied the performance of vector precoding under both \ac{rs} and \ac{rsb} and showed that the analysis under \ac{rs} yields a significantly loose bound on the true performance. The replica ansatz with one-step of \ac{rsb}, however, was shown to lead to a tighter bound consistent with rigorous performance bound available in the literature. A similar observation was recently reported for the problem of least-square-error precoding in \cite{sedaghat2016lse,bereyhi2017asymptotics}. The replica analyses of compressive sensing in \cite{kabashima2009typical,takeda2011statistical}, moreover, discussed the necessity of investigating the performance of $\ell_p$ minimization recovery schemes under \ac{rsb} ans\"atze for some choices of $p$.

%by studying large compressive sensing systems via the replica method
\subsection{Compressive Sensing}
\label{subsec:compressive_sensing}
The \ac{map} estimation of a source vector from a set of noisy linear observations arises in several applications, such as \ac{mimo} and sampling systems. To address one, we consider large compressive sensing systems and employ our asymptotic results to analyze the large-system performance \cite{donoho2006compressed,candes2006robust,candes2006near}. In context of compressive sensing, \eqref{eq:sys-1} represents a noisy sampling system in which the source vector $\bx$ is being sampled linearly via $\mA$ and corrupted by additive noise $\bz$. In the ``noise-free'' case, i.e. $\lambda_0 = 0$, the source vector $\bx$ is exactly recovered from the observation vector $\by$, if the number of observations $k$ is as large as the source length $n$ and the sampling matrix $\mA$ is full rank. As the number of observations reduces, the possible answers to the exact reconstruction problem may increase regarding the source support $\setX$, and therefore, the recovered source vector from the observation vector is not necessarily unique. In this case, one needs to enforce some extra constraints on the properties of the source vector in order to recover it uniquely among all possible solutions. %
In compressive sensing, the source vector is supposed to be sparse, i.e., a certain fraction of entries are zero. This property of the source imposes an extra constraint on the solution which allows for exact recovery in cases with $k < n$. In fact, in this case, one should find a solution to $\by=\mA \bv$ over
\begin{align}
\setS=\left\lbrace \bv_{n\times 1} \in \setX^n: \ \norm{\bv}_0 <\alpha n \right\rbrace \label{eq:cs-zero}
\end{align}
where $\norm{\cdot}_0$ denotes the ``$\ell_0$ norm'' and is defined as $\norm{\bv}_0 \coloneqq \sum_{i=1}^n \mone\{v_i \neq 0\}$, and $\alpha \leq 1$ is the source's sparsity factor defined as the fraction of non-zero entries. Depending on $\mA$ and $\setX$, the latter problem can have a unique solution even for $k < n$ \cite{donoho2001uncertainty,elad2002generalized,donoho2003optimally}. Searching for this solution optimally over $\setS$, however, is an NP-hard problem and therefore intractable. The main goal in noise-free compressive sensing is to study feasible reconstruction schemes and derive tight bounds on the sufficient compression rate, i.e., $k/n$, for exact source recovery via these schemes. %.

In noisy sampling systems, exact recovery is only possible for some particular choices of $\setX$. Nevertheless, considering either cases in which exact recovery is not possible or choices of $\setX$ for which the source vector can be exactly recovered from noisy observations, the recovery approaches in these sensing systems need to take the impact of noise into account. The classical strategy in this case is to find a vector in $\setS$ such that the recovery distortion is small. Consequently, a recovery scheme for noisy sensing system based on the $\ell_0$ norm is given by
\begin{align}
\bhx=\arg \min_{\bv \in \setX^n} \ \left[ \frac{1}{2\lambda} \norm{\by-\mA \bv}^2 + \norm{\bv}_0 \right] \label{eq:int-12}
\end{align}
which is the \ac{map} estimator defined in \eqref{eq:int-2} with $u(\bv)=\norm{\bv}_0$. It is straightforward to show that for $\lambda_0=0$, i.e., zero noise variance, \eqref{eq:int-12} reduces to the optimal noise-free recovery scheme as $\lambda \downarrow 0$. Similar to the noise-free case, the scheme in \eqref{eq:int-12} results in a non-convex optimization problem, and therefore, is computationally infeasible. Alternatively, a computationally feasible schemes is introduced by replacing the $\ell_0$ norm in the cost function with the $\ell_1$ norm. The resulting recovery scheme is known as LASSO \cite{tibshirani1996regression} or basis pursuit denoising \cite{chen2001atomic}. Based on these formulations, several iterative and greedy algorithms have been introduced for recovery taking into account the sparsity pattern and properties of the sampling matrix \cite{needell2009cosamp,dai2009subspace,cai2011orthogonal}. The main body of work in noisy compressive sensing then investigates the trade-off between the compression rate and recovery distortion. % In most of the cases, the asymptotic performance is of interest. In fact, due to the self averaging property, for almost any realization of the system matrix the reconstructed vector converges to a deterministic vector in the large-system limit, and therefore, the system can be deterministically characterized.
%The approach was taken further in \cite{rangan2012asymptotic} under the \ac{rs} assumption to study the asymptotics of the $\ell_0$ norm, LASSO and linear schemes.

For large compressive sensing systems, it is common to consider a random sensing matrix, since for these matrices, properties such as restricted isometry property are shown to hold with high probability \cite{foucart2013mathematical}. In this case, the performance of a reconstruction schemes is analyzed by determining the considered performance metric, e.g., \ac{mse} and probability of exact recovery in the noisy and noise-free case, respectively, for a given realization of the sensing matrix. The average performance is then calculated by taking the expectation over the matrix distribution.~Comparing \eqref{eq:int-12} with \eqref{eq:int-2}, one can utilize the \ac{map} formulation, illustrated at the beginning of this section, to study the large-system performance of several reconstruction schemes. This similarity was considered in a series of papers, e.g., \cite{guo2009single,rangan2012asymptotic}, and therefore, earlier replica results were employed to study compressive sensing systems. The extension of analyses from the context of multiuser estimation had the disadvantage that the assumed sampling settings were limited to those setups which are consistent with the estimation problems in the literature. Compressive sensing systems, however, might require a wider set of assumptions, and thus, a large class of settings could not be addressed by earlier investigations. As the result, a body of work deviated from this approach and applied the replica method directly to the compressive sensing problem; see for example \cite{tulino2013support,vehkapera2014analysis,kabashima2010statistical,kabashima2009typical,wen2016sparse,zheng2017does}. 

Although in general the replica method is considered to be mathematically non-rigorous, several recent studies have justified the validity of the replica results in the context of compressive sensing by using some alternative tools for analysis. A widely investigated approach is based on the asymptotic analysis of \ac{amp} algorithms. In the context of compressive sensing, the \ac{amp} algorithms were initially introduced to address iteratively the convex reconstruction schemes based on $\ell_1$ norm minimization, such as LASSO and basis pursuit, with low computational complexity \cite{donoho2009message,donoho2010message}. The proposed approach was later employed to extend the algorithm to a large variety of estimation problems including \ac{map} and \ac{mmse} estimation; see for example \cite{rangan2011generalized,montanari2012graphical}. The primary numerical investigations of \ac{amp} demonstrated that for large sensing systems the sparsity-compression rate tradeoff of these iterative algorithms, as well as the compression rate-distortion tradeoff in noisy cases, is derived by the fixed-points of ``state evolution'' and recovers the asymptotics of convex reconstruction schemes \cite{donoho2010message}. This observation was then rigorously justified for \ac{iid} sub-Gaussian sensing matrices in \cite{bayati2011dynamics} by using the conditioning technique developed in \cite{bolthausen2009high}. The study was recently extended to cases with rotationally invariant system matrices in \cite{rangan2017vector,takeuchi2017rigorous}. The investigations in \cite{donoho2013information,krzakala2012statistical} moreover showed that using \ac{amp} algorithms for spatially coupled measurements, the fundamental limits on the required compression rate \cite{wu2012optimal,wu2011mmse} can be achieved in the asymptotic regime. The methodology proposed by \ac{amp} algorithms and their state evolution also provided a justification for validity of several earlier studies based on the replica method. In fact, the results given by the replica method were recovered through state evolution of the corresponding \ac{amp} algorithms. Invoking this approach along with other analytical tools, the recent study in \cite{barbier2017mutual} further approved the validity of the replica prediction for the asymptotic \ac{mmse} and mutual information of the linear system in \eqref{eq:sys-1} with \ac{iid} Gaussian measurements. Similar results were demonstrated in \cite{reeves2016replica} using a different approach. 

\subsection{Contribution and Outline}
\label{subsec:contribution}
In this paper, we determine the asymptotic distortion for a general distortion function for cases where the \ac{map} estimator is employed to estimate the source vector from the observation given in \eqref{eq:sys-1}. We represent the  asymptotic distortion in \eqref{eq:int-2.1} as a macroscopic parameter of a corresponding spin glass and study the spin glass via the replica method. The general replica ansatz is then given for an arbitrary replica correlation matrix, and its special cases are studied considering the \ac{rs} and \ac{rsb} assumptions. The asymptotic distortion is determined for rotationally invariant random system matrices invoking results for asymptotics of spherical integrals \cite{guionnet2002large,guionnet2005fourier}. Using our asymptotic results, we derive a more general form of the decoupling principle by restricting the distortion function to be of a special form and employing the moment method \cite{akhiezer1965classical,simon1998classical}. We show that the vector-valued system in \eqref{eq:sys-1} estimated by \eqref{eq:int-2} decouples into a bank of similar noisy single-user systems followed by single-user \ac{map} estimators. This result holds for any replica correlation matrix; however, the structure of the decoupled single-user system depends on the supposed structure of the correlation matrix. Under the \ac{rsb} assumption with $b$ steps of breaking ($b$\ac{rsb}), the noisy single-user system is given in the form of an input term added by an impairment term. The impairment term, moreover, is expressed as the sum of an independent Gaussian random variable and $b$ correlated interference terms. By reducing the assumption to \ac{rs}, the result reduces to the formerly studied \ac{rs} decoupling principle of the \ac{map} estimators \cite{rangan2012asymptotic} for rotationally invariant system matrix ensembles. In fact, our investigations collect the previous results regarding the decoupling principle in addition to a new set of setups under a single umbrella given by a more general form of the decoupling principle. More precisely, we extend the scope of the decoupling principle to
\begin{itemize}
\item the systems whose measuring matrix belongs to the class of rotationally invariant random matrices, and
\item the replica ansatz with general replica correlations which include the \ac{rs} and \ac{rsb} ans\"atze.
\end{itemize}

To address a particular application, we study the large-system performance of a compressive sensing system under the $\ell_p$ minimization recovery schemes. We address the linear reconstruction, as well as the LASSO and the $\ell_0$ norm scheme considering both the sparse Gaussian and finite alphabet sources. Our general setting allows to investigate the asymptotic performance with respect to different metrics and for multiple sensing matrices such as random \ac{iid} and projector. The numerical investigations show that the \ac{rs} ansatz becomes unstable for some regimes of system parameters and predicts the performance of $\ell_0$ minimization recovery loosely within a large range of compression rates. This observation agrees with the earlier discussions on the necessity of \ac{rsb} investigations reported in \cite{kabashima2009typical,takeda2011statistical}. We therefore study the performance under \ac{rsb} and discuss the impact of the symmetry breaking. Throughout the numerical investigations, it is demonstrated that the performance enhancement obtained via random orthogonal measurements, reported in \cite{vehkapera2014analysis}, also holds for sparse finite alphabet sources in which sensing via random projector matrices results in phase transitions at higher rates.

The rest of the manuscript is organized as follows. In Section \ref{sec:problem_formulation}, the problem is formulated. We illustrate our statistical mechanical approach in Section \ref{sec:statisc_app} and explain briefly the replica method. The general replica ansatz, as well as the general decoupling principle is given in Section \ref{sec:results}. The ansatz under the \ac{rs} and \ac{rsb} assumptions are expressed in Sections \ref{sec:rs} and \ref{sec:rsb}. Based on $b$\ac{rsb} decoupled system, we propose the idea of a replica simulator in \ref{sec:rep_sim} and describe the given ans\"atze in terms of the corresponding decoupled systems. To address an application of our study, we consider large compressive sensing systems in Section \ref{sec:cs} and discuss several examples. The numerical investigations of the examples are then given in Section \ref{sec:numerics}. Finally, we conclude the manuscript in Section \ref{sec:conclusion}.
\subsection{Notations}
Throughout the manuscript, we represent scalars, vectors and matrices with lower, bold lower, and bold upper case letters, respectively. The set of real numbers is denoted by $\setR$, and $\mA^{\trp}$ and $\mA^{\mathsf{H}}$ indicate the transposed and Hermitian of $\mA$. $\mI_m$ is the $m\times m$ identity matrix and $\mone_m$ is an $m \times m$ matrix with all entries equal to $1$. For a random variable $x$, $\mathrm{p}_x$ represents either the \ac{pmf} or \ac{pdf}, and $\mathrm{F}_x$ represents the \ac{cdf}. We denote the expectation over $x$ by $\mathsf{E}_x$, and an expectation over all random variables involved in a given expression by $\mathsf{E}$. $\setZ$ and $\setR$ represent the set of integer and real numbers and the superscript $+$, e.g. $\setR^+$, indicates the corresponding subset of all non-negative numbers. For sake of compactness, the set of integers $\{1, \ldots,n \}$ is abbreviated as $[1:n]$, the zero-mean and unit-variance Gaussian \ac{pdf} is denoted by $\phi(\cdot)$, and the Gaussian averaging is expressed as
\begin{align}
\int \left( \cdot \right) \md t \coloneqq  \int_{-\infty}^{+\infty} \left(\cdot\right) \phi(t) \dif t .
\end{align}
Moreover, in many cases, we drop the set on which a sum, minimization or an integral is taken. Whenever needed, we consider the entries of $\bx$ to be discrete random variables, namely the support $\setX$ to be discrete. The results of this paper, however, are in full generality and hold also for continuous distributions.
\subsection{Announcements}
Some of the results of this manuscript were presented at the IEEE Information Theory Workshop \cite{bereyhi2016itw} and the Information Theory and Applications Workshop \cite{bereyhi2017replica}. Even though the results have a mathematical flavor, the stress is not on investigating the rigor of the available tools such as the replica method, but rather to employ them for deriving formulas which can be used in different problems. %in communications. 

\section{Problem Formulation}
\label{sec:problem_formulation}
Consider the vector-valued linear system described by \eqref{eq:sys-1}. Let the system satisfy the following properties.
%\begin{align}
%\by=\mA\bx+\bz. \label{eq:sys-1}
%\end{align}

\begin{enumerate}[label=(\alph*)]
\item $\bx_{n \times 1}$ is an \ac{iid} random vector with each entry being distributed with $\rmp_x$ over $\setX \subseteq \setR$.
\item $\mA_{k \times n}$ is randomly generated over $\setA^{k \times n}$ with $\setA \subseteq \setR$ from rotationally invariant random ensembles. The random matrix $\mA$ is said to be rotationally invariant when its Gramian, i.e., $\mJ=\mA^{\trp} \mA$, has the eigendecomposition 
\begin{align}
\mJ= \mU \mD \mU^{\trp} \label{eq:sys-1.2}
\end{align}
with $\mU_{n \times n}$ being an orthogonal Haar distributed matrix and $\mD_{n \times n}$ being a diagonal matrix. For a given $n$, we denote the empirical \ac{cdf} of $\mJ$'s eigenvalues (cumulative density of states) with $\mathrm{F}_{\mJ}^n$ and define it as
\begin{align}
\mathrm{F}_{\mJ}^n (\lambda) =\frac{1}{n} \sum_{i=1}^n \mone \{ \lambda^{\mJ}_i < \lambda \}.
\end{align}
where $\lambda^{\mJ}_i$ for $i \in [1:n]$ denotes the $i$th diagonal entry of $\mD$. We assume that $\mathrm{F}_{\mJ}^n$ converges, as $n\uparrow\infty$, to a deterministic \ac{cdf} $\mathrm{F}_{\mJ}$.
\item $\bz_{k \times 1}$ is a real \ac{iid} zero-mean Gaussian random vector in which the variance of each entry is $\lambda_0$.%, i.e., $\bz \sim \man(\boldsymbol{0},\lambda_0 \mI)$.
\item The number of observations $k$ is a deterministic function of the transmission length $n$, such that
\begin{align}
\lim_{n \uparrow \infty} \frac{k(n)}{n}=\frac{1}{\sfr} < \infty . \label{eq:eq:sys-1.3}
\end{align}For sake of compactness, we drop the explicit dependence of $k$ on $n$.
\item $\bx$, $\mA$ and $\bz$ are independent.
\end{enumerate}
The source vector $\bx$ is reconstructed from the observation vector $\by$ with a system matrix $\mA$ that is known at the estimator. Thus, for a given $\mA$, the source vector is recovered by $\bhx=\bgg(\by|\mA)$ where $\bgg(\cdot|\mA)$ is given in \eqref{eq:int-2}. %
%\begin{align}
%\bhx=\bgg(\by|\mA)= \arg \min_{\bv} \ \left[ \frac{1}{2\lambda} \norm{\by-\mA \bv}^2 + u(\bv) \right]  \label{eq:sys-2}
%\end{align}
Here, the non-negative scalar $\lambda$ is the estimation parameter and the non-negative cost function $u(\cdot)$ is referred to the ``utility function''. The utility function $u(\cdot)$ is supposed to decouple which means that it takes arguments with different lengths, i.e., $u(\cdot): \setR^\ell \mapsto \setR^+$ for any positive integer $\ell$, and
\begin{align}
u(\bx)=\sum_{i=1}^n u(x_i).
\end{align}
In order to use the estimator in \eqref{eq:int-2}, one needs to guarantee the uniqueness of the estimation output. Therefore, we impose the following constraint on our problem.
\begin{enumerate}[label=(\alph*)]\addtocounter{enumi}{5}
\item For a given observation vector $\by$, the objective function in \eqref{eq:int-2} has a unique minimizer over the support $\setX^n$.
\end{enumerate}

\subsection{\ac{map} Estimator}
The \ac{map} estimator in \eqref{eq:int-2} can be considered as the optimal estimator in the sense that it minimizes the reconstruction's error probability postulating a source prior distribution proportional to $e^{-u(\bx)}$ and a noise variance~$\lambda$. To clarify the argument, assume $\setX$ is a finite set\footnote{i.e., the entries of $\bx$ are discrete random variables.}. In this case, we can define the reconstruction's error probability~as
\begin{align}
\mathsf{P_e} = \Pr \{ \bx \neq \tilde{\bgg}(\by| \mA) \} \label{eq:sys-3}
\end{align}for some estimator $\tilde{\bgg}(\cdot| \mA)$. In order to minimize $\mathsf{P_e}$, $\tilde{\bgg}(\cdot| \mA)$ is chosen such that the posterior distribution over the input support $\setX^n$ is maximized, i.e.,
\begin{subequations}
\begin{align}
\tilde{\bgg}(\by| \mA) &= \arg \max_{\bv} \rmp_{\bx|\by,\mA}(\bv|\by,\mA) \label{eq:sys-4a} \\
&= \arg \max_{\bv} \rmp_{\by|\bx,\mA}(\by|\bv,\mA) \rmp_{\bx|\mA}(\bv|\mA) \label{eq:sys-4b} \\
&\stackrel{\star}{=} \arg \min_{\bv} \left[ - \log \rmp_{\by|\bx,\mA}(\by|\bv,\mA) - \log \rmp_{\bx}(\bv) \right]. \label{eq:sys-4c}
\end{align}
\end{subequations}
where $\star$ comes from the independency of $\bx$ and $\mA$. Here, $\rmp_{\by|\bx,\mA}({\by|\bv,\mA})=\rmp_{\bz}({\by-\mA\bv})$, and $\rmp_{\bx}$ is the prior distribution of the source vector. Now, let the estimator postulate the noise variance to be $\lambda$ and the prior to be
\begin{align}
\rmp_{\bx}(\bv)=\frac{e^{-u(\bv)}}{\sum_{\bv}e^{-u(\bv)}} \label{eq:sys-5}
\end{align}
for some non-negative function $u(\cdot)$. Substituting into \eqref{eq:sys-4c}, the estimator $\tilde{\bgg}(\cdot| \mA)$ reduces to ${\bgg}(\cdot| \mA)$ defined in \eqref{eq:int-2}. The estimator in \eqref{eq:int-2} models several particular reconstruction schemes in compressive sensing. We address some of these schemes later on in Section~\ref{sec:cs}.

\begin{remark}
\normalfont
In the case that the entries of $\bx$ are continuous random variables, the above argument needs some modifications. In fact, in this case the reconstruction's error probability as defined in \eqref{eq:sys-3} is always one, and therefore, it cannot be taken as the measure of error. In this case, the \ac{map} estimator is considered to maximize the posterior \ac{pdf} postulating $\rmp_{\bx}(\bv)\varpropto e^{-u(\bv)}$ and noise variance $\lambda$.
\end{remark}

\subsection{Asymptotic Distortion and Conditional Distribution}
\label{subsec:distortion_measure}
In many applications, the distortion is given in terms of the average \ac{mse}, while in some others the average symbol error rate is considered. In fact, the former takes the $\ell_2$ norm as the distortion function, and the latter considers the $\ell_0$ norm. The distortion function, however, can be of some other form in general. Here, we study the asymptotic performance by considering a general distortion function which determines the imperfection level of the estimation. Thus, we consider a distortion function $\sfd(\cdot;\cdot)$ which
\begin{align}
\sfd(\cdot;\cdot): \setX \times \setX \to \setR. \label{eq:sys-6}
\end{align}

The term ``average distortion'' usually refers to the case when the averaging weights are uniform. It means that each tuple of source-estimated entries $(x_i,\hx_i)$ is weighted equally when the distortion is averaged over all the entries' tuples. It is however possible to average the distortion by a non-uniform set of weights. In the following, we define the average distortion for a class of binary weights which includes the case of uniform averaging,~as~well.
%The distortion, however, can be averaged non-uniformly for some arbitrary weight vector. We define the average distortion in a wider sense which also includes the uniform average distortion.

\begin{definition}[Asymptotic distortion]
\label{def:average_dist}
Consider the vectors $\bx_{n \times 1}$ and $\bhx_{n \times 1}$ defined over $\setX^n$, and let $\sfd(\cdot;\cdot)$ be a distortion function defined in \eqref{eq:sys-6}. Define the index set $\setW(n)$ as a subset of $[1:n]$, and let $\abs{\setW(n)}$ grow with $n$ such that
\begin{align}
\lim_{n \uparrow \infty} \frac{\abs{\setW(n)}}{n}= \eta \label{eq:sys-7}
\end{align}
for some $\eta \in [0,1]$. Then, the average distortion of $\bhx$ and $\bx$ over the index set $\setW(n)$ regarding the distortion function $\sfd(\cdot;\cdot)$ is defined as
\begin{align}
\sfD^{\setW(n)}(\bhx;\bx) \coloneqq \frac{1}{\abs{\setW(n)}} \sum_{w \in \setW(n)} \sfd(\hx_w;x_w). \label{eq:sys-8}
\end{align}
Assuming the limit of $\eqref{eq:sys-8}$ when $n \uparrow \infty$ exists, we denote
\begin{align}
\sfD^{\setW}(\bhx;\bx) \coloneqq \lim_{n\uparrow\infty} \sfD^{\setW(n)}(\bhx;\bx) \label{eq:sys-9}
\end{align}
and refer to it as the ``asymptotic distortion'' over the limit of the index set $\setW(n)$.
\end{definition}

Definition \ref{def:average_dist} is moreover utilized to investigate the asymptotic conditional distribution of the estimator which plays a key role in studying the decoupling principle. For further convenience, we define the asymptotic conditional distribution of the \ac{map} estimator as follows.
\begin{definition}[Asymptotic conditional distribution]
\label{def:asymp_distribution}
Consider the source vector $\bx_{n \times 1}$ passed through the linear system in \eqref{eq:sys-1}, and let $\bhx_{n \times 1}$ be its \ac{map} estimation as given in \eqref{eq:int-2}. For a given $n$, we take an index $j \in [1:n]$ and denote the conditional distribution of $\hx_j$ given $x_j$ by $\rmp_{\hx|x}^{j(n)}$ which at the mass point $(\hv,v) \in \setX\times\setX$ reads
\begin{align}
\rmp_{\hx|x}^{j(n)}(\hv|v) \coloneqq  \left[\rmp_{x_j}(v)\right]^{-1} \rmp_{\hx_j,x_j}(\hv,v) \label{eq:sys-10}
\end{align}
with $\rmp_{\hx_j,x_j}(\hv,v)$ being the marginal joint distribution of $x_j$ and $\hx_j$ at the mass point $(\hv,v)$. Then, in the large-system limit, we define the asymptotic conditional distribution of $\hx_j$ given $x_j$ at $(\hv,v)$ as 
\begin{align}
\rmp_{\hx|x}^{j} (\hv|v) \coloneqq \lim_{n \uparrow \infty} \rmp_{\hx|x}^{j(n)} (\hv|v). \label{eq:sys-11}
\end{align}
\end{definition}

We study the asymptotic distortion over the limit of a desired index set $\setW(n)$ and distortion function $\sfd(\cdot;\cdot)$ by defining it as a macroscopic parameter of the corresponding spin glass and employing the replica method to evaluate it. Using the result for the asymptotic distortion, we determine then the asymptotic conditional distribution and investigate the decoupling property of the estimator.

\section{Statistical Mechanical Approach}
\label{sec:statisc_app}
The Hamiltonian in \eqref{eq:int-7} introduces a spin glass which corresponds to the \ac{map} estimator. The spin glass at zero temperature describes the asymptotics of the \ac{map} estimator. For further convenience, we formally define the ``corresponding spin glass'' as follows.

\begin{definition}[Corresponding spin glass]
\label{def:cors_spin}
\normalfont
Consider the integer $n\in \setZ^+$. The corresponding spin glass with the microstate $\bv_{n \times 1} \in \setX^n$ and the quenched randomizers $\by_{n \times 1}$ and $\mA_{k\times n}$ is given by
\begin{itemize}
\item $\mA$ is a rotationally invariant random matrix, and $\by$ is constructed as in \eqref{eq:sys-1} from the source vector $\bx$. %satisfying the constraints in Section \ref{sec:problem_formulation}
\item For any realization of $\mA$ and $\by$, the Hamiltonian reads
\begin{align}
\mae(\bv|\by,\mA)=\frac{1}{2\lambda} \norm{\by-\mA \bv}^2 + u(\bv). \label{eq:sm-0.1}
\end{align}
for $\bv \in \setX^n$.
\end{itemize}
\end{definition}
For the corresponding spin glass, at the inverse temperature $\upbeta$, the following properties are~directly~concluded.
\begin{itemize}
\item The conditional distribution of the microstate reads
\begin{align}
\rmp^{\upbeta}_{\bv|\by,\mA}(\bv|\by,\mA)=\frac{e^{-\upbeta\mae(\bv|\by,\mA) }}{\maz(\upbeta|\by,\mA)} \label{eq:sm-0.2}
\end{align}
with $\maz(\upbeta|\by,\mA)$ being the partition function
\begin{align}
\maz(\upbeta|\by,\mA) = \sum_{\bv} e^{-\upbeta\mae(\bv|\by,\mA)}. \label{eq:sm-0.3}
\end{align}
\item The normalized free energy is given by
\begin{align}
\sfF(\upbeta)=-\frac{1}{\upbeta n} \sfE \log \maz(\upbeta|\by,\mA), \label{eq:sm-0.4}
\end{align}
where the expectation is taken over the quenched randomizers.
\item The entropy of the spin glass is determined as
\begin{align}
\rmH(\upbeta) &= \upbeta^2 \frac{\dif}{\dif \upbeta} \left[ \sfF(\upbeta) \right]. \label{eq:sm-0.5}
\end{align}
\end{itemize}

Regarding the \ac{map} estimator, one can represent the asymptotic distortion as a macroscopic parameter of the corresponding spin glass. More precisely, using Definition \ref{def:asymp_distribution}, the asymptotic distortion reads
\begin{align}
\sfD^{\setW}(\bhx;\bx)=\lim_{\upbeta\uparrow\infty}\lim_{n\uparrow\infty}\sfE_{\bv}^\upbeta\sfD^{\setW(n)}(\bv;\bx)\label{eq:sm-1}
\end{align}
where $\sfE_{\bv}^\upbeta$ indicates the expectation over $\bv \in \setX^n$ with respect to the conditional Boltzmann-Gibbs distribution $\rmp^{\upbeta}_{\bv|\by,\mA}$ defined in \eqref{eq:sm-0.2}. In fact, by introducing the macroscopic parameter $\sfD^{\setW}(\upbeta)$\footnote{In general, $\sfD^{\setW}(\upbeta)$ is a function of $\upbeta$, $\by$, $\bx$ and $\mA$. However, we drop the other arguments for sake of compactness.} at the temperature $\upbeta^{-1}$ as
\begin{align}
\sfD^{\setW}(\upbeta)\coloneqq\lim_{n\uparrow\infty}\sfE_{\bv}^\upbeta\sfD^{\setW(n)}(\bv;\bx), \label{eq:sm-2}
\end{align}
the asymptotic distortion can be interpreted as the macroscopic parameter at zero temperature. Here, we take a well-known strategy in statistical mechanics which modifies the partition function to
\begin{align}
\maz(\upbeta,h) = \sum_{\bv}  e^{-\upbeta\mae(\bv|\by,\mA)+h n\sfD^{\setW(n)} (\bv;\bx)}. \label{eq:sm-3}
\end{align}
In this case, the expectation in \eqref{eq:sm-2} is taken as
\begin{align}
\sfD^{\setW}(\upbeta)=\lim_{n \uparrow \infty} \lim_{h \downarrow 0} \frac{1}{n}\frac{\partial}{\partial h} \log \maz(\upbeta,h). \label{eq:sm-4}
\end{align}
The macroscopic parameter defined in \eqref{eq:sm-2} is random, namely depending on the quenched randomizer $\{ \by , \mA \}$. As discussed in Section \ref{sec:spin_glasses}, under the self averaging property, the macroscopic parameter is supposed to converge in the large-system limit to its expected value over the quenched random variables. For the corresponding spin glass defined in Definition \ref{def:cors_spin}, the self averaging property has not been rigorously justified, and the proof requires further mathematical investigations as in \cite{guerra2002infinite}. However, as it is widely accepted in the literature, we assume that the property holds at least for the setting specified here. Therefore, we state the following assumption.

\begin{assumption}[Self Averaging Property]
\label{asp:1}
\normalfont
Consider the corresponding spin glass defined in Definition \ref{def:cors_spin}. For almost all realizations of the quenched randomizers $\mA$ and $\by$,
\begin{align}
\sfD^{\setW}(\upbeta) = \sfE_{\by,\mA} \sfD^{\setW}(\upbeta). \label{eq:sm-5}
\end{align}
%with the macroscopic parameter $\sfD^{\setW}(\upbeta)$ defined in \eqref{eq:sm-2}\footnote{In fact, one could directly impose this assumption on the problem that the system setup and distortion function are chosen such that the asymptotic distortion in \eqref{eq:sys-9} exists and is almost surely constant for the given statistics of the system.}.
\end{assumption}
Using the self averaging property of the system, the asymptotic distortion is written as
\begin{align}
\sfD^{\setW}(\bhx;\bx)=\lim_{\upbeta\uparrow\infty}\lim_{n\uparrow\infty}\lim_{h \downarrow 0} \frac{1}{n} \frac{\partial}{\partial h} \sfE \log \maz(\upbeta,h). \label{eq:sm-6}
\end{align}
Evaluation of \eqref{eq:sm-6}, as well as the normalized free energy defined in \eqref{eq:sm-0.4}, confronts the nontrivial problem of determining the logarithmic expectation. The task can be bypassed by using the Riesz equality \cite{riesz1930valeurs} which for a given random variable $t$ states that
\begin{align}
\E \log t = \lim_{m \downarrow 0} \frac{1}{m}\log \E t^m. \label{eq:sm-7}
\end{align}
Using the Riesz equality, the asymptotic distortion can be finally written as 
\begin{align}
\sfD^{\setW}(\bhx;\bx)=\lim_{\upbeta\uparrow\infty}\lim_{n\uparrow\infty}\lim_{h \downarrow 0} \lim_{m \downarrow 0} \frac{1}{m} \frac{1}{n} \frac{\partial}{\partial h} \log \sfE [\maz(\upbeta,h)]^m. \label{eq:sm-8}
\end{align}
\eqref{eq:sm-8} expresses the asymptotic distortion in terms of the moments of the modified partition function; however, it does not yet simplify the problem. In fact, one faces two main difficulties when calculating the right hand side of \eqref{eq:sm-8}:
\begin{inparaenum}
\item the moment needs to be evaluated for real values of $m$ (at least within a right neighborhood of $0$), and
\item the limits need to be taken in the order stated.
\end{inparaenum}
Here is where the replica method plays its role. The replica method suggests to determine the moment for an arbitrary non-negative integer $m$ as an analytic function in $m$ and then assume the two following statements:
\begin{enumerate}
\item The moment function analytically continues from the set of integer numbers onto the real axis (at least for some $m$ at a right neighborhood of $0$) which means that an analytic expression found for integer moments directly extends to all (or some) real moments. Under this assumption, the expression determined for integer moments can be replaced in \eqref{eq:sm-8}, and the limit with respect to $m$ taken when $m \downarrow 0$. This assumption is the main part where the replica method lacks rigor and is known as the ``replica continuity''.
\item In \eqref{eq:sm-8}, the limits with respect to $m$ and $n$ exchange. We refer to this assumption as the ``limit exchange''.
\end{enumerate}

In order to employ the replica method, we need to suppose the validity of the above two statements; therefore, we state the following assumption.
\begin{assumption}[Replica Continuity and Limit Exchange]
\label{asp:2}
\normalfont
For the spin glass defined in Definition \ref{def:cors_spin}, the replica continuity and limit exchange assumptions hold.
\end{assumption}

By means of Assumption \ref{asp:2}, the calculation of asymptotic distortion reduces to the evaluation of integer moments of the modified partition function which is written as
\begin{subequations}
\begin{align}
\sfZ(m) &\coloneqq \sfE [\maz(\upbeta,h)]^m \label{eq:sm-9a} \\
&=\sfE \prod_{a=1}^m \ \sum_{\bv_a}  e^{-\upbeta\mae(\bv_a|\by,\mA)+h n\sfD^{\setW(n)} (\bv_a;\bx)} \label{eq:sm-9b}\\
&=\sfE_{\bx} \sfE_{\mA} \sfE_{\bz} \sum_{\{\bv_a\}}   e^{-\upbeta\sum_{a=1}^m\mae(\bv_a|\by,\mA)+h n\sum_{a=1}^m\sfD^{\setW(n)} (\bv_a;\bx)} \label{eq:sm-9c}.
\end{align}
\end{subequations}
Here, we refer to $\bv_a\in\setX^n$ for $a \in [1:m]$ as the replicas, and define $\{\bv_a\} \coloneqq \{\bv_1, \ldots, \bv_m\} \in \setX^n \times \cdots \times \setX^n$ as the set of replicas. After taking the expectation with respect to $\bz$ and $\mA$, it is further observed that, in the large limit, the expectation with respect to $\bx$ can be dropped due to the law of large numbers. By inserting the final expression for $\sfZ(m)$ in \eqref{eq:sm-8} and taking the limits, the asymptotic distortion is determined as~in~Proposition~\ref{proposition:1}~given~below.

\section{Main Results}
\label{sec:results}
Proposition \ref{proposition:1} states the general replica ansatz. The term ``general'' is emphasized here, in order to indicate that no further assumption needs to be considered for derivation. Using Proposition \ref{proposition:1} along with results in the classical moment problem \cite{simon1998classical}, a general form of the decoupling principle is justified for the \ac{map} estimator.

Before stating the general replica ansatz, let us define the $\rmR$-transform of a probability distribution. Considering a random variable $t$, the corresponding Stieltjes transform over the upper complex half plane is defined as 
\begin{align}
\mathrm{G}_t(s)= \E (t-s)^{-1}. \label{eq:rep-1}
\end{align}
Denoting the inverse with respect to composition by $\mathrm{G}_t^{-1}(\cdot)$, the $\rmR$-transform is given by
\begin{align}
\mathrm{R}_t(\omega)= \mathrm{G}_t^{-1}(-\omega) - \omega^{-1} \label{eq:rep-2}
\end{align}
such that $\lim_{\omega \downarrow 0} \mathrm{R}_{t}(\omega) = \E t$. 

The definition can also be extended to matrix arguments. Assume the matrix $\mM_{n \times n}$ to have the decomposition $\mM=\mU \mathbf{\Lambda} \mU^{-1}$ where $\mathbf{\Lambda}_{n \times n}$ is a diagonal matrix whose nonzero entries represent the eigenvalues of $\mM$, i.e., $\mathbf{\Lambda}=\mathrm{diag}[\lambda_1, \ldots, \lambda_n]$, and $\mU_{n \times n}$ is the matrix of eigenvectors. $\rmR_t(\mM)$ is then an $n \times n$ matrix defined as
\begin{align}
\mathrm{R}_t(\mM)=\mU \ \mathrm{diag}[\mathrm{R}_t(\lambda_1), \ldots, \mathrm{R}_t(\lambda_n)] \ \mU^{-1}. \label{eq:rep-3}
\end{align}

\subsection{General Replica Ansatz}
Proposition \ref{proposition:1} expresses the macroscopic parameters of the corresponding spin glass, including the asymptotic distortion, in terms of the parameters of a new spin glass of finite dimension. It is important to note that the new spin glass, referred to as ``spin glass of replicas'', is different from the corresponding spin glass defined in Definition \ref{def:cors_spin}. In fact, the spin glass of replicas is the projection of the corresponding spin glass on the reduced support $\setX^m$ with $m$ indicating the number of replicas. The macroscopic parameters of the spin glass of replicas can therefore readily be determined.
\begin{definition}[Spin glass of replicas]
\label{def:replica_spin}
\normalfont
For the finite integer $m$, the spin glass of replicas with the microstate $\bvv_{m \times 1} \in \setX^m$ and quenched randomizer $\bxx_{m \times 1}$ is defined as follows.
\begin{itemize}
\item All the entries of $\bxx$ equal to $x$ where $x$ is a random variable distributed with the source distribution $\rmp_x$.
\item For a given realization of $\bxx$, the Hamiltonian reads
\begin{align}
\mae^\sfR(\bvv|\bxx)= (\bxx-\bvv)^{\trp} \mT \rmR_{\mJ}(-2 \upbeta \mT \mQ) (\bxx-\bvv) + u(\bvv) \label{eq:rep-4}
\end{align}
where $\rmR_{\mJ}(\cdot)$ is the $\rmR$-transform corresponding to $\mathrm{F}_{\mJ}$, $\mT_{m \times m}$ is defined as
\begin{align}
\mT \coloneqq \frac{1}{2 \lambda}\left[ \mI_m - \frac{\upbeta\lambda_0}{\lambda+m \upbeta \lambda_0} \mone_m \right], \label{eq:rep-5}
\end{align}
and $\mQ_{m\times m}$ is the so-called replica correlation matrix defined as
\begin{align}
\mQ= \E (\bxx - \bvv)(\bxx - \bvv)^{\trp}. \label{eq:rep-6}
\end{align}
with the expectation taken over $\bxx$ and $\bvv$ at thermal equilibrium.
\item At thermal equilibrium, the microstate is distributed according to the Boltzmann-Gibbs distribution $\rmp^{\upbeta}_{\bvv|\bxx}$
\begin{align}
\rmp^{\upbeta}_{\bvv|\bxx}(\bvv|\bxx)=\frac{e^{-\upbeta\mae^\sfR(\bvv|\bxx) }}{\maz^\sfR(\upbeta|\bxx)} \label{eq:rep-7}
\end{align}
where $\maz^\sfR(\upbeta|\bxx)$ denotes the partition function of the system defined as
\begin{align}
\maz^\sfR(\upbeta|\bxx) \coloneqq \sum_{\bvv} e^{-\upbeta\mae^\sfR(\bvv|\bxx)}. \label{eq:rep-8}
\end{align}
\item The normalized free energy of the system at the inverse temperature $\upbeta$ is given by
\begin{align}
\sfF^\sfR(\upbeta,m)=-\frac{1}{\upbeta m} \sfE \log \maz^\sfR(\upbeta|\bxx), \label{eq:rep-9}
\end{align}
where the expectation is taken over $\bxx$ with respect to $\rmp_x$. The average energy and entropy at the inverse temperature $\upbeta$ are further found using \eqref{eq:int-5.1a} and \eqref{eq:int-5.1b}.
\item For the system at thermal equilibrium, the replicas' average distortion regarding the distortion function $\sfd(\cdot,\cdot)$ at the inverse temperature $\upbeta$ is
\begin{align}
\sfD^\sfR(\upbeta,m)=\frac{1}{m} \sfE \sum_{a=1}^m \sfd(\vv_a,x), \label{eq:rep-10}
\end{align}
with expectation being taken over $\bxx$ and $\bvv$ with respect to $\rmp_\bxx\rmp_{\bvv|\bxx}$.
\end{itemize}
\end{definition}

Considering Definition \ref{def:replica_spin}, the evaluation of the system parameters such as the replicas' average distortion $\sfD^\sfR(\upbeta,m)$ or the normalized free energy $\sfF^\sfR(\upbeta,m)$ needs the replica correlation matrix $\mQ$ to be explicitly calculated first. In fact, \eqref{eq:rep-6} describes a fixed point equation in terms of $\mQ$ when one writes out the expectation using the conditional distribution in \eqref{eq:rep-7}. The solution can then be substituted in the distribution and the parameters of the system can be calculated via \eqref{eq:rep-8}-\eqref{eq:rep-10}. The fixed point equation, however, may have several solutions and thus result in multiple outputs for the system. Nevertheless, we express the asymptotic distortion of the \ac{map} estimator in terms of a single output of the spin glass of replicas for which the limits exist and the free energy is minimized.

\begin{proposition}
\label{proposition:1}
Let the linear system \eqref{eq:sys-1} fulfill the constraints of Section \ref{sec:problem_formulation}. Suppose Assumptions \ref{asp:1} and \ref{asp:2} hold, and consider the spin glass of replicas as defined in Definition \ref{def:replica_spin} for a finite integer $m$. Then, the free energy of the corresponding spin glass defined in Definition \ref{def:cors_spin} is given by
\begin{align}
\sfF(\upbeta) = \lim_{m\downarrow 0} \frac{1}{m} \left[ \int_0^1 \tr{\mT \mQ \rmR_{\mJ}(-2\upbeta\omega\mT\mQ)} \dif \omega - \tr{\mQ^\trp \mT \rmR_{\mJ}(-2\upbeta \mT \mQ)} \right] +\sfF^\sfR (\upbeta,m) \label{eq:rep-11}
\end{align}
where $\mQ$ is the replica correlation matrix satisfying \eqref{eq:rep-6}. The asymptotic distortion of the \ac{map} estimator regarding the distortion function $\sfd(\cdot;\cdot)$ is then determined as
\begin{align}
\sfD^{\setW}(\bhx;\bx)= \lim_{\upbeta\uparrow\infty}\lim_{m\downarrow 0} \sfD^\sfR(\upbeta,m). \label{eq:rep-12}
\end{align}
In case that multiple solutions are available for the replica correlation matrix, the replica's average distortion in \eqref{eq:rep-12} is evaluated via a correlation matrix which minimizes the free energy at zero temperature, i.e., $\upbeta \uparrow \infty$.
\end{proposition}

\begin{proof}
The proof is given in Appendix \ref{app:a}. However, we explain briefly the strategy in the following.\\
Starting from \eqref{eq:sm-9c}, the expectation with respect to the noise term is straightforwardly taken. Using the results from \cite{guionnet2005fourier}, the expectation with respect to the system matrix is further taken as discussed in Appendix \ref{app:d}. Then, by considering the following variable exchange,
\begin{align}
[\mQ]_{ab} = \frac{1}{n} (\bx-\bv_a)^\trp (\bx-\bv_b). \label{eq:rep-13}
\end{align}
$\sfZ(m)$ is determined in terms of the replica correlation matrix $\mQ$. Finally, by employing the law of large numbers, the $m$th moment of the partition function is given as
\begin{align}
\sfZ(m)=\sfE_{\bx} \int e^{-n \left[ \mg(\mT \mQ) - \mai(\mQ)\right] +\epsilon_n} \dif \mQ \label{eq:rep-14}
\end{align}
where $\dif \mQ \coloneqq \prod_{a,b=1}^{m} \dif [\mQ]_{ab}$ with the integral being taken over $\setR^{m\times m}$, $\epsilon_n$ tends to zero as $n \uparrow \infty$ and $\mT$ is given by \eqref{eq:rep-5}. Moreover, $e^{n \mai(\mQ)}$ denotes the non-normalized probability weight of the vectors of replicas with a same correlation matrix and is explicitly determined in \eqref{eq:a-20b} in Appendix \ref{app:a}, and $\mg(\cdot)$ reads
\begin{align}
\mg(\mM) =  \int_{0}^{{\upbeta}} \Tr \{\mM \mathrm{R}_{\mJ}(-2\omega\mM)\} \dif \omega \label{eq:rep-16}
\end{align}
for some diagonal matrix $\mM$ with $\tr{\cdot}$ denoting the trace operator, and $\mathrm{R}_{\mJ}(\cdot)$ being the $\mathrm{R}$-transform with respect to $\mathrm{F}_{\mJ}$. In \eqref{eq:rep-14}, the term $e^{n \mai(\mQ)} \dif \mQ$ is a probability measure which satisfies the large deviations property. Using results from large deviations \cite{dembo2009large}, the integral in \eqref{eq:rep-14} for large values of $n$ can be written as the integrand at the saddle point $\tilde{\mQ}$ multiplied by some bounded coefficient $\mathsf{K}_n$ which results in
\begin{align}
\sfZ(m) \doteq \mathsf{K}_n e^{-n \left[ \mg (\mT\tilde{\mQ})-\mai (\tilde{\mQ}) \right]} \label{eq:rep-17}
\end{align}
with $\doteq$ denoting the asymptotic equivalence in exponential scale\footnote{$a(\cdot)$ and $b(\cdot)$ are asymptotically equivalent in exponential scale over a non-bounded set $\setX$, if for $x \in \setX$ we have $\displaystyle \lim_{x \uparrow \infty} \log |\frac{a(x)}{b(x)}| =0$.}. Consequently, by substituting $\sfZ(m)$ in \eqref{eq:sm-8}, and exchanging the limits with respect to $n$ and $m$, as suggested in Assumption \ref{asp:2}, the asymptotic distortion is found as in Proposition \ref{proposition:1} where \eqref{eq:rep-6} determines the saddle point of the integrand function in \eqref{eq:rep-14}. The free energy is further determined as in \eqref{eq:rep-11} by substituting \eqref{eq:rep-17} in \eqref{eq:sm-0.4}. Finally by noting the fact that the free energy is minimized at the equilibrium, the proof is concluded.
\end{proof}

Proposition \ref{proposition:1} introduces a feasible way to determine the asymptotics of the \ac{map} estimator; its validity depends only on Assumptions \ref{asp:1} and \ref{asp:2} and has no further restriction. To pursue the analysis, one needs to solve the fixed point equation \eqref{eq:rep-6} for the replica correlation matrix $\mQ$ and calculate the parameters of the spin glass of replicas explicitly. The direct approach to find $\mQ$, however, raises both complexity and analyticity issues. In fact, finding the saddle point by searching over the set of all possible choices of $\mQ$ is a hard task to do; moreover, several solutions may not be of use since they do not lead to analytic $\sfF^\sfR(\upbeta,m)$ and $\sfD^\sfR(\upbeta,m)$ in $m$, and thus, they cannot be continued analytically to the real axis via Assumption \ref{asp:2}.

To overcome both the issues, the approach is to restrict our search into a parameterized set of replica correlation matrices and find the solution within this set. Clearly, the asymptotics found via this approach may fail as several other solutions are missed by restriction. The result, in this case, becomes more trustable by extending the restricted set of replica correlation matrices. Several procedures of restrictions can be considered. The procedures introduced in the literature are roughly divided into \ac{rs} and \ac{rsb} schemes. The former considers the $m$ replicas to interact symmetrically while the latter recursively breaks this symmetry in a systematic manner. In fact, the \ac{rs} scheme was first introduced due to some symmetric properties observed in the analysis of the spin glasses \cite{mezard2009information}. The properties, however, do not force the correlation matrix to have a symmetric structure, and later several examples were found showing that \ac{rs} leads to wrong conclusions. For these examples, the \ac{rsb} scheme was further considered as an extension of the symmetric structure of the correlation matrix to a larger set. We consider both the \ac{rs} and \ac{rsb} schemes in this manuscript; however, before pursuing our study, let us first investigate the general decoupling property of the estimator which can be concluded from Proposition \ref{proposition:1}.

\subsection{General Decoupling Property of \ac{map} Estimator}
Regardless of any restriction on $\mQ$, the general ansatz leads to the decoupling property of the \ac{map} estimator. In fact by using Proposition \ref{proposition:1}, it can be shown that for almost any tuple of input-output entries, the marginal joint distribution converges to a deterministic joint distribution which does not depend on the entries' index. The explicit term for the joint distribution, however, depends on the assumptions imposed on the correlation matrix. 

\begin{proposition}[General Decoupling Principle]
\label{proposition:2}
Let the linear system \eqref{eq:sys-1} fulfill the constraints of Section \ref{sec:problem_formulation}. Suppose Assumptions \ref{asp:1} and \ref{asp:2} hold, and consider the spin glass of replicas as defined in Definition \ref{def:replica_spin} with the replica correlation matrix $\mQ$. Then, for $j \in [1:n]$, the asymptotic conditional distribution of the \ac{map} estimator $\rmp_{\hx|x}^{j}$ defined in Definition \ref{def:asymp_distribution} is, almost sure in $\mA$, independent of $j$, namely
\begin{align}
\rmp_{\hx|x}^{j} (\hv|v) = \rmp_{\hxx|\xx} (\hv|v) \label{eq:rep-18}
\end{align}
for some conditional distribution $\rmp_{\hxx|\xx} $ at the mass point $(\hv,v) \in \setX \times \setX$. Consequently, the marginal joint distribution of the entries $x_j$ and $\hx_j$ is described by the input-output joint distribution of the single-user system specified by the conditional distribution $\rmp_{\hxx|\xx}$ and the input $\xx \sim \rmp_x$. The explicit form of $\rmp_{\hxx|\xx}$ depends on $\mQ$.
\end{proposition}

\begin{proof}
The proof follows from Proposition \ref{proposition:1} and the moment method \cite{akhiezer1965classical,simon1998classical}. From the classical moment problem, we know that the joint probability of the tuple of random variables $(t_1,t_2)$ are uniquely specified with the sequence of integer joint moments, if the joint moments are uniformly bounded. More precisely, by defining the $(k,\ell)$-joint moment of the tuple $(t_1,t_2)$ as
\begin{align}
\sfM_{k,\ell}\coloneqq \E t_1^k t_2^\ell, \label{eq:rep-19}
\end{align}
the sequence of $\left\lbrace \sfM_{k,\ell} \right\rbrace$ for $(k,\ell ) \in \setZ^+ \times \setZ^+$ is uniquely mapped to the probability distribution $\rmp_{t_1,t_2}$, if $\sfM_{k,\ell}$ is uniformly bounded for all integers $k$ and $\ell$. Consequently, one can infer that any two tuples of the random variables $(t_1,t_2)$ and $(\htt_1,\htt_2)$ with the same sequences of the joint moments are identical in distribution.

To determine the joint moment of input and output entries, consider the distortion function
\begin{align}
\sfd(\hx,x)= \hx^k x^\ell \label{eq:rep-20}
\end{align}
in Proposition \ref{proposition:1}, and evaluate the asymptotic distortion over the limit of $\setW(n)=[j:j+\eta n]$ for some $\eta$ in a right neighborhood of zero. The $(k,\ell)$-joint moment of $(\hx_j,x_j)$ is then determined by taking the limit $\eta\downarrow 0$.

Substituting the distortion function and the index set in Proposition \ref{proposition:1}, it is clear that the asymptotic distortion is independent of $\eta$ and $j$, and therefore, the limit with respect to $\eta$ exists and is independent of $j$ as well. Noting that the evaluated moments are uniformly bounded, it is inferred that the asymptotic joint distribution of $(\hx_j,x_j)$ is uniquely specified and does not depend on the index $j$. Finally, by using the fact that the source vector is \ac{iid} and the distribution of the entry $j$ is independent of the index, we conclude that the asymptotic conditional distribution $\rmp_{\hx|x}^{j}$ is independent of $j$. The exact expression for $\rmp_{\hx|x}^{j}$ is then found by determining the solution $\mQ$ to the fixed point equation and determining the joint moments.
\end{proof}

Proposition \ref{proposition:2} is a generalized form of the \ac{rs} decoupling principle for the \ac{map} estimators studied in \cite{rangan2012asymptotic}. In fact, Proposition \ref{proposition:2} indicates that a vector system followed by a \ac{map} estimator always decouples into a bank of identical single-user systems regardless of any restriction on the replica correlation matrix. %Therefore, the \ac{rs} decoupling principle can be considered as a special case of the general principle when $\mQ$ is restricted to be \ac{rs}.

\subsection{Consistency Test}
\label{subsec:cons_test}
If one restricts the replica correlation matrix $\mQ$ to be of a special form, the asymptotics determined under the assumed structure do not necessarily approximate true asymptotics accurately. Several methods were introduced in the literature to check the consistency of the solution. A primary method is based on calculating the entropy of the corresponding spin glass at zero temperature. As the temperature tends to zero, the distribution of the microstate tends to an indicator function at the point of the estimated vector, and consequently, the entropy of the corresponding spin glass converges to zero\footnote{Note that the entropy of the spin glass at temperature $\upbeta^{-1}$ denotes either the conditional entropy (for discrete supports) or the conditional differential entropy (for continuous supports) of a random vector $\bv$ distributed conditioned to $\by$ and $\mA$ with~Boltzmann-Gibbs~distribution~$\rmp^{\upbeta}_{\bv|\by,\mA}$.}. One consistency check is therefore the zero temperature entropy of a given solution.

Several works invoked this consistency test and showed that for the settings in which the \ac{rs} ansatz fails to give a tight bound on the exact solution, the zero temperature entropy determined from the \ac{rs} ansatz does not converge to zero; see for example \cite{zaidel2012vector}. %As an example, the investigations in \cite{zaidel2012vector} have shown that under \ac{rs}
%In several cases in which the \ac{rs} ansatz clearly fails, the zero temperature entropy not only does not converge to zero, but it also takes negative values; on the other hand, the \ac{rsb} ans\"atze take values near to zero; see for example \cite{zaidel2012vector}. 
This observation illustrates the invalidity of the \ac{rs} assumption and hints at \ac{rsb} ans\"atz giving better bounds on the true solution. Inspired by the aforementioned results, we evaluate the zero temperature entropy of the corresponding spin glass as a measure of consistency.

In order to determine the zero temperature entropy, we invoke \eqref{eq:sm-0.5} which determines the entropy in terms of the free energy at inverse temperature $\upbeta$. Considering the free energy of the corresponding spin glass as given in Proposition \ref{proposition:1}, the entropy $\rmH(\upbeta)$ reads
\begin{align}
\rmH(\upbeta)= \lim_{m\downarrow 0} \frac{\upbeta^2}{m}  \frac{\partial}{\partial \upbeta} \left[ \int_0^1 \tr{\mT \mQ \rmR_{\mJ}(-2\upbeta\omega\mT\mQ)} \dif \omega - \tr{\mQ^\trp \mT \rmR_{\mJ}(-2\upbeta \mT \mQ)}  \right] +\rmH^\sfR (\upbeta,m) \label{eq:rep-21}
\end{align}
where $\rmH^\sfR (\upbeta,m)$ denotes the normalized entropy of the spin glass of replicas. As $\rmH^\sfR (\upbeta,m)$ determines the entropy of a thermodynamic system, for any $m \in \setR^+$ we have
\begin{align}
\lim_{\upbeta \uparrow \infty} \rmH^\sfR (\upbeta,m) = 0, \label{eq:rep-22}
\end{align}
and therefore, the zero temperature entropy is given by
\begin{align}
\rmH^0= \lim_{\upbeta \uparrow \infty} \lim_{m\downarrow 0} \frac{\upbeta^2}{m}  \frac{\partial}{\partial \upbeta} \left[ \int_0^1 \tr{\mT \mQ \rmR_{\mJ}(-2\upbeta\omega\mT\mQ)} \dif \omega - \tr{\mQ^\trp \mT \rmR_{\mJ}(-2\upbeta \mT \mQ)}  \right], \label{eq:rep-23}
\end{align}
which obviously depends on the structure of the replica correlation matrix. In \cite{zaidel2012vector}, the authors determined the zero temperature entropy for the spin glass which corresponds to the vector precoding problem considering the \ac{rs} and 1\ac{rsb} assumptions, and observed that it takes the same form under both assumptions. They then conjectured that the zero temperature entropy is of a similar form for the general \ac{rsb} structure regardless of the number of breaking steps\footnote{In fact in this case, the dependence of the zero temperature entropy on the correlation matrix is completely described via the scalar $\chi$ which corresponds to the diagonal entries of the correlation matrix. See Assumption \ref{asp:3}-\ref{asp:5} for more illustrations.}. Using \eqref{eq:rep-23}, we later show that the conjecture in \cite{zaidel2012vector} is true.

\section{RS Ansatz and RS Decoupling Principle}
\label{sec:rs}
The most elementary structure which can be imposed on the replica correlation matrix is \ac{rs}. Here, one assumes the correlation matrix to be of a symmetric form which means that the replicas of the spin glass defined in Definition~\ref{def:replica_spin} are invariant under any permutation of indices. Using the definition of the replica correlation matrix as given in \eqref{eq:rep-6}, it consequently reads that
\begin{equation}
(x-\vv_a)(x-\vv_b)=
\begin{cases}
q_0 & \text{if}\ a \neq b \\
q_1 & \text{if}\ a=b.
\end{cases} \label{eq:rs-1}
\end{equation}

\begin{assumption}[\ac{rs} Structure]
\label{asp:3}
\normalfont
Considering the spin glass of replicas as defined in Definition \ref{def:replica_spin}, the replica correlation matrix is of the form
\begin{align}
\mQ=\frac{\chi}{\upbeta} \mI_m + q \mone_m \label{eq:rs-2}
\end{align}
where $\chi$ and $q$ are some non-negative real scalars.
\end{assumption}

Considering \eqref{eq:rs-1}, Assumption \ref{asp:3} supposes $q_0=q$ and $q_1={\chi}{\upbeta}^{-1}+q$. Substituting \eqref{eq:rs-2} in Definition \ref{def:replica_spin}, the spin glass of replicas is then specified by the scalars $\chi$ and $q$. The scalars are moreover related via a set of saddle point equations which are obtained from \eqref{eq:rep-6}. Finally, using Proposition \ref{proposition:1}, the asymptotics~of~the~system~are~found.

\begin{proposition}[RS Ansatz]
\label{proposition:3}
Let the linear system \eqref{eq:sys-1} fulfill the constraints of Section \ref{sec:problem_formulation}. Suppose Assumptions \ref{asp:1} and \ref{asp:2}, as well as Assumption \ref{asp:3} hold. Let $x\sim\rmp_x$, and
\begin{align}
\rmg = \arg \min_{v} \left[ \frac{1}{2\lams}(x+ \sqrt{\lams_0} z -v)^2 +  u(v) \right] \label{eq:rs-3}
\end{align}
with $v \in \setX$, and $\lams_0$ and $\lams$ being defined as
\begin{subequations}
\begin{align}
\lams_0 &\coloneqq \left[ \rmR_{\mJ}(-\frac{\chi}{\lambda})\right]^{-2} \frac{\partial}{\partial \chi} \left\lbrace \left[\lambda_0 \chi - \lambda q\right] \rmR_{\mJ}(-\frac{\chi}{\lambda}) \right\rbrace \label{eq:rs-4b}\\
\lams &\coloneqq \left[ \rmR_{\mJ}(-\frac{\chi}{\lambda}) \right]^{-1} \lambda \label{eq:rs-4a}
\end{align}
\end{subequations}
for some non-negative real $\chi$ and $q$ and the real variable $z$. Then, the asymptotic distortion defined in \eqref{eq:sys-9} reads
\begin{align}
\sfD^{\setW}= \E \int \sfd(\rmg,x) \md z, \label{eq:rs-5}
\end{align}
for $\chi$ and $q$ which satisfy
\begin{subequations}
\begin{align}
\sqrt{\lams_0} \ \chi&= \lams \ \E \int (\rmg-x) z \ \md z \label{eq:rs-6a} \\
q &= \E \int (\rmg-x)^2 \ \md z \label{eq:rs-6b}
\end{align}
\end{subequations}
and minimize the zero temperature free energy $\sfF^0_{\mathsf{rs}}$ which is given by
\begin{align}
\sfF^0_{\mathsf{rs}}=\frac{1}{2\lambda} \left[ \int_0^1 \rmF(\omega) \dif \omega -  \rmF (1) \right] + \E \int \frac{1}{2\lams} \left[ (x+\sqrt{\lams_0}z-\rmg)^2 - \lams_0 z^2 \right] + u(\rmg) \ \md z \label{eq:rs-7}
\end{align}
with $\rmF(\cdot)$ being defined as
\begin{align}
\rmF(\omega)=\left[q-\frac{\lambda_0}{\lambda} \chi \right] \frac{\dif}{\dif \omega} \left[ \omega \rmR_{\mJ}(-\frac{\chi}{\lambda} \omega) \right]. \label{eq:rs-8}
\end{align}
\end{proposition}
\begin{proof}
See Appendix \ref{app:b}.
\end{proof}

The asymptotic distortion under the \ac{rs} ansatz is equivalent to the average distortion of a scalar \ac{awgn} channel followed by a single-user estimator as shown in Fig. \ref{fig:1}. In this block diagram, the single-user estimator $\mathrm{g}_{\mathsf{map}}[(\cdot);\lams,u]$ maximizes the posterior probability over a postulated scalar \ac{awgn} channel. We refer to this estimator as the ``decoupled \ac{map} estimator'' and define it as follows.

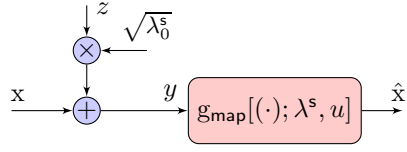
\begin{figure}[t]
\begin{center}
\begin{tikzpicture}[auto, node distance=2.5cm,>=latex']
    \node [input, name=input] {};
    \node [sum, right of=input,fill=blue!20] (sum) {$+$};
    \node [sum, above of=sum, node distance=.8cm,fill=blue!20] (times) {$\times$};
    \node [input, above of=times, node distance=.6cm] (noise) {};
    \node [input, right of=times, node distance=.8cm] (cof) {};
    \node [block, right of=sum,fill=red!20] (estimator) { $\mathrm{g}_{\mathsf{map}}[(\cdot);\lams,u]$ };
    \node [output, right of=estimator, node distance=1.7cm] (output) {};

    \draw [draw,->] (input) -- node[pos=.1] {$\xx$} (sum) ;
    \draw [->] (sum) -- node[pos=.8] {$y$} (estimator);
    \draw [->] (estimator) -- node[pos=.9] [name=x_h] {$\hxx$}(output);
	\draw [draw,->] (noise) -- node[pos=.1] {$z$}(times);
	\draw [draw,->] (times) -- node[pos=.1] {}(sum);	
	\draw [draw,->] (cof) -- node[pos=.01,above] {\small $\sqrt{\lams_0}$}(times);
\end{tikzpicture}
\end{center}
	\caption{The decoupled single-user system under the \ac{rs} ansatz.}
	\label{fig:1}
\end{figure}

\begin{definition}[Decoupled \ac{map} estimator]
\label{def:single_user_estimator}
The decoupled \ac{map} estimator $\mathrm{g}_{\mathsf{map}}[(\cdot);\lams,u]$ with the estimation parameter $\lams$ and the utility function $u(\cdot)$ is defined as
\begin{subequations}
\begin{align}
\mathrm{g}_{\mathsf{map}}[(y);\lams,u] &\coloneqq \arg \max_{v} \ \poster{v}{y} \label{eq:rs-10.1} \\
&=\arg \min_{v} \left[ \frac{1}{2\lams} (y -v)^2+u(v) \right], \label{eq:rs-10}
\end{align}
\end{subequations}
where $\poster{v}{y}$ denotes the ``decoupled posterior distribution'' postulated by the estimator which reads
\begin{align}
\poster{v}{y} =\mathsf{K} \ e^{-\left[ \tfrac{1}{2\lams} (y -v)^2+u(v) \right]} \label{eq:rs-10.2}
\end{align}
for some real constant $\mathsf{K}$.
\end{definition}

Using the definition of the decoupled \ac{map} estimator, the \ac{rs} decoupled system is defined next.

\begin{definition}[\ac{rs} decoupled system]
\label{def:rs_single_user}
Define the \ac{rs} decoupled system to be consistent with the block diagram in Fig. \ref{fig:1} in which
\begin{itemize}
\item the source symbol $\xx$ is distributed with $\rmp_{\xx}$ over the support $\setX$,
\item $z$ is a zero-mean and unit-variance Gaussian random variable,
\item $\xx$ and $z$ are independent,
\item $\hxx$ is estimated from the observation $y \coloneqq \xx+ \sqrt{\lams_0} z$ as
\begin{align}
\hxx=\mathrm{g}_{\mathsf{map}}[(y);\lams,u]. \label{eq:rs-9}
\end{align}
\item $\mathrm{g}_{\mathsf{map}}[(\cdot);\lams,u]$ is the decoupled \ac{map} estimator with the estimation parameter $\lams$ and the utility function $u(\cdot)$ as defined in Definition \ref{def:single_user_estimator}.
\item $\lams_0$ and $\lams$ are defined as in Proposition \ref{proposition:3}.
\end{itemize} 
\end{definition}

Using Proposition \ref{proposition:2}, the equivalency in the asymptotic distortion can be extended to the asymptotic conditional distribution as well. In fact, by considering the decoupling principle, Definition \ref{def:rs_single_user} describes the structure of the decoupled single-user system under the \ac{rs} assumption.

\begin{proposition}[\ac{rs} Decoupling Principle]
\label{proposition:4}
Let the linear system \eqref{eq:sys-1} fulfill the constraints of Section \ref{sec:problem_formulation} and be estimated via the \ac{map} estimator in \eqref{eq:int-2}. Consider the \ac{rs} decoupled system as defined in Definition \ref{def:rs_single_user}, and suppose Assumptions \ref{asp:1}, \ref{asp:2} and \ref{asp:3} hold. Then, for any $j\in [1:n]$, the tuple $(\hx_j,x_j)$ converges in distribution to $(\hxx,\xx)$ if $\rmp_x=\rmp_\xx$, almost sure in $\mA$.
\end{proposition}
\begin{proof}
Using Proposition \ref{proposition:2}, for any two different indices $j,q \in [1:n]$ we have
\begin{align}
\rmp^j_{\hx|x}(\hv|v)=\rmp^q_{\hx|x}(\hv|v) \label{eq:rs-11}
\end{align}
at the mass point $(\hv,v)$. Therefore for any index $j$, we have
\begin{align}
\E \hx_j^k x_j^\ell=\lim_{n\uparrow\infty} \frac{1}{n} \E \sum_{i=1}^n \hx_i^k x_i^\ell. \label{eq:rs-12}
\end{align}
Consequently, the asymptotic $(k,\ell)$-joint moment of $(\hx_j,x_j)$ under the \ac{rs} assumption is determined by letting $\setW(n)=[1:n]$ and the distortion function in the \ac{rs} ansatz
\begin{align}
\sfd(\hx;x)= \hx^k x^\ell \label{eq:rs-13}
\end{align}
and determining the asymptotic distortion. Substituting in Proposition \ref{proposition:3}, the asymptotic joint moment reads
\begin{align}
\sfM_{k,\ell}^j= \E \int \rmg^k x^\ell \md z \label{eq:rs-14}
\end{align}
where $\rmg$ is defined in \eqref{eq:rs-3}. Considering Definition \ref{def:rs_single_user} and assuming $\rmp_\xx=\rmp_x$, \eqref{eq:rs-14} describes the $(k,\ell)$-joint moment of $(\hxx,\xx)$ as well. Noting that $\sfM_{k,\ell}^j$ is uniformly bounded for any pair of integers $k$ and $\ell$, it is concluded that the asymptotic joint distribution of $(\hx_j,x_j)$ and the joint distribution of $(\hxx,\xx)$ are equivalent.
\end{proof}

Proposition \ref{proposition:4} gives a more general form of the \ac{rs} decoupling principles investigated in \cite{rangan2012asymptotic} and \cite{tulino2013support}. In fact, by restricting the system matrix and source distribution as in \cite{rangan2012asymptotic} and \cite{tulino2013support}, one can recover the formerly studied \ac{rs} decoupling principles.

\subsection*{RS Zero Temperature}
To have a basic measure of \ac{rs} ansatz's consistency, we evaluate the zero temperature entropy under the \ac{rs} assumption following the discussion in Section \ref{subsec:cons_test}. Substituting \eqref{eq:rs-2} in \eqref{eq:rep-23} and taking the same steps as in Appendix \ref{app:b}, the \ac{rs} zero temperature entropy is determined as
\begin{align}
\rmH^0_{\mathsf{rs}}=\lim_{\upbeta\uparrow\infty} \frac{\upbeta^2}{2\lambda} \frac{\partial}{\partial \upbeta} \left[ \int_0^1 \rmF^{\upbeta}(\omega) \dif \omega -  \rmF^{\upbeta} (1) \right] \label{eq:rs-15}
\end{align}
where the function $\rmF^{\upbeta}(\cdot)$ is defined as
\begin{align}
\rmF^{\upbeta}(\omega) = \frac{\chi}{\upbeta} \rmR_{\mJ}(-\frac{\chi}{\lambda} \omega) + \left[q-\frac{\lambda_0}{\lambda} \chi \right] \frac{\dif}{\dif \omega} \left[ \omega \rmR_{\mJ}(-\frac{\chi}{\lambda} \omega) \right]. \label{eq:rs-16}
\end{align}
Taking the derivative first and then the limit, it finally reads
\begin{align}
\rmH^0_{\mathsf{rs}}= \frac{\chi}{2\lambda}\left[ \rmR_{\mJ}(-\frac{\chi}{\lambda}) - \int_0^1 \rmR_{\mJ}(-\frac{\chi}{\lambda} \omega) \dif \omega  \right]. \label{eq:rs-17}
\end{align}
We later on see that the zero temperature entropy takes the same form under the \ac{rsb} assumptions.

\section{RSB Ans\"atze and RSB Decoupling Principle}
\label{sec:rsb}
In \cite{parisi1980sequence}, Parisi introduced a breaking scheme to broaden the restricted set of correlation matrices. The scheme recursively extends the set of matrices to larger sets. The breaking scheme was then employed to broaden the \ac{rs} structure of the correlation matrices, and therefore, the obtained structure was identified as the broken \ac{rs} or, \ac{rsb} structure. The key feature of Parisi's breaking scheme is that, by starting from the \ac{rs} structure, the new structure after breaking can be reduced to the structure before breaking. Thus, the set of fixed point solutions found by assuming the broken structure includes the solutions of the previous structure as well.

\begin{definition}[Parisi's breaking scheme]
\label{def:breaking_scheme}
Let $m$ be a multiple of the integer $\xi$, and $\mQ^{[\ell]}$ represent an $\frac{m}{\xi}\times\frac{m}{\xi}$ correlation matrix. Parisi's breaking scheme then constructs a new $m \times m$ correlation matrix $\mQ^{[\ell+1]}$ as
\begin{align}
\mQ^{[\ell+1]}=  \mI_\xi \otimes \mQ^{[\ell]} + \kappa \mone_m \label{eq:rsb-1}
\end{align}
for some real scalar $\kappa$.
\end{definition}

By choosing $\mQ^{[0]}$ to be an \ac{rs} correlation matrix in Definition \ref{def:breaking_scheme}, the matrix $\mQ^{[1]}$ finds the \ac{rsb} structure with one step of breaking (1\ac{rsb}). The steps of breaking can be further increased recursively by inserting $\mQ^{[1]}$ in the next breaking scheme, determining the new correlation matrix $\mQ^{[2]}$, and repeating the procedure. We start by the 1\ac{rsb} correlation matrix, and then, extend the result to higher \ac{rsb} ans\"atze with more steps of breaking.

\begin{assumption}[1\ac{rsb} Structure]
\label{asp:4}
\normalfont
Considering the spin glass of replicas as defined in Definition \ref{def:replica_spin}, the replica correlation matrix is of the form
\begin{align}
\mQ=\frac{\chi}{\upbeta} \mI_m + p \mI_{\frac{m\upbeta}{\mu}} \otimes \mone_{\frac{\mu}{\upbeta}} + q \mone_m \label{eq:rsb-2}
\end{align}
where $\chi$, $p$, $q$ and $\mu$ are some non-negative real scalars.
\end{assumption}

Regarding Parisi's breaking scheme, Assumption \ref{asp:4} considers $\mQ=\mQ^{[1]}$ by letting $\mQ^{[0]}$ have the \ac{rs} structure with parameters $\chi{\upbeta}^{-1}$ and $p$, $\xi={\mu}{\upbeta}^{-1}$ and $\kappa=q$. Here, the 1\ac{rsb} structure reduces to \ac{rs} by setting $p=0$. Therefore, the set of 1\ac{rsb} correlation matrices contains the one considered in Assumption \ref{asp:3}. 

\begin{proposition}[1RSB Ansatz]
\label{proposition:5}
Let the linear system \eqref{eq:sys-1} fulfill the constraints of Section \ref{sec:problem_formulation}. Suppose Assumptions \ref{asp:1} and \ref{asp:2}, as well as Assumption \ref{asp:4} hold. Let $x\sim\rmp_x$ and
\begin{align}
\rmg = \arg \min_{v} \left[ \frac{1}{2\lams}(x+ \sqrt{\lams_0} z_0 + \sqrt{\lams_1} z_1 -v)^2 +  u(v) \right] \label{eq:rsb-3}
\end{align}
with $v \in \setX$ and $\lams_0$, $\lams_1$ and $\lams$ being defined by
\begin{subequations}
\begin{align}
\lams_0 &= \left[ \mathrm{R}_{\mJ}(-\frac{\chi}{\lambda}) \right]^{-2} \frac{\partial}{\partial \chi} \left\lbrace \left[\lambda_0(\chi + \mu p)-\lambda q \right] \mathrm{R}_{\mJ}(- \frac{\chi+\mu p}{\lambda})\right\rbrace, \label{eq:rsb-3a} \\
\lams_1 &=\left[ \mathrm{R}_{\mJ}(-\frac{\chi}{\lambda}) \right]^{-2} \left[ \mathrm{R}_{\mJ}(- \frac{\chi}{\lambda}) - \mathrm{R}_{\mJ}(- \frac{\chi+\mu p}{\lambda}) \right] \lambda \mu^{-1}, \label{eq:rsb-3b} \\
\lams&= \left[ \mathrm{R}_{\mJ}(-\frac{\chi}{\lambda}) \right]^{-1} \lambda \label{eq:rsb-3c}
\end{align}
\end{subequations}
for some non-negative real $\chi$, $p$, $q$ and $\mu$ and real variables $z_0$ and $z_1$. Then, the asymptotic distortion~in~\eqref{eq:sys-9}~reads
\begin{align}
\sfD^{\setW}= \E \int \sfd(\rmg,x) \ \tilde{\Lambda}{(z_1|z_0,x)} \md z_1 \md z_0, \label{eq:rsb-4}
\end{align}
with $\tilde{\Lambda}{(z_1|z_0,x)} \coloneqq \left[ \int \Lambda{(z_1,z_0,x)} \md z_1 \right]^{-1} \Lambda{(z_1,z_0,x)}$ and $\Lambda {(z_1,z_0,x)}$ being defined by
\begin{align}
\Lambda {(z_1,z_0,x)} = e^{- \mu \left[ \tfrac{1}{2 \lams}(x+ \sqrt{\lams_0} z_0 + \sqrt{\lams_1} z_1 -\rmg)^2 - \tfrac{1}{2\lams} (\sqrt{\lams_0} z_0 + \sqrt{\lams_1} z_1 )^2 + u(\mathrm{g})\right]}. \label{eq:rsb-5}
\end{align}
The scalars $\chi$, $p$ and $q$ satisfy
\begin{subequations}
\begin{align}
\sqrt{\lams_0}\left[ \chi+\mu p \right] &= \lams \E \int (\mathrm{g}-x) z_0 \ \tilde{\Lambda}{(z_1|z_0,x)} \md z_1 \md z_0 \label{eq:rsb-6a}\\
\sqrt{\lams_1}\left[ \chi+\mu q + \mu p \right] &= \lams \E \int (\mathrm{g}-x) z_1 \ \tilde{\Lambda}{(z_1|z_0,x)} \md z_1 \md z_0 \label{eq:rsb-6b}\\
q+p &= \E \int (\rmg-x)^2 \ \tilde{\Lambda}{(z_1|z_0,x)} \md z_1 \md z_0  \label{eq:rsb-6c}
\end{align}
\end{subequations}
and $\mu$ is a solution of the fixed point equation
\begin{align}
\frac{\mu}{2\lams}\left[\mu\frac{\lams_1}{\lams}q+p\right]-\frac{1}{2\lambda}\int_0^{\mu p} \mathrm{R}_{\mJ}(-\frac{\chi+\omega}{\lambda}) \dif \omega = \E \int \log\tilde{\Lambda}{(z_1|z_0,x)} \ \tilde{\Lambda}{(z_1|z_0,x)} \md z_1 \md z_0. \label{eq:rsb-7}
\end{align}
Moreover, $\chi$, $p$, $q$ and $\mu$ minimize the zero temperature free energy $\sfF^{0[1]}_{\mathsf{rsb}}$ which reads
\begin{align}
\sfF^{0[1]}_{\mathsf{rsb}}=\frac{1}{2\lambda} \left[ \int_0^1 \rmF(\omega) \dif \omega -  \rmF (1) \right] - \frac{1}{\mu} \E \int \log \left[ \int \Lambda{(z_1,z_0,x)} \md z_1 \right] \ \md z_0 \label{eq:rsb-8}
\end{align}
with $\rmF(\cdot)$ being defined as
\begin{align}
\rmF(\omega)=\frac{1}{\mu} \frac{\dif}{\dif \omega} \left[\int_{\chi \omega}^{\left[\chi+\mu p\right] \omega} \rmR_{\mJ}(-\frac{t}{\lambda} ) \dif t \right] + \left[q-\lambda_0 \frac{\chi+\mu p}{\lambda} \right] \frac{\dif}{\dif \omega} \left[ \omega \rmR_{\mJ}(-\frac{\chi+\mu p}{\lambda} \omega) \right]. \label{eq:rsb-9}
\end{align}
\end{proposition}
\begin{proof}
See Appendix \ref{app:c}.
\end{proof}

Similar to our approach under the \ac{rs} ansatz, we employ Proposition \ref{proposition:5} to introduce the decoupled 1\ac{rsb} single-user system which describes the statistical behavior of the \ac{map} estimator's input-output entries under 1\ac{rsb} assumption. The decoupled system under 1\ac{rsb} differs from \ac{rs} within a new impairment term which is correlated with the source and noise symbols through a joint distribution. The impairment term intuitively plays the role of a correction factor which compensates the possible inaccuracy of the \ac{rs} ansatz. The decoupled \ac{map} estimator, however, follows the same structure as for \ac{rs}.

\begin{figure}[t]
\begin{center}
\begin{tikzpicture}[auto, node distance=2.5cm,>=latex']
    \node [input, name=input] {};
    \node [margin, right of=input, node distance=1.6cm,minimum height=2.2em,minimum width=3.5em,fill=yellow!20] (margin) {};
    \node [sum, right of=input,node distance=1.6cm, fill=blue!20] (sum0) {$+$};
    \node [sum, above of=sum0, node distance=.8cm,fill=blue!20] (times0) {$\times$};
    \node [input, right of=times0, node distance=.9cm] (cof0) {};
    \node [sum, right of=sum0, node distance=1.8cm,fill=blue!20] (sum1) {$+$};
    \node [sum, above of=sum1, node distance=.8cm,fill=blue!20] (times1) {$\times$};
    \node [input, right of=times1, node distance=.7cm] (cof1) {};
    \node [block, above of=times0, node distance=1.2cm,minimum height=1.7em, minimum width=5em,fill=red!20] (cond) { $\rmp_{z_1|z_0,\xx}$ };
    \node [input, above of=times1, node distance=.8cm] (noise1) {};    
    \node [block, right of=sum1,fill=red!20] (estimator) { $\mathrm{g}_{\mathsf{map}}[(\cdot);\lams,u]$ };
    \node [output, right of=estimator, node distance=1.7cm] (output) {};

    \draw [draw,->] (input) -- node[pos=.1] {$\xx$} (sum0) ;
    \draw [->] (sum0) -- node { } (sum1);
    \draw [->] (sum1) -- node[pos=.8] {$y$} (estimator);
    \draw [->] (estimator) -- node[pos=.9] [name=x_h] {$\hxx$}(output);
	\draw [draw,->] (times0) -- node[pos=.1] {}(sum0);
	\draw [draw,->] (times1) -- node[pos=.1] {}(sum1);
	\draw [draw,->] (cof0) -- node[pos=.01,above] {\small $\sqrt{\lams_1}$}(times0);
	\draw [draw,->] (cof1) -- node[pos=.01,above] {\small $\sqrt{\lams_0}$}(times1);
	
	\draw [draw,->] (cond) -- node[pos=.3] {$z_1$}(times0);
	\draw [draw,->] (noise1) -- node[pos=.1] {$z_0$}(times1);	
\end{tikzpicture}
\end{center}
	\caption{The decoupled scalar system under the 1\ac{rsb} ansatz.}
	\label{fig:2}
\end{figure}

\begin{definition}[1\ac{rsb} decoupled system]
\label{def:1rsb_single_user}
Fig. \ref{fig:2} defines the 1\ac{rsb} decoupled system in which
\begin{itemize}
\item the source symbol $\xx$ is distributed with $\rmp_{\xx}$ over the support $\setX$,
\item $z_0$ is a zero-mean and unit-variance Gaussian random variable,
\item $z_1$ is a random variable correlated with $\xx$ and $z_0$,
\item $\xx$ and $z_0$ are independent, and 
\begin{align}
\rmp_{z_1|\xx,z_0}=\tilde{\Lambda} {(z_1|z_0,x)} \phi(z_1) \label{eq:rsb-10}
\end{align}
with $\tilde{\Lambda}$ defined in Proposition \ref{proposition:5}.
\item $\hxx$ is estimated from the observation $y \coloneqq \xx+ \sqrt{\lams_0} z_0+ \sqrt{\lams_1} z_1$ as
\begin{align}
\hxx=\mathrm{g}_{\mathsf{map}}[(y);\lams,u]. \label{eq:rsb-11}
\end{align}
\item $\mathrm{g}_{\mathsf{map}}[(\cdot);\lams,u]$ is the decoupled \ac{map} estimator with the estimation parameter $\lams$ and the utility function $u(\cdot)$ as defined in Definition \ref{def:single_user_estimator}.
\item $\lams_0$, $\lams_1$ and $\lams$ are defined as in Proposition \ref{proposition:5}.
\end{itemize} 
\end{definition}

\begin{proposition}[1\ac{rsb} Decoupling Principle]
\label{proposition:6}
Let the linear system \eqref{eq:sys-1} fulfill the constraints of Section \ref{sec:problem_formulation} and be estimated via the \ac{map} estimator in \eqref{eq:int-2}. Consider the 1\ac{rsb} decoupled system as defined in Definition \ref{def:1rsb_single_user}, and suppose Assumptions \ref{asp:1}, \ref{asp:2} and \ref{asp:4} hold. Then, for all $j\in [1:n]$, the tuple $(\hx_j,x_j)$ converges in distribution to $(\hxx,\xx)$  if $\rmp_x=\rmp_\xx$.
\end{proposition}
\begin{proof}
The proof takes exactly same steps as for the \ac{rs} decoupling principle using Proposition \ref{proposition:5}.
\end{proof}

The 1\ac{rsb} decoupled system, in general, provides a more accurate approximation of the estimator's asymptotics by searching over a larger set of solutions which include the \ac{rs} ansatz. To investigate the latter statement, consider the case of $p=0$. In this case, the 1\ac{rsb} structure reduces to \ac{rs}. Setting $p=0$ in Proposition \ref{proposition:5}, $\lambda_1$ becomes zero, and consequently $\tilde{\Lambda}(z_1|z_0,x)=1$. Moreover, the fixed point equations in \eqref{eq:rsb-7} hold for any choice of $\mu$, and the scalars $\chi$ and $q$ couple through the same set of equations as in the \ac{rs} ansatz. The zero temperature free energy of the system, furthermore, reduces to its \ac{rs} form under the assumption of $p=0$. Denoting the parameters of the \ac{rs} ansatz by $\left[ \chi_{\mathsf{rs}}, q_{\mathsf{rs}} \right]$, it is then concluded that $\left[ \chi,p, q, \mu \right]=\left[ \chi_{\mathsf{rs}},0, q_{\mathsf{rs}},0 \right]$ is a solution to the 1\ac{rsb} fixed point equations, when an stable solution to the \ac{rs} fixed point exists. The solution, however, does not give necessarily the 1\ac{rsb} ansatz, the stable solution to the 1\ac{rsb} fixed point equations with minimum free energy may occur at some other point. We investigate the impact of replica breaking for some examples later through numerical results.

Parisi's breaking scheme can be employed to extend the structure of the correlation matrix to the \ac{rsb} structure with more steps of breaking by recursively repeating the scheme. In fact, one can start from an \ac{rs} structured $\mQ^{[0]}$ and employ the breaking scheme for $b$ steps to determine the correlation matrix $\mQ^{[b]}$. In this case, the replica correlation matrix is referred to as the $b$\ac{rsb} correlation matrix.

\begin{assumption}[$b$\ac{rsb} Structure]
\label{asp:5}
\normalfont
Considering the spin glass of replicas as defined in Definition \ref{def:replica_spin}, the replica correlation matrix is of the form
\begin{align}
\mQ=\frac{\chi}{\upbeta} \mI_m + \sum_{\kappa=1}^b p_\kappa \mI_{\frac{m\upbeta}{\mu_\kappa}} \otimes \mone_{\frac{\mu_\kappa}{\upbeta}} + q \mone_m \label{eq:rsb-12}
\end{align}
where scalars $\chi$ and $q$, and the sequences $\{p_\kappa\}$ and $\{\mu_\kappa\}$ with $\kappa\in[1:b]$ are non-negative reals, and $\{\mu_\kappa\}$  satisfies the following constraint
\begin{align}
\frac{\mu_{\kappa+1}}{\mu_{\kappa}} \in \setZ^+ \label{eq:rsb-13}
\end{align}
for $\kappa\in[1:b-1]$.
\end{assumption}
Considering the correlation matrix in Proposition \ref{proposition:1} to be of the form indicated in Assumption \ref{asp:5} the previous ans\"atze are extended to a more general ansatz which can reduce to the 1\ac{rsb} as well as \ac{rs} ansatz. Proposition~\ref{proposition:7} expresses the replica ansatz under the $b$\ac{rsb} assumption.

\begin{proposition}[$b$RSB Ansatz]
\label{proposition:7}
Let the linear system \eqref{eq:sys-1} fulfill the constraints of Section \ref{sec:problem_formulation}. Suppose Assumptions \ref{asp:1} and \ref{asp:2}, as well as Assumption \ref{asp:5} hold. For $\kappa \in [0:b]$, define the sequence $\{ \tilde{\chi}_\kappa \}$, such that $\tilde{\chi}_0=\chi$ and
\begin{align}
\tilde{\chi}_\kappa \coloneqq \chi+\sum_{\varsigma=1}^{\kappa} \mu_\varsigma p_\varsigma \label{eq:rsb-14}
\end{align}
for $\kappa \in [1:b]$. Let $x\sim\rmp_x$, and
\begin{align}
\rmg = \arg \min_{v} \left[ \frac{1}{2\lams}(x+ \sum_{\kappa=0}^b \sqrt{\lams_\kappa} z_\kappa-v)^2 +  u(v) \right] \label{eq:rsb-15}
\end{align}
with $v \in \setX$ and $\lams_0$, $\{ \lams_\kappa \}$ and $\lams$ being defined as
\begin{subequations}
\begin{align}
\lams_0 &= \left[ \mathrm{R}_{\mJ}(-\frac{\tilde{\chi}_0}{\lambda}) \right]^{-2} \frac{\partial}{\partial \tilde{\chi}_{b}} \left\lbrace \left[ \lambda_0 \tilde{\chi}_{b} - \lambda q \right] \rmR_{\mJ}(- \frac{\tilde{\chi}_{b}}{\lambda}) \right\rbrace, \label{eq:rsb-16a} \\
\lams_\kappa &=\left[ \mathrm{R}_{\mJ}(-\frac{\tilde{\chi}_0}{\lambda}) \right]^{-2} \left[ \rmR_{\mJ}(- \frac{\tilde{\chi}_{\kappa-1}}{\lambda}) - \rmR_{\mJ}(-\frac{\tilde{\chi}_{\kappa}}{\lambda}) \right] \lambda \mu_\kappa^{-1}, \label{eq:rsb-16b} \\
\lams&= \left[ \mathrm{R}_{\mJ}(-\frac{\tilde{\chi}_0}{\lambda}) \right]^{-1} \lambda \label{eq:rsb-16c}
\end{align}
\end{subequations}
for some non-negative real scalar $q$, sequences $\{ \tilde{\chi}_\kappa \}$ and $\{\mu_\kappa\}$ and sequence of real variables $\{z_\kappa\}$. Then, the asymptotic distortion defined in \eqref{eq:sys-9} reads
\begin{align}
\sfD^{\setW}&= \E \int \sfd(\rmg;x) \prod_{\kappa=1}^b \tilde{\Lambda}_\kappa{(z_\kappa| \{z_\varsigma\}_{\kappa+1}^b, z_0,x)} \ \md z_\kappa \md z_0, \label{eq:rsb-17}
\end{align}
where $ {\{z_\varsigma\}_{\kappa}^b\coloneqq \{z_\kappa,\ldots,z_b\}}$ and $\tilde{\Lambda}_\kappa {(z_\kappa| \{z_\varsigma\}_{\kappa+1}^b, z_0,x)} = \left[ \int \Lambda_\kappa {(\{z_\varsigma\}_{\kappa}^b, z_0,x)} \md z_\kappa\right]^{-1} \Lambda_\kappa {(\{z_\varsigma\}_{\kappa}^b, z_0,x)}$ with
\begin{align}
\Lambda_{1}  {(\{z_\varsigma\}_{1}^b, z_0,x)} \coloneqq  e^{-\mu_1  \left[  \tfrac{1}{2 \lams}(x+ \sum\limits_{\kappa=0}^b \sqrt{\lams_\kappa} z_\kappa -\rmg)^2 - \tfrac{1}{2\lams} (\sum\limits_{\kappa=0}^b \sqrt{\lams_\kappa} z_\kappa )^2 + u(\mathrm{g}) \right]} \label{eq:rsb-18}
\end{align}
and $\{\Lambda_\kappa {(\{z_\varsigma\}_{\kappa}^b, z_0,x)}\}$ for $\kappa\in[2:b]$ being recursively determined by
\begin{align}
\Lambda_{\kappa} {(\{z_\varsigma\}_{\kappa}^b, z_0,x)} \coloneqq \left[ \int \Lambda_{\kappa-1} {(\{z_\varsigma\}_{\kappa-1}^b, z_0,x)} \ \md z_{\kappa-1} \right]^{\tfrac{\mu_{\kappa}}{\mu_{\kappa-1}}}. \label{eq:rsb-19}
\end{align}
The scalar $q$ and sequences $\{ \tilde{\chi}_\kappa \}$ and $\{p_\kappa\}$ satisfy
\begin{subequations}
\begin{align}
\sum_{\kappa=1}^b p_\kappa + q &= \E \int (\rmg-x)^2 \prod_{\kappa=1}^b \tilde{\Lambda}_\kappa  {(z_\kappa| \{z_\varsigma\}_{\kappa+1}^b, z_0,x)} \ \md z_\kappa \md z_0,\label{eq:rsb-20a} \\
\tilde{\chi}_{\kappa-1}+\mu_\kappa \left(\sum_{\varsigma=\kappa}^b p_\varsigma + q \right) &= \frac{\lams}{\sqrt{\lams_\kappa}} \E \int (\rmg-x) z _\kappa \prod_{\kappa=1}^b \tilde{\Lambda}_\kappa {(z_\kappa| \{z_\varsigma\}_{\kappa+1}^b, z_0,x)} \ \md z_\kappa \md z_0 , \label{eq:rsb-20b} \\
\tilde{\chi}_{b} &= \frac{\lams}{\sqrt{\lams_0}} \E \int (\rmg-x) z _0  \prod_{\kappa=1}^b \tilde{\Lambda}_\kappa  {(z_\kappa| \{z_\varsigma\}_{\kappa+1}^b, z_0,x)} \ \md z_\kappa \md z_0, \label{eq:rsb-20c}
\end{align}
\end{subequations}
and the sequence $\{ \mu_\kappa\}$ is given by
\begin{align}
\{ \mu_\kappa \} =\arg\min_{\{ \tilde{\mu}_\kappa \} } \left[ \frac{1}{2\lambda} \int_0^1 \rmF(\omega;\{ \tilde{\mu}_\kappa \}) \dif \omega - \frac{1}{\tilde{\mu}_b} \E \int \log \left[ \int \Lambda_b {(z_b, z_0,x)} \md z_b \right] \ \md z_0 - \frac{1}{2\lams} \Delta(\{ \tilde{\mu}_\kappa \}) \right] \label{eq:rsb-20.1}
\end{align}
where the function $\rmF(\cdot;\{ \mu_\kappa \})$ is defined as
\begin{align}
\rmF(\omega;\{ \mu_\kappa \}) \coloneqq \sum\limits_{\kappa=1}^b \frac{1}{\mu_\kappa} \frac{\dif}{\dif \omega} \int_{\tilde{\chi}_{\kappa-1} \omega}^{\tilde{\chi}_{\kappa} \omega} \rmR_{\mJ}(-\frac{t}{\lambda} ) \dif t + \left[q- \frac{\lambda_0}{\lambda} \tilde{\chi}_b \right] \frac{\dif}{\dif \omega} \left[ \omega \rmR_{\mJ}(-\frac{\tilde{\chi}_b}{\lambda} \omega) \right],
\end{align}
$\Delta(\cdot)$ is defined as
\begin{align}
\Delta(\{ \mu_\kappa \})=\sum_{\kappa=1}^b \frac{1}{\mu_\kappa} \left[\zeta_\kappa\tilde{\chi}_\kappa -\zeta_{\kappa-1} \tilde{\chi}_{\kappa-1} \right]+\zeta_b q -\frac{\lams_0}{\lams} \tilde{\chi}_b \label{eq:rsb-21}
\end{align}
with $\zeta_0=1$, and $\zeta_\kappa$ for $\kappa \in [1:b]$ denoting
\begin{align}
\zeta_\kappa \coloneqq 1 - \sum_{\varsigma=1}^{\kappa} \mu_\varsigma \frac{\lams_\varsigma}{\lams}, \label{eq:rsb-22}
\end{align}
and $\{ \tilde{\mu}_\kappa \} \in \setS_{\bmu}$ in which
\begin{align}
\setS_{\bmu} \coloneqq \left\lbrace \{ \mu_1, \ldots, \mu_b \} \ni \mu_\kappa\in\setR^+ \ \wedge \ \frac{\mu_{\kappa+1}}{\mu_{\kappa}} \in \setZ^+ \ \forall \kappa \in [1:b-1] \right\rbrace. \label{eq:rsb-23}
\end{align}
In the case of multiple solutions for $\chi$, $q$, $\{ p_\kappa \}$ and $\{ \mu_\kappa \}$, the ansatz is chosen such that the free energy at zero temperature
\begin{align}
\sfF^{0[b]}_{\mathsf{rsb}}=\frac{1}{2\lambda} \left[ \int_0^1 \rmF(\omega;\{\mu_\kappa\}) \dif \omega -  \rmF (1;\{\mu_\kappa\}) \right] - \frac{1}{\mu_b} \E \int \log \left[ \int \Lambda_b  {(z_b, z_0,x)} \md z_b \right] \ \md z_0 \label{eq:rsb-24}
\end{align}
is minimized.
\end{proposition}
\begin{proof}
See Appendix \ref{app:d}.
\end{proof}

One can simply observe that Proposition \ref{proposition:7} reduces to Propositions \ref{proposition:5} and \ref{proposition:3} by letting $b=1$ and $p_\kappa=0$ for $\kappa\in[1:b]$, respectively. The ansatz, moreover, extends the corresponding decoupled single-user system of the estimator considering the general decoupling principle investigated in Proposition \ref{proposition:2}. By taking the same steps as in Proposition \ref{proposition:6}, the decoupled $b$\ac{rsb} single-user system is found which represents the extended version of the 1\ac{rsb} system with $b$ additive impairment taps. In fact, considering the impairment terms to intuitively play the role of correction factors, the $b$\ac{rsb} ansatz takes more steps of correction into account. The decoupled \ac{map} estimator, moreover, remains unchanged.

\begin{figure}[t]
\begin{center}
\begin{tikzpicture}[auto, node distance=2.5cm,>=latex']
    \node [input, name=input] {};
    \node [margin, right of=input, node distance=3.6cm,minimum width=15em, fill=yellow!20] (margin) {};
    \node [sum, right of=input,node distance=1.9cm, fill=blue!20] (sum0) {$+$};
    \node [sum, above of=sum0, node distance=.8cm,fill=blue!20] (times0) {$\times$};
    \node [input, right of=times0, node distance=.9cm] (cof0) {};
    \node [input, right of=sum0,node distance=1.2cm] (inter1) {};
    \node [input, right of=inter1,node distance=1cm] (inter2) {};
	
	\node [sum, right of=inter2,node distance=1.2cm,fill=blue!20] (sumb) {$+$};
    \node [sum, above of=sumb, node distance=.8cm,fill=blue!20] (timesb) {$\times$};
    \node [input, right of=timesb, node distance=.9cm] (cofb) {};

    \node [sum, right of=sumb, node distance=1.8cm, fill=blue!20] (sum1) {$+$};
    \node [sum, above of=sum1, node distance=.8cm,fill=blue!20] (times1) {$\times$};
    \node [input, right of=times1, node distance=.7cm,fill=red!20] (cof1) {};
    
    \node [block, above of=times0, node distance=1.2cm,minimum height=1.7em, minimum width=8em,fill=red!20] (cond) { $\rmp_{z_1|z_2,\ldots,z_b,z_0,\xx}$ };
     \node [block, above of=timesb, node distance=1.2cm,minimum height=1.7em, minimum width=5em,fill=red!20] (condb) { $\rmp_{z_b|z_0,\xx}$ };
    \node [input, above of=times1, node distance=.8cm] (noise1) {};    
    \node [block, right of=sum1,fill=red!20] (estimator) { $\mathrm{g}_{\mathsf{map}}[(\cdot);\lams,u]$ };
    \node [output, right of=estimator, node distance=1.7cm] (output) {};

    \draw [draw,->] (input) -- node[pos=.1] {$\xx$} (sum0) ;
    \draw [->] (sumb) -- node { } (sum1);
    \draw [->] (sum1) -- node[pos=.8] {$y$} (estimator);
    \draw [->] (estimator) -- node[pos=.9] [name=x_h] {$\hxx$}(output);
	\draw [draw,->] (times0) -- node[pos=.1] {}(sum0);
	\draw [draw,->] (times1) -- node[pos=.1] {}(sum1);
	\draw [draw,->] (timesb) -- node[pos=.1] {}(sumb);
	\draw [-,dotted] (inter1) -- node { } (inter2);
	\draw [draw,->] (cof0) -- node[pos=.01,above] {\small $\sqrt{\lams_1}$}(times0);
	\draw [draw,->] (cofb) -- node[pos=.01,above] {\small $\sqrt{\lams_b}$}(timesb);
	\draw [draw,->] (cof1) -- node[pos=.01,above] {\small $\sqrt{\lams_0}$}(times1);
	
	\draw [draw,->] (cond) -- node[pos=.3] {$z_1$}(times0);
	\draw [draw,->] (condb) -- node[pos=.3] {$z_b$}(timesb);
	\draw [draw,->] (noise1) -- node[pos=.1] {$z_0$}(times1);	
\end{tikzpicture}
\end{center}
	\caption{The decoupled scalar system under the \ac{brsb} ansatz.}
	\label{fig:3}
\end{figure}

\begin{definition}[$b$\ac{rsb} decoupled system]
\label{def:brsb_single_user}
Define the $b$\ac{rsb} decoupled system as a single-user system illustrated in Fig. \ref{fig:3} in which
\begin{itemize}
\item the source symbol $\xx$ is distributed with $\rmp_{\xx}$ over the support $\setX$,
\item $z_0$ is a zero-mean and unit-variance Gaussian random variable,
\item $z_\kappa$ is a random variable correlated with $\xx$, $z_0$ and $\{ z_{\kappa+1}, \ldots, z_b \}$
\item $\xx$ and $z_0$ are independent, and 
\begin{align}
\rmp_{z_\kappa|z_{\kappa+1}, \ldots, z_b,z_0, \xx}=\tilde{\Lambda}_\kappa  {(z_\kappa| \{z_\varsigma\}_{\kappa+1}^b, z_0,x)} \phi(z_\kappa) \label{eq:rsb-25}
\end{align}
with $\tilde{\Lambda}$ defined in Proposition \ref{proposition:7}.
\item $\hxx$ is estimated from the observation $y \coloneqq \xx+ \sum\limits_{\kappa=0}^b \sqrt{\lams_\kappa} z_\kappa$ as
\begin{align}
\hxx=\mathrm{g}_{\mathsf{map}}[(y);\lams,u]. \label{eq:rsb-26}
\end{align}
\item $\mathrm{g}_{\mathsf{map}}[(\cdot);\lams,u]$ is the decoupled \ac{map} estimator with the estimation parameter $\lams$ and the utility function $u(\cdot)$ as defined in Definition \ref{def:single_user_estimator}, and
\item $\lams_0$, $\{\lams_\kappa\}$ and $\lams$ for $\kappa\in[1:b]$ are as in Proposition \ref{proposition:7}.
\end{itemize} 
\end{definition}

\begin{proposition}[$b$\ac{rsb} Decoupling Principle]
\label{proposition:8}
Let the linear system \eqref{eq:sys-1} fulfill the constraints of Section \ref{sec:problem_formulation} and be estimated via the \ac{map} estimator in \eqref{eq:int-2}. Consider the $b$\ac{rsb} decoupled system as defined in Definition \ref{def:brsb_single_user}, and suppose Assumptions \ref{asp:1}, \ref{asp:2} and \ref{asp:5} hold. Then, for all $j\in [1:n]$, the tuple $(\hx_j,x_j)$ converges in distribution to $(\hxx,\xx)$  if $\rmp_x=\rmp_\xx$.
\end{proposition}
\begin{proof}
Using Proposition \ref{proposition:5}, it takes same steps as for Proposition \ref{proposition:4}.
\end{proof}

\subsection*{RSB Zero Temperature}
In Appendix \ref{app:d}, it is shown that under the $b$\ac{rsb} assumption on the replica correlation matrix the free energy of the corresponding spin glass at the inverse temperature $\upbeta$ reads 
\begin{align}
\sfF(\upbeta)=\frac{1}{2\lambda} \left[ \int_0^1 \rmF^{\upbeta}(\omega) \dif \omega -  \rmF^{\upbeta} (1) \right] +\sfF^{\sfR}(\upbeta). \label{eq:rsb-27}
\end{align}
Here, $\sfF^{\sfR}(\upbeta)$ denotes the normalized free energy of the spin glass of replicas defined in \eqref{eq:rep-9} in the limit $m\downarrow0$, and the function $\rmF^{\upbeta}(\cdot)$ is defined as
\begin{align}
\rmF^{\upbeta}(\omega) = \sum\limits_{\kappa=1}^b \frac{1}{\mu_\kappa} \frac{\dif}{\dif \omega} \int_{\tilde{\chi}_{\kappa-1} \omega}^{\tilde{\chi}_{\kappa} \omega} \rmR_{\mJ}(-\frac{t}{\lambda} ) \dif t + \frac{\chi}{\upbeta} \rmR_{\mJ}(-\frac{\chi}{\lambda} \omega) + \left[q- \frac{\lambda_0}{\lambda} \tilde{\chi}_b \right] \frac{\dif}{\dif \omega} \left[ \omega \rmR_{\mJ}(-\frac{\tilde{\chi}_b}{\lambda} \omega) \right].  \label{eq:rsb-28}
\end{align}
Following the discussion in Section \ref{subsec:cons_test}, the entropy at the zero temperature reads
\begin{align}
\rmH^{0[b]}_{\mathsf{rsb}}=\lim_{\upbeta\uparrow\infty} \frac{\upbeta^2}{2\lambda} \frac{\partial}{\partial \upbeta} \left[ \int_0^1 \rmF^{\upbeta}(\omega) \dif \omega -  \rmF^{\upbeta} (1) \right] \label{eq:rs-29}
\end{align}
which reduces to
\begin{align}
\rmH^{0[b]}_{\mathsf{rsb}}= \frac{\chi}{2\lambda}\left[ \rmR_{\mJ}(-\frac{\chi}{\lambda}) - \int_0^1 \rmR_{\mJ}(-\frac{\chi}{\lambda} \omega) \dif \omega  \right]. \label{eq:rsb-30}
\end{align}
\eqref{eq:rsb-30} justifies the conjecture in \cite{zaidel2012vector} and states that the zero temperature entropy under any number of breaking steps, including the \ac{rs} case, is of the similar form and only depends on the scalar $\chi$. In fact, the Hamiltonian in \eqref{eq:int-7} reduces to the one considered in vector precoding by considering $\bx$ to be the deterministic vector of zeros, $\lambda_0=0$, $\lambda=1$ and $u(\bv)=0$. Substituting in \eqref{eq:rsb-30}, the zero temperature entropy reduces to the one determined in \cite{zaidel2012vector} within a factor of $2$. The factor comes from the difference in the prior assumption on the support~of~microstate\footnote{In \cite{zaidel2012vector}, the authors considered $\bv$ to be a complex vector.}.

\section{Replica Simulator: Characterization via the Single-User Representation}
\label{sec:rep_sim}
The general $b$\ac{rsb} decoupling principle determines an equivalent single-user system which describes the input-output statistics of the \ac{map} estimator under the $b$\ac{rsb} ansatz. In order to specify the exact parameters of the decoupled single-user system, the set of fixed point equations needs to be solved explicitly. In this section, we propose an alternative approach which describes an ansatz in terms of the corresponding decoupled system's input-output statistics. We define the exact form of the decoupled system as the ``steady state'' of a transition system named ``replica simulator''. The proposed approach enables us to investigate the properties of the \ac{rs} and \ac{rsb} ans\"atze by studying the replica simulator. In order to clarify the idea of the replica simulator, let us define a set of input-output statistics regarding the $b$\ac{rsb} decoupled system.

\begin{definition}
\label{def:single_user_parameters}
Consider the single-user system consistent with the block diagram in Fig. \ref{fig:3}. Denote the joint distribution of the source and impairment terms with $\rmp_{\xx,z_0,\ldots,z_b}$. For this system,
\begin{itemize}
\item the $\kappa$th noise-error correlation is defined as
\begin{align}
\sfC_\kappa = \frac{1}{\sqrt{\lams_\kappa}} \E (\hxx-\xx) z_\kappa \label{eq:sim-1}
\end{align}
for $\kappa\in[0:b]$, and
\item the \ac{mse} is denoted by
\begin{align}
\mse = \E (\hxx-\xx)^2. \label{eq:sim-2}
\end{align}
\end{itemize}
\end{definition}

Invoking Definition \ref{def:single_user_parameters}, the $b$\ac{rsb} ansatz can be completely represented in terms of the input-output statistics of the decoupled system. In fact, by means of Definition \ref{def:single_user_parameters}, the fixed point equations in \eqref{eq:rsb-20a}-\eqref{eq:rsb-20c} can be expressed as
\begin{subequations}
\begin{align}
\sum_{\kappa=1}^b p_\kappa + q &= \mse,\label{eq:sim-3} \\
\tilde{\chi}_{\kappa-1}+\mu_\kappa \left(\sum_{\varsigma=\kappa}^b p_\varsigma + q \right) &=\lams \sfC_\kappa, \label{eq:sim-4} \\
\tilde{\chi}_{b} &= \lams \sfC_0, \label{eq:sim-5}
\end{align}
\end{subequations}
for $\kappa\in[1:b]$; moreover, the factor $\Lambda_1$ is given as
\begin{align}
\Lambda_{1} {(\{z_\varsigma\}_{1}^b, z_0,\xx)} =  e^{-\mu_1  \left[  \tfrac{1}{2 \lams}(y -\hxx)^2 - \tfrac{1}{2\lams} (y-\xx)^2 + u(\hxx) \right]}, \label{eq:sim-6}
\end{align}
which reduces to
\begin{align}
\Lambda_{1}  {(\{z_\varsigma\}_{1}^b, z_0,\xx)} = \rmp_\xx(\xx)^{\mu_1} \left[ \frac{\poster{\hxx}{y}}{\poster{\xx}{y}} \right]^{\mu_1} \label{eq:sim-6.1}
\end{align}
with $\poster{\cdot}{y}$ indicating the decoupled posterior distribution defined in Definition \ref{def:single_user_estimator}. The second term on the right hand side of \eqref{eq:sim-6.1} is an extended form of the likelihood ratio. By defining
\begin{align}
\Gamma_{1}  {(\{z_\varsigma\}_{1}^b, z_0,\xx)} =\left[ \frac{\poster{\hxx}{y}}{\poster{\xx}{y}} \right]^{\mu_1}, \label{eq:sim-6.2}
\end{align}
\eqref{eq:sim-6.1} reads
\begin{align}
\Lambda_{1}  {(\{z_\varsigma\}_{1}^b, z_0,\xx)} = \rmp_\xx(\xx)^{\mu_1} \Gamma_{1}  {(\{z_\varsigma\}_{1}^b, z_0,\xx)} , \label{eq:sim-6.3}
\end{align}
and $\Lambda_\kappa {(\{z_\varsigma\}_{\kappa}^b, z_0,\xx)}$ for $\kappa\in[2:b]$ are determined by
\begin{align}
\Lambda_{\kappa}  {(\{z_\varsigma\}_{\kappa}^b, z_0,\xx)} = \rmp_\xx(\xx)^{\mu_\kappa} \Gamma_\kappa  {(\{z_\varsigma\}_{\kappa}^b, z_0,\xx)} \label{eq:sim-6.4}
\end{align}
where $\Gamma_\kappa {(\{z_\varsigma\}_{\kappa}^b, z_0,\xx)}$ are recursively defined as
\begin{align}
\Gamma_{\kappa} {(\{z_\varsigma\}_{\kappa}^b, z_0,\xx)} = \left[ \int \Gamma_{\kappa-1}  {(\{z_\varsigma\}_{\kappa-1}^b, z_0,\xx)} \ \md z_{\kappa-1} \right]^{\tfrac{\mu_{\kappa}}{\mu_{\kappa-1}}}. \label{eq:sim-6.4}
\end{align}
The fixed point in \eqref{eq:rsb-20.1} is therefore rewritten accordingly.

The above alternative representation of the $b$\ac{rsb} ansatz leads us to a new interpretation. In fact, one can define a transition system in which the vector of replica parameters denotes the state, and the decoupled single-user system defines the transition rule \cite{hansen2003algorithms,finkel1998fundamental}. We refer to this transition system as the ``replica simulator'', and define it formally as the following.

\begin{definition}[Replica simulator]
\label{def:rep_sim}
Let $b$ be a non-negative integer. Consider the vector $\bss$ as
\begin{align}
\bss\coloneqq \left[ \chi, \mu_1, \ldots, \mu_b, p_1, \ldots, p_b,q\right] \label{eq:sim-7}
\end{align}
with entries satisfying the corresponding constraints in Proposition \ref{proposition:7}, and denote its support by $\setS_b$. The transition rule $\sfT^\sfR_b: \setS_b \mapsto \setS_b$ maps the prior state $\bss_i\in \setS_b$ to the posterior state $\bss_f\in\setS_b$ in the following way:

\textit{$\sfT^\sfR_b$ realizes the $b$\ac{rsb} decoupled system considering the entries of $\bss_i$ as the replica parameters. It then constructs the entries of $\bss_f$ by determining a new set of replica parameters from the statistics of the decoupled system using the equivalent representation of the fixed point equations given in \eqref{eq:sim-3}-\eqref{eq:sim-6.4}.}

The replica simulator of breaking order $b$ is then defined as the transition system $\rsim[b] \coloneqq\left( \setS_b, \sfT^\sfR_b \right)$.
\end{definition}

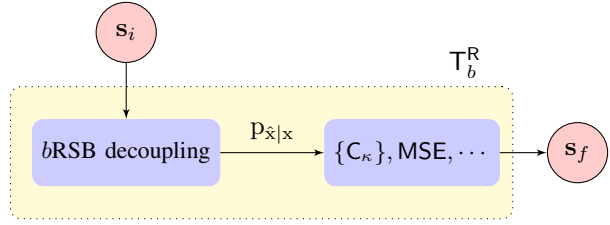
\begin{figure}[t]
\begin{center}
\begin{tikzpicture}[auto,>=latex']
    \node [state,fill=red!20] (s_i) {$\bss_i$};
    \node [input, right of=s_i, node distance=1.8cm] (help) {};
    \node [margin, below of=help, fill=yellow!20, node distance=1.6cm,minimum width=19em, minimum height=5em] (margin) {};
    \node [block, below of=s_i, fill=blue!20, blue!20,node distance=1.6cm] (sys) {\textcolor{black}{\small $b$\ac{rsb} decoupling}};
    \node [input, right of=sys, node distance=2cm] (mid) {};
    \node [block, right of=mid, node distance=1.8cm, fill=blue!20,blue!20] (stat) {\textcolor{black}{\small $\{ \sfC_\kappa \}, \mse, \cdots$}};
    \node [state, right of=stat,  node distance=2.2cm,fill=red!20] (s_f) {$\bss_f$};
    \node [input, right of=help, node distance=2.7cm] (h2) {};
    \node [draw=none,fill=none, below of=h2, node distance=.4cm] {$\sfT^\sfR_b$};

    \draw [draw,->] (s_i) -- node {} (sys) ;
    \draw [draw,->] (sys) -- node {$\rmp_{\hxx|\xx}$} (stat) ;
    \draw [draw,->] (stat) -- node {} (s_f) ;
\end{tikzpicture}
\end{center}
	\caption{Replica Simulator of breaking order $b$}
	\label{fig:4}
\end{figure}

The structure of the replica simulator is illustrated in Fig. \ref{fig:4}. For the replica simulator of breaking order $b$, a sequence of states $\left\lbrace \bss_t \right\rbrace$ is considered to be a ``process'', if for $t \in [1:\infty]$
\begin{align}
\bss_{t} \stackrel{\sfT^\sfR_b}{\longrightarrow} \bss_{t+1}. \label{eq:sim-8}
\end{align}
The state $\bss^\star$ is then called the ``steady state'', if setting $\bss_t=\bss^\star$ results in $\bss_{t+1}=\bss^\star$. Regarding Proposition \ref{proposition:7}, the $b$\ac{rsb} ansatz is in fact the steady state of the replica simulator which minimizes the free energy function. Our conclusion also extends to the \ac{rs} case, if we set $b=0$.

Considering Definition \ref{def:rep_sim}, as well as the above discussions, the $b$\ac{rsb} ansatz can be numerically investigated using the methods developed in the literature of transition systems. This approach may reduce the complexity of numerical analysis; however, it does not impact the computational complexity\footnote{We conjecture that in some cases, our $b$\ac{rsb} decoupled system represents the asymptotic of a decision-feedback system. The validity of this conjecture can further reduce the computation complexity.}. In fact, assuming that one realizes the $b$\ac{rsb} decoupled system for any desired state vector denoted in \eqref{eq:sim-7} via some methods of realization, e.g., Monte Carlo simulation, the $b$\ac{rsb} ansatz can be found by means of an iterative algorithm which has been designed to find the steady state of a transition system. The latter statement can be clarified as in Scheme \ref{scheme:1}.

\begin{algorithm}[h]
%\captionsetup{name=MegaAlgorithm}
\floatname{algorithm}{Scheme}
\algblock[Name]{Start}{End}
\algblockdefx[NAME]{START}{END}%
[2][Unknown]{Start #1(#2)}%
{Ending}
\algblockdefx[NAME]{}{OTHEREND}%
[1]{Until (#1)}
\begin{algorithmic}
\BEGIN
\Set replica simulator of breaking order $b$, $\rsim[b]$
	\If{$b=0$}
	\State $\sfT^\sfR_b$ corresponds to the \ac{rs} decoupled system
    \Else
    \State $\sfT^\sfR_b$ corresponds to the $b$\ac{rsb} decoupled system
    \EndIf

\Initiate initial state $\bss^0 \in \setS_b$
\Evaluate $\bss^0 \stackrel{\sfT^\sfR_b}{\longrightarrow} \bss$
\Comment{\textsf{A}}
	\If{$\bss=\bss^0$}
	\State $\bss^\star \gets \bss$
	\State \textbf{break}
	\Else
	\State \textbf{consider} mapping rule $\sfI\sfM(\cdot,\cdot): \setS_b\times\setS_b \mapsto \setS_b$
	\Comment{\textsf{B}}
	\State $\bss^0 \gets \sfI \sfM(\bss, \bss^0)$
	\State \textbf{return} \textsf{A}
	\EndIf
\Output $\bss^{\star}$
\ENDall
\end{algorithmic}
\caption{Analysis via Replica Simulator}
\label{scheme:1}
\end{algorithm}

In Scheme \ref{scheme:1}, $\sfT^\sfR_b$ in step \textsf{A} can be realized via different methods. One may determine the input-output distribution of the single-user system analytically or simulate the system by generating impairment and source samples numerically via Monte Carlo technique. Another degree of freedom is in step \textsf{B} where different mapping rules with different convergence speeds can be employed. For algorithms designed based on Scheme \ref{scheme:1}, the computational complexity depends on the realization method while the convergence speed is mainly restricted with some given mapping rule $\sfI\sfM(\cdot,\cdot)$.

The replica simulator introduces a systematic approach for investigating the replica an\"atze based on the decoupling principle. Moreover, it gives an intuition about the impact of symmetry breaking. To clarify the latter statement, let us consider an example. 

\textbf{Example.} (\ac{rs} vs. 1\ac{rsb} ansatz) Let $b=0$; thus, the \ac{rs} fixed point equations read
\begin{subequations}
\begin{align}
q&=\mse, \\
\chi &= \lams \sfC_0.
\end{align}
\end{subequations}
The equations under the 1\ac{rsb} assumption are moreover given by
\begin{subequations}
\begin{align}
q+p&=\mse, \\
\chi + \mu p &= \lams \sfC_0, \\
\chi + \mu p+ \mu q &= \lams \sfC_1, \label{eq:fix2}
\end{align}
\end{subequations}
and $\mu$ is determined through the fixed point equation
\begin{align}
\frac{\mu}{2\lams}\left[\mu\frac{\lams_1}{\lams}q+p\right]-\frac{1}{2\lambda}\int_0^{\mu p} \mathrm{R}_{\mJ}(-\frac{\chi+\omega}{\lambda}) \dif \omega &= \rmI(z_1;z_0,\xx ) + \sfD_{\sfK \sfL} ( \rmp_{z_1} \Vert \phi ) \label{eq:fix1}
\end{align}
where $\rmI(\cdot;\cdot )$ and $\sfD_{\sfK \sfL} ( \cdot \parallel \cdot )$ denote the mutual information and Kullback-Leibler divergence, respectively. Assuming the system matrix to be \ac{iid} and setting $z_1$ to be independent of $z_0$ and $\xx$, the right hand side of \eqref{eq:fix1} tends to zero, and therefore, the solutions $\mu=0$ and $p=0$ are concluded. Consequently, \eqref{eq:fix2} becomes ineffective, and the fixed point equations reduce to \ac{rs}. The latter observation can be interpreted in terms of the ``state evolution'' of the replica simulator. More precisely, assume that the initial state of the replica simulator with breaking order one is chosen such that in the corresponding decoupled system, $z_1$ is sufficiently correlated with the source and noise symbols. In this case, by assuming the mapping rule $\mathsf{IM}(\cdot,\cdot)$ to be converging, the correlation in each iteration of Scheme \ref{scheme:1} reduces, and thus, at the steady state, $z_1$ becomes independent of $z_0$ and $\xx$.

The above discussion can be extended to replica simulators with larger breaking orders. Moreover, further properties of the \ac{rsb} ans\"atze could be studied using methods developed in the literature of transition systems\footnote{The concept of replica simulator may further clarify the connection between \ac{rsb} ans\"atze and message passing based algorithms. Such investigations however are skipped here and can be considered as a possible future work.}. We leave the further investigations as a possible future work.

\section{Large Compressive Sensing Systems}
\label{sec:cs}
Considering the setting represented in Section \ref{sec:problem_formulation}, a large compressive sensing system can be studied through our results by restricting the source's \ac{cdf} to be of the form
\begin{align}
\rmF_x(x)=(1-\alpha)\mone\left\lbrace x\geq 0 \right\rbrace+\alpha \breve{\rmF}_x(x). \label{eq:cs-1}
\end{align}
In the large limit, the source vector distributed as \eqref{eq:cs-1} has $(1-\alpha)n$ entries equal to zero while the remaining $\alpha n$ entries are distributed with $\breve{\rmF}_x$. In this case, $\bx$ is an $\alpha n$-sparse vector, and thus, \eqref{eq:sys-1} is considered to represent a large compressive sensing system with the sensing matrix $\mA$.

Considering the prior as in \eqref{eq:cs-1}, different recovery schemes are then investigated by restricting the prior setup of the system, correspondingly. In this section, we study the asymptotics of several recovery schemes using our $b$\ac{rsb} decoupling principle for both the cases of continuous and finite alphabet sources.

\subsection{Continuous Sources}
Assuming $\setX=\setR$, \eqref{eq:cs-1} describes a continuous random variable multiplied by an $\alpha$-Bernoulli random variable. In this case, by varying the utility function $u(\cdot)$, different reconstruction schemes are considered. Here, we address the linear, LASSO and $\ell_0$ norm recovery schemes. The results can however be employed to investigate a general $\ell_p$ norm recovery scheme \cite{zheng2017does}.\\

\exmpl{ex:1}
(linear recovery scheme) The \ac{map} estimation is reduced to the linear recovery scheme when the utility function is set to be
\begin{align}
u(v)=\frac{v^2}{2}. \label{eq:cs-2}
\end{align}
In fact, in this case, the \ac{map} estimator postulates the prior distribution to be a zero-mean and unit-variance Gaussian and performs considerably inefficient when the source is sparse. Using the $b$\ac{rsb} decoupling principle, we conclude that in the large-system limit the source entry $x_j$ and the estimated entry $\hx_j$, for any $j\in[1:n]$, converge in probability to a sparse random variable $\xx$ distributed as in \eqref{eq:cs-1} and the estimated symbol $\hxx\coloneqq\mathrm{g}_{\mathsf{map}}[(y);\lams,u]$ where the decoupled system reduces to
\begin{align}
\mathrm{g}_{\mathsf{map}}[(y);\lams,u]= \frac{y}{1+\lams}, \label{eq:cs-3}
\end{align}
with $y$ being given by
\begin{align}
y=\xx+\sum_{\kappa=0}^b \sqrt{\lams_\kappa} z_\kappa, \label{eq:cs-4}
\end{align}
and the scalars $\lams$ and $\left\lbrace \lams_\kappa \right\rbrace$ for $\kappa\in [0:b]$ are defined as in Proposition \ref{proposition:7}. By letting $b=0$, the result reduces to the \ac{rs} decoupling principle reported in the literature, see \cite{rangan2012asymptotic,tulino2013support}; however, the result here holds for a wider set of sensing matrices and source distributions.\\

\exmpl{ex:2}
(LASSO recovery scheme) To study the LASSO recovery scheme, we set 
\begin{align}
u(v)=\abs{v}. \label{eq:cs-5}
\end{align}
Regarding the $b$\ac{rsb} decoupling principle, the prior distribution of the decoupled system's input $\xx$ is postulated to be ``Laplacian'' or ``double exponential'' with unit variance. This postulation results in a better performance of the recovery scheme in many cases, since the Laplacian \ac{pdf} could be a more realistic approximation of the sparse source distribution. Consequently, the decoupled system's output is found as $\hxx=\mathrm{g}_{\mathsf{map}}[(y);\lams,u]$ with
\begin{align}
\mathrm{g}_{\mathsf{map}}[(y);\lams,u]= [y-\lams \ \mathsf{sgn}(y)] \ w_{\lams}(y) \label{eq:cs-6}
\end{align}
where $y$ is denoted as in \eqref{eq:cs-4}, $w_{\lams}(\cdot)$ is the null window function with window width $\lams$ defined as
\begin{align}
w_{\lams}(y) = \left\{ \begin{array}{rl}
1 &\mbox{ $\abs{y} > \lams$} \\
0 &\mbox{ $\abs{y} \leq \lams$}
\end{array} \right. , \label{eq:cs-7}
\end{align}
and $\mathsf{sgn}(y)$ is the sign indicator. The decoupled single-user estimator in \eqref{eq:cs-7} which is often referred to as the soft-thresholding operator recovers the earlier \ac{rs} results by setting $b=0$ and the sensing matrix to be \ac{iid} \cite{rangan2012asymptotic}.\\

\exmpl{ex:3}
($\ell_0$ norm recovery scheme) The $\ell_0$ norm recovery scheme which considers
\begin{align}
u(v)= \mone \left\lbrace v\neq 0\right\rbrace \label{eq:cs-8}
\end{align}
can perform significantly better that the latter schemes when the sparsity increases. In this case, the prior distribution of the $b$\ac{rsb} decoupled system's input $\xx$ is found as the limit of
\begin{align}
\rmp_{\xx}^{\uptheta,\updelta}(\xx) = \frac{1}{2\uptheta + 2 \updelta\left(e-1\right) }\left\{ \begin{array}{rl}
e &\mbox{ $\phantom{0 \leq}\hspace*{1mm} \abs{\xx} \leq \updelta$} \\
1 &\mbox{ $\updelta <\abs{\xx} \leq \uptheta$}
\end{array} \right. , \label{eq:cs-9}
\end{align}
when $\uptheta \uparrow \infty$ and $\updelta\downarrow 0$. For finite values of $\uptheta$ and $\updelta$, $\rmp_{\xx}^{\uptheta,\updelta}$ can be considered as a sparse distribution in which non-zero symbols occur uniformly. The prior explains the better performance of the $\ell_0$ norm recovery scheme compared to the linear and LASSO schemes. For this case, the output of the decoupled single-user system reads $\hxx=\mathrm{g}_{\mathsf{map}}[(y);\lams,u]$, such that
\begin{align}
\mathrm{g}_{\mathsf{map}}[(y);\lams,u]= y w_{\vartheta}(y) \label{eq:cs-10}
\end{align}
where $y$ is denoted as in \eqref{eq:cs-4}, and $w_{\vartheta}(\cdot)$ is the null window function with $\vartheta=\sqrt{2 \lams}$. Here, $\mathrm{g}_{\mathsf{map}}[(\cdot);\lams,u]$ is the hard-thresholding operator and recovers the analysis in \cite{rangan2012asymptotic} for a wider class of settings.

The above examples have been also studied in earlier replica based studies, e.g., \cite{rangan2012asymptotic}, \cite{vehkapera2014analysis}. The given results can be analytically derived from the above expressions by considering the \ac{rs} ansatz and properly substituting the corresponding $\rmR$-transforms. We address the results of two important cases reported in the literature in following.

\subsubsection*{Special Case 1}
In \cite{rangan2012asymptotic}, the authors addressed the case of which an \ac{iid} sparse source is sampled by an \ac{iid} sensing matrix where the matrix entries are zero-mean random variables with the variance vanishing proportional to $k^{-1}$. The asymptotic performance of the estimator was then addressed when the linear, LASSO, and $\ell_0$ norm recovery schemes are employed using the \ac{rs} \ac{map} decoupling principle. The results reported in \cite{rangan2012asymptotic} can be derived directly by setting the $\mathrm{R}$-transform in Proposition \ref{proposition:4} to be 
\begin{align}
\mathrm{R}_{\mJ} (\omega) = \frac{1}{1-\sfr \omega}.
\end{align}

\subsubsection*{Special Case 2}
The results of \cite{rangan2012asymptotic} extended in \cite{vehkapera2014analysis} to a larger set of sensing matrices, and the \ac{rs} prediction of the asymptotic \ac{mse} was determined for sparse Gaussian sources. The given results can be recovered by considering the distortion function
\begin{align}
\sfd(\bhx;\bx)=\norm{\bhx-\bx}^2,
\end{align}
and the source distribution to be \eqref{eq:cs-1} with $\breve{\rmF}_x$ representing the zero-mean and unit-variance Gaussian \ac{cdf}.

\subsection{Finite Alphabet Sources}
Our result can be further employed to study the sampling problem of finite alphabet sources. Considering%Assuming the source's support $\setX$ to be of the form
\begin{align}
\setX = \left\lbrace 0, t_1, \ldots, t_{\ell-1} \right\rbrace, \label{eq:cs-11}
\end{align}
in which the symbol $0$ occurs with probability $1-\alpha$, and other $\ell-1$ outcomes are distributed due to $\rmp_t$. Consequently, the source distribution reads
\begin{align}
\rmp_x(x)=\left(1-\alpha\right) \mone \{ x=0\} + \alpha  \mone \{ x\neq 0\} \rmp_t(x) \label{eq:cs-12}
\end{align}
which can be interpreted as the multiplication of the non-zero discrete random variable $t$ distributed with $\rmp_t$ and an $\alpha$-Bernoulli random variable. For sake of brevity, we denote the sorted version of the symbols in the support $\setX$ by $\rmc_1, \ldots, \rmc_\ell$ in which%Without loss of generality, we denote the symbols in $\setX$ with $\rmc_1, \ldots, \rmc_\ell$, such that
\begin{align}
-\infty < \rmc_1 < \rmc_2 < \ldots < \rmc_\ell < + \infty. \label{eq:cs-13}
\end{align}
For notational compactness, we further define $\rmc_0$ and $\rmc_{\ell+1}$ to be $-\infty$ and $+\infty$, respectively. Similar to the continuous case, different choices of the utility function $u(\cdot)$ address different types of reconstruction schemes which we investigate in the sequel.\\

\exmpl{ex:4}
(linear recovery scheme) We consider the case in which the finite alphabet source is reconstructed via the linear recovery scheme as introduced in Example \ref{ex:1}. Using the $b$\ac{rsb} decoupling principle, the source and estimated symbols $(x_j, \hx_j)$ converge to the random variables $\xx$ and $\hxx$ for all $j\in[1:n]$ where $\xx$ is distributed with $\rmp_x$ defined in \eqref{eq:cs-12}, and $\hxx$ is found as $\hxx=\mathrm{g}_{\mathsf{map}}[(y);\lams,u]$ with
\begin{align}
\mathrm{g}_{\mathsf{map}}[(y);\lams,u]= \rmc_k \qquad \text{if} \qquad  y \in \left( v^{\ell_2}_k,v^{\ell_2}_{k+1} \right] \label{eq:cs-14}
\end{align}
for $k \in [1:\ell]$. The scalar $y$ indicates the observation symbol in the equivalent decoupled system, and the boundary point $v^{\ell_2}_k$ is defined as
\begin{align}
v^{\ell_2}_k \coloneqq \frac{ 1+\lams }{2} \left( \rmc_{k-1}+\rmc_k \right). \label{eq:cs-15} \\ \nonumber
\end{align}

\exmpl{ex:5}
(LASSO recovery scheme)  Replacing the reconstruction scheme in Example \ref{ex:4} with LASSO, the single-user estimator of the $b$\ac{rsb} decoupled system is of the form
\begin{align}
\mathrm{g}_{\mathsf{map}}[(y);\lams,u]= \rmc_k \qquad \text{if} \qquad  y \in \left( v^{\ell_1}_k,v^{\ell_1}_{k+1} \right] \label{eq:cs-16}
\end{align}
for $k \in [1:\ell]$ where the boundary point $v^{\ell_l}_k$ reads
\begin{align}
v^{\ell_1}_k \coloneqq \frac{1}{2} \left( \rmc_{k-1}+\rmc_k \right) + \lams \frac{\abs{\rmc_k}-\abs{\rmc_{k-1}}}{\rmc_k-\rmc_{k-1}}. \label{eq:cs-17} \\ \nonumber
\end{align}

\exmpl{ex:6}
($\ell_0$ norm recovery scheme)  For finite alphabet sources, the $\ell_0$ norm recovery scheme is optimal in terms of symbol error rate, since it realizes the sparse uniform distribution. In fact, for the case of which the non-zero symbols of the source are uniformly distributed, the $\ell_0$ norm utility function exactly models the source's true prior, and therefore, can be considered as the optimal scheme. Under the $\ell_0$ norm recovery scheme, the $b$\ac{rsb} decoupled system reduces to
\begin{align}
\mathrm{g}_{\mathsf{map}}[(y);\lams,u]= \rmc_k \qquad \text{if} \qquad  y \in \left( v^{\ell_0}_k,v^{\ell_0}_{k+1} \right] \label{eq:cs-18}
\end{align}
with the boundary point $v^{\ell_0}_k$ being defined as
\begin{align}
v^{\ell_0}_k \coloneqq \frac{1}{2} \left( \rmc_{k-1}+\rmc_k \right) + \lams \frac{\mone\{\rmc_k\neq 0\}-\mone\{\rmc_{k-1}\neq 0\}}{\rmc_k-\rmc_{k-1}}. \label{eq:cs-19}
\end{align}
for $k \in [1:\ell]$.

\section{Numerical Results for Large Compressive Sensing Systems}
\label{sec:numerics}
In this section, we numerically investigate the examples of large compressive sensing systems for some known setups. For this purpose, we simulate the decoupled systems by setting the source distribution and sensing matrix to a specific form and determine the expected distortion of the equivalent scalar system. We, then, discuss the validity of \ac{rs} and \ac{rsb} assumptions for these examples.

\subsection{Simulation Setups}
The settings and distortion functions being considered in the numerical investigations are as follows.
\subsubsection{Sensing Matrices}{%e sensing matrix $\mA_{k\times n}$ \textcolor{blue}{can in general be any rotationally invariant random matrix}. In our
Throughout the numerical investigations, we consider the two important cases of random ``\ac{iid}'' and ``projector'' matrices.
\begin{itemize}
\item \textbf{\ac{iid} Random Matrix:} In this case, the entries of $\mA_{k\times n}$ are supposed to be generated \ac{iid} from an arbitrary distribution $\rmp_a$. Without loss of generality, we assume that the entries are zero-mean random variables with variance $k^{-1}$. This structure is the most primary and also the most discussed case in random matrix theory. For this matrix, regardless of $\rmp_a$, it is well known that the asymptotic empirical eigenvalue \ac{cdf} of the Gramian $\mJ$ follows the Marcenko-Pastur law which states
\begin{align}
\rmF_\mJ(\lambda)=\left[ 1-\sfr^{-1}\right]^+ \mone\{\lambda>0\} + \int_{-\infty}^\lambda \frac{\sqrt{\sfr-(1-u)^2}}{2 \pi \sfr u} \dif u \label{eq:cs-20}
\end{align}
where $\left[x\right]^+$ returns $x$ when $x$ is non-negative and is zero otherwise \cite{muller2013applications,tulino2004random,couillet2011random}. Using the definition of $\mathrm{R}$-transform, it is straightforward to show that $\rmR_{\mJ} (\cdot)$ reads
\begin{align}
\rmR_{\mJ} (\omega) = \frac{1}{1-\sfr \omega}. \label{eq:cs-21}
\end{align}
\item \textbf{Projector Matrix:} Here, the only constraint on the sensing matrix is that the row vectors are orthogonal. The matrices are also referred to as the ``row orthogonal'' matrices. For the sensing matrix $\mA_{k \times n}$, we assume the case that the row vectors are normalized by the number of rows and $k\leq n$; thus, the~outer~product~$\mA \mA^\trp$~reads
\begin{align}
\mA \mA^{\trp}=\frac{n}{k} \mI_k. \label{eq:cs-22}
\end{align}
Consequently, the Gram matrix $\mJ$ takes two different eigenvalues: $\lambda=0$ with multiplicity $n-k$, and $\lambda=n k^{-1}$ with multiplicity $k$. Considering the definition of the factor $\sfr$ in \eqref{eq:eq:sys-1.3}, the asymptotic empirical \ac{cdf} of the eigenvalues read
\begin{align}
\sfF_\mJ(\lambda)=\left[ 1-\sfr^{-1} \right] \mone\{\lambda > 0 \} + \sfr^{-1} \mone \{ \lambda > \sfr \} \label{eq:cs-23}
\end{align}
which results in the $\mathrm{R}$-transform of the form
\begin{align}
\rmR_{\mJ} (\omega) = \frac{\sf \omega - 1 + \sqrt{(\sfr \omega -1)^2+4 \omega}}{2 \omega}. \label{cs:24}
\end{align}
\end{itemize}
}
\subsubsection{Source Model}{
We consider the continuous and finite alphabet sources to be distributed with ``sparse Gaussian'' and ``sparse uniform'' distributions, respectively. More precisely, we assume the entries of the continuous and finite alphabet sources to be generated from a Gaussian and uniform distribution, respectively, and multiplied by a Bernoulli random variable with probability $\alpha$ to take $1$ and $1-\alpha$ to take $0$. Moreover, we assume the nonzero outcomes of the finite alphabet source to be of the symmetric form
\begin{align}
\left\lbrace \pm a, \ldots, \pm  \kappa a \right\rbrace \label{eq:cs-source}
\end{align}
for some positive real $a$ and integer $\kappa$.
}
\subsubsection{Distortion Function}{
Regarding the continuous sources, we determine the performance of different estimators by considering the \ac{mse} as the distortion function, i.e.,
\begin{align}
\sfd(\bhx;\bx)=\norm{\bhx-\bx}^2. \label{cs:25}
\end{align}
For the finite alphabet sources, moreover, we determine the probability of the error as the measure for which
\begin{align}
\sfd(\bhx;\bx)=\sum_{i=1}^n \mone \left\lbrace\hx_i \neq x_i \right\rbrace \label{cs:26}
\end{align}
as well as the \ac{mse}. Considering either case, it is clear that the \ac{mse} obtained by any $\ell_p$ norm recovery scheme is bounded from below by the \ac{mmse} bound reported in the literature \cite{barbier2017mutual,reeves2016replica}.
}

\subsection{Numerical Results for Continuous Sources}

Considering Examples \ref{ex:1}-\ref{ex:3}, we consider the case in which a sparse Gaussian source with sparsity factor $\alpha=0.1$ is sampled via a random sensing matrix. Fig. \ref{fig:5} shows the \ac{rs} prediction of normalized \ac{mse}, defined as
\begin{align}
{\mse}^0=\frac{\mse}{\E \abs{x}^2}=\alpha^{-1} \mse, \label{cs:27}
\end{align}
\begin{figure}[t]
\centering
\resizebox{1\linewidth}{!}{
\pstool[width=.35\linewidth]{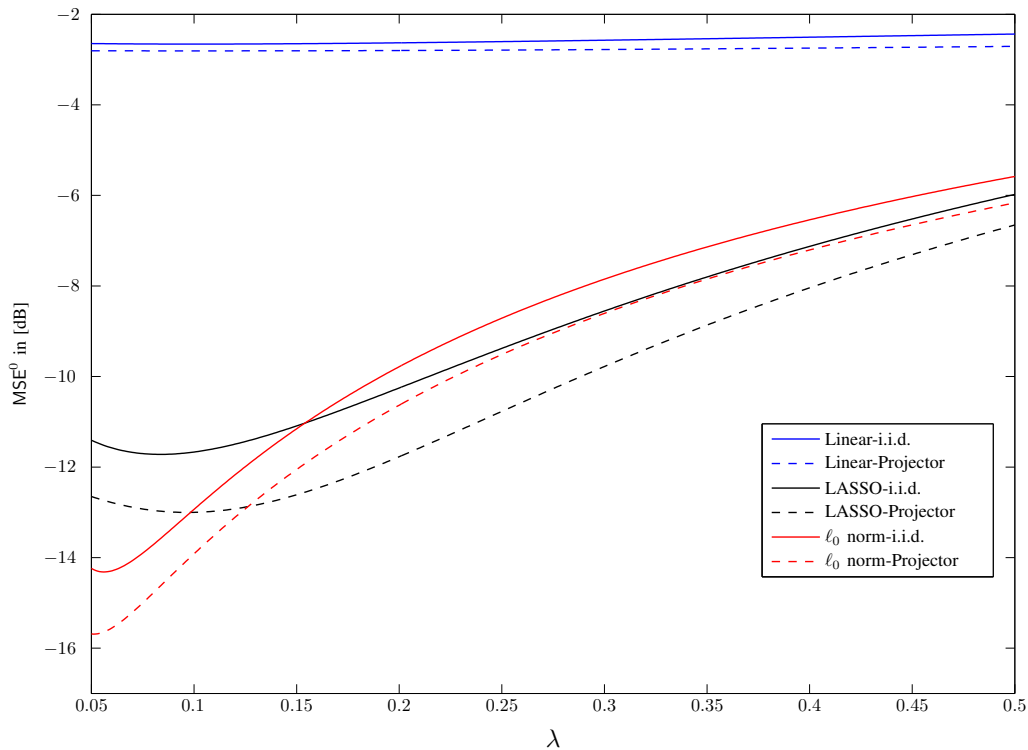}{
\psfrag{lambda}[c][c][0.35]{$\lambda$}
\psfrag{norm-mse}[c][c][0.25]{$\mse^0$ in [dB]}
\psfrag{AAABBBCCC001AA}[l][l][0.25]{Linear-\ac{iid}}
\psfrag{AAABBBCCC001BB}[l][l][0.25]{Linear-Projector}

\psfrag{AAABBBCCC002AA}[l][l][0.25]{LASSO-\ac{iid}}
\psfrag{AAABBBCCC002BB}[l][l][0.25]{LASSO-Projector}

\psfrag{AAABBBCCC003AA}[l][l][0.25]{$\ell_0$ norm-\ac{iid}}
\psfrag{AAABBBCCC003BB}[l][l][0.25]{$\ell_0$ norm-Projector}

%y-axis
\psfrag{0}[r][c][0.25]{$0$}
\psfrag{-2}[r][c][0.25]{$-2$}
\psfrag{-4}[r][c][0.25]{$-4$}
\psfrag{-6}[r][c][0.25]{$-6$}
\psfrag{-8}[r][c][0.25]{$-8$}
\psfrag{-10}[r][c][0.25]{$-10$}
\psfrag{-12}[r][c][0.25]{$-12$}
\psfrag{-14}[r][c][0.25]{$-14$}
\psfrag{-16}[r][c][0.25]{$-16$}
\psfrag{-18}[r][c][0.25]{$-18$}
\psfrag{-20}[r][c][0.25]{$-20$}

%x-axis
\psfrag{0.05}[c][c][0.25]{$0.05$}
\psfrag{0.1}[c][c][0.25]{$0.1$}
\psfrag{0.15}[c][c][0.25]{$0.15$}
\psfrag{0.2}[c][c][0.25]{$0.2$}
\psfrag{0.25}[c][c][0.25]{$0.25$}
\psfrag{0.3}[c][c][0.25]{$0.3$}
\psfrag{0.35}[c][c][0.25]{$0.35$}
\psfrag{0.4}[c][c][0.25]{$0.4$}
\psfrag{0.45}[c][c][0.25]{$0.45$}
\psfrag{0.5}[c][c][0.25]{$0.5$}

}}
\caption{\ac{rs} predicted normalized \ac{mse} versus the estimation parameter $\lambda$ for the linear, LASSO and $\ell_0$ norm recovery schemes considering the compression rate $\sfr=2$. The sparsity factor is set to be $\alpha=0.1$ and the noise variance $\lambda_0$ is set such that the source power to noise power ratio becomes $10$ dB. The dashed and solid lines respectively indicate the cases with random projector and \ac{iid} measurements. The curves match the numerical results reported in \cite{rangan2012asymptotic, vehkapera2014analysis}.}
\label{fig:5}
\end{figure}
as a function of the estimation parameter $\lambda$. The compression rate is set to be $\sfr=2$, and both the \ac{iid} random and projector sensing matrices are considered. The curves match the results reported in \cite{vehkapera2014analysis} and \cite{rangan2012asymptotic}. As it is seen, the $\ell_0$-norm recovery scheme with the optimal choice of estimation parameter outperforms the LASSO scheme; however, the non-optimal choice of the estimation parameter can make the $\ell_0$ norm's performance even worse than the LASSO. Moreover, in contrast to the noiseless case, the projector matrix is always outperforming the \ac{iid} matrix in the noisy case; this fact has been also reported in \cite{vehkapera2014analysis}.

\begin{figure}[t]
\centering
\resizebox{1\linewidth}{!}{
\pstool[width=.35\linewidth]{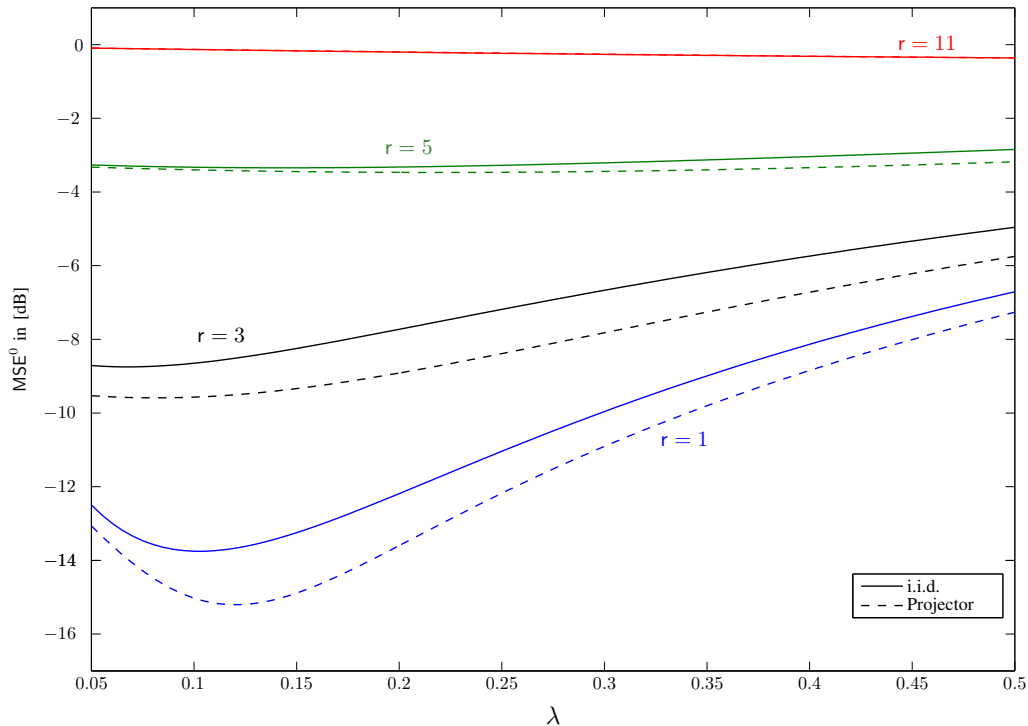}{
\psfrag{lambda}[c][c][0.35]{$\lambda$}
\psfrag{norm-mse}[c][c][0.25]{$\mse^0$ in [dB]}
\psfrag{AAABBBAA01}[l][l][0.25]{\ac{iid}}
\psfrag{AAABBBAA02}[l][l][0.25]{Projector}

%y-axis
\psfrag{0}[r][c][0.25]{$0$}
\psfrag{-2}[r][c][0.25]{$-2$}
\psfrag{-4}[r][c][0.25]{$-4$}
\psfrag{-6}[r][c][0.25]{$-6$}
\psfrag{-8}[r][c][0.25]{$-8$}
\psfrag{-10}[r][c][0.25]{$-10$}
\psfrag{-12}[r][c][0.25]{$-12$}
\psfrag{-14}[r][c][0.25]{$-14$}
\psfrag{-16}[r][c][0.25]{$-16$}
\psfrag{-18}[r][c][0.25]{$-18$}
\psfrag{-20}[r][c][0.25]{$-20$}

%x-axis
\psfrag{0.05}[c][c][0.25]{$0.05$}
\psfrag{0.1}[c][c][0.25]{$0.1$}
\psfrag{0.15}[c][c][0.25]{$0.15$}
\psfrag{0.2}[c][c][0.25]{$0.2$}
\psfrag{0.25}[c][c][0.25]{$0.25$}
\psfrag{0.3}[c][c][0.25]{$0.3$}
\psfrag{0.35}[c][c][0.25]{$0.35$}
\psfrag{0.4}[c][c][0.25]{$0.4$}
\psfrag{0.45}[c][c][0.25]{$0.45$}
\psfrag{0.5}[c][c][0.25]{$0.5$}

\psfrag{r=1}[l][c][0.3]{\textcolor{blue}{$\sfr=1$}}
\psfrag{r=3}[r][c][0.3]{\textcolor{black}{$\sfr=3$}}
\psfrag{r=5}[r][c][0.3]{\textcolor{OliveGreen}{$\sfr=5$}}
\psfrag{r=11}[r][c][0.3]{\textcolor{red}{$\sfr=11$}}

}}
\caption{\ac{rs} predicted normalized \ac{mse} versus the estimation parameter $\lambda$ for different compression rates considering LASSO recovery. The sparsity factor is set to be $\alpha=0.1$ and the true noise variance $\lambda_0$ is set to be $0.01$. The dashed and solid lines respectively indicate the random projector and \ac{iid} matrices. As the compression rate increases, $\mse^0$ converges to $0$ dB.} %In fact, for asymptotically large compression rates, the source and observation vectors become statistically independent and the \ac{mse} converges to the source power.}
\label{fig:6}
\end{figure}

In \cite{kabashima2009typical}, the authors showed that in the noiseless sampling case with an \ac{iid} Gaussian matrix, the \ac{rs} ansatz for linear and LASSO recovery schemes is locally stable against perturbations that break the symmetry of the replica correlation matrix. This result in fact agrees with the general belief that convex optimization problems do not exhibit \ac{rsb}~\cite{moustakas2007outage}. The result in \cite{kabashima2009typical}, however, indicated that for the $\ell_0$ norm reconstruction, the \ac{rs} ansatz becomes unstable, and therefore the \ac{rsb} ans\"atze are needed for accurately assessing the performance. 

In order to investigate the observation of \cite{kabashima2009typical}, we have plotted the normalized \ac{mse} of the LASSO recovery scheme predicted by the \ac{rs} ansatz in terms of the estimation parameter $\lambda$ in Fig. \ref{fig:6} considering different compression rates. It is observed that for a given estimation parameter $\lambda$, the normalized \ac{mse} increases as the compression rate grows. For large compression rates, the normalized \ac{mse} converges to $0$ dB which agrees with the fact that for asymptotically large compression rates, the source and observation vectors are independent, and thus, the \ac{mse} converges to the source power. To investigate the $\ell_0$ norm recovery scheme, we plot the corresponding curves for the $\ell_0$ norm scheme considering Fig. \ref{fig:6} as the reference.

For $\ell_0$ norm recovery, Fig. \ref{fig:7} shows the normalized \ac{mse} predicted by the \ac{rs} ansatz as a function of the estimation parameter for two different compression rates. The system setup has been set to be similar to the one considered in Fig. \ref{fig:6}, and the curves have been plotted for both the \ac{iid} and projector measurements. In contrast to the LASSO recovery scheme, the \ac{rs} ansatz starts to give invalid predictions for the $\ell_0$ norm scheme as the compression rate increases. As it is observed, the normalized \ac{mse} drops unexpectedly down for an interval of the estimation parameters when the compression rate grows.

\begin{figure}[t]
\centering
\resizebox{1\linewidth}{!}{
\pstool[width=.35\linewidth]{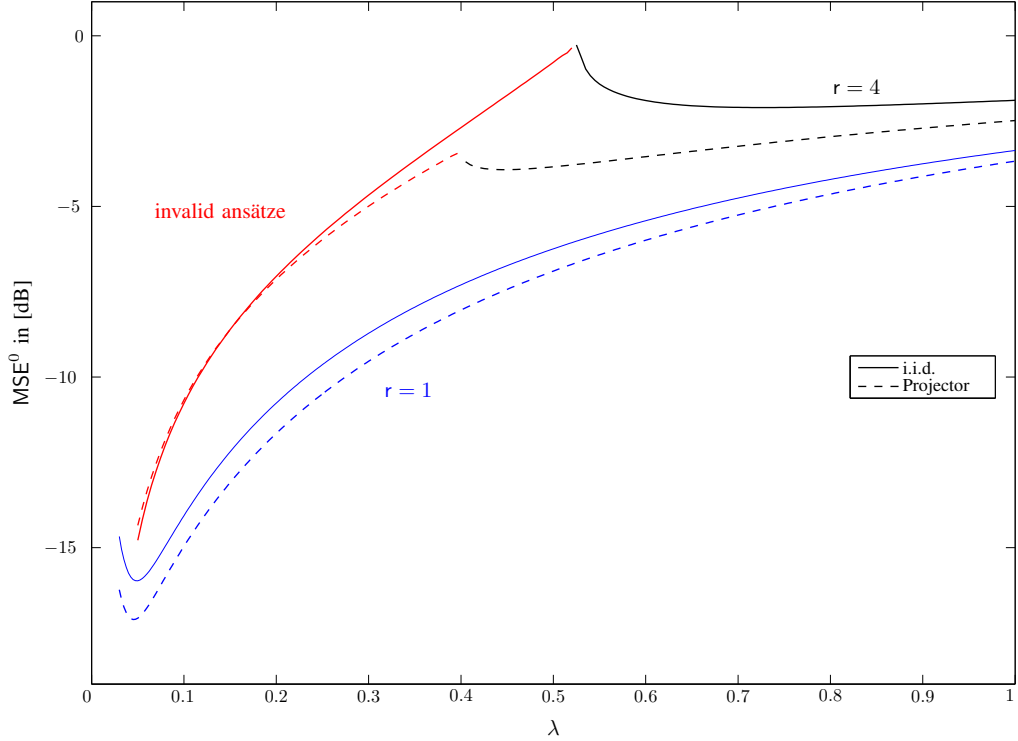}{
\psfrag{lambda}[c][c][0.35]{$\lambda$}
\psfrag{norm-mse}[c][c][0.25]{$\mse^0$ in [dB]}
\psfrag{AAABBBAA01}[l][l][0.25]{\ac{iid}}
\psfrag{AAABBBAA02}[l][l][0.25]{Projector}

%y-axis
\psfrag{0}[r][c][0.25]{$0$}
\psfrag{-5}[r][c][0.25]{$-5$}
\psfrag{5}[r][c][0.25]{$5$}
\psfrag{-10}[r][c][0.25]{$-10$}
\psfrag{-15}[r][c][0.25]{$-15$}
\psfrag{-20}[r][c][0.25]{$-20$}
\psfrag{-25}[r][c][0.25]{$-25$}

%x-axis
\psfrag{0.05}[c][c][0.25]{$0.05$}
\psfrag{0.1}[c][c][0.25]{$0.1$}
\psfrag{0.15}[c][c][0.25]{$0.15$}
\psfrag{0.2}[c][c][0.25]{$0.2$}
\psfrag{0.25}[c][c][0.25]{$0.25$}
\psfrag{0.3}[c][c][0.25]{$0.3$}
\psfrag{0.35}[c][c][0.25]{$0.35$}
\psfrag{0.4}[c][c][0.25]{$0.4$}
\psfrag{0.45}[c][c][0.25]{$0.45$}
\psfrag{0.5}[c][c][0.25]{$0.5$}
\psfrag{0.55}[c][c][0.25]{$0.55$}
\psfrag{0.6}[c][c][0.25]{$0.6$}
\psfrag{0.65}[c][c][0.25]{$0.65$}
\psfrag{0.7}[c][c][0.25]{$0.7$}
\psfrag{0.75}[c][c][0.25]{$0.75$}
\psfrag{0.8}[c][c][0.25]{$0.8$}
\psfrag{0.85}[c][c][0.25]{$0.85$}
\psfrag{0.9}[c][c][0.25]{$0.9$}
\psfrag{0.95}[c][c][0.25]{$0.95$}
\psfrag{1}[r][c][0.25]{$1$}

\psfrag{r=1}[l][c][0.3]{\textcolor{blue}{$\sfr=1$}}
\psfrag{r=4}[r][c][0.3]{\textcolor{black}{$\sfr=4$}}
\psfrag{invalid solution}[c][c][0.3]{\textcolor{red}{invalid ans\"atze}}
\psfrag{r=11}[r][c][0.3]{\textcolor{red}{$\sfr=11$}}

\psfrag{lambda}[c][c][0.3]{$\lambda$}
\psfrag{mse0}[c][c][0.3]{$\mse^0$ in [dB]}

}}
\caption{The normalized \ac{mse} as a function of the estimation parameter $\lambda$ determined via the \ac{rs} ansatz for different compression rates considering the $\ell_0$ norm recovery scheme. The sparsity factor and the noise variance are considered to be $\alpha=0.1$ and $\lambda_0=0.01$, respectively. The dashed and solid lines respectively indicate the random projector and \ac{iid} matrices. As the compression rate increases, the \ac{rs} ansatz starts to give invalid solutions at low estimation parameters. In fact, as $\sfr$ grows, the \ac{rs} fixed point equations have no valid solutions as $\lambda$ reduces. The interval in which the solution is invalid becomes larger as the compression rate increases.}
\label{fig:7}
\end{figure}

In fact, in this interval, the \ac{rs} fixed point equations have either an unstable solution or no solution. To illustrate this result further, let us consider Examples \ref{ex:2} and \ref{ex:3} under the \ac{rs} ansatz when an \ac{iid} sensing matrix is employed. In this case, the equivalent noise power and estimation parameter $\lams_0$ and $\lams$ read
\begin{subequations}
\begin{align}
\lams&=\lambda+\sfr \chi \label{eq:cs-28a} \\
\lams_0&=\lambda_0+\sfr q \label{eq:cs-28b}.
\end{align}
\end{subequations}
By increasing the compression rate, the interference increases, and thus, the \ac{mse} takes larger values. Therefore, for small $\lambda$ and $\lambda_0$, one can consider $\sfr \chi \gg \lambda$ and $\sfr q \gg \lambda_0$ as $\sfr$ takes large values and write
\begin{subequations}
\begin{align}
\lams&\approx \sfr \chi \label{eq:cs-29a} \\
\lams_0&\approx \sfr q \label{eq:cs-29b}.
\end{align}
\end{subequations}
Considering Example \ref{ex:2}, by substituting \eqref{eq:cs-29a} and \eqref{eq:cs-29b} in the \ac{rs} ansatz, the fixed point equations, as $\sfr$ grows large, is written in the following form
\begin{subequations}
\begin{align}
u{\sqrt{\sfr} \ \phi (u\sqrt{\sfr})} & \approx { \int_{u \sqrt{\sfr}}^\infty z^2 \ \md z} + \epsilon_r \label{eq:cs-30a} \\
q &\approx {2 \alpha} \int_0^{u \sqrt{\sfr}} \md z + \epsilon_r \label{eq:cs-30b}.
\end{align}
%\frac{2 \sfr^2 \chi^2}{1-\sfr \chi} \int_{u \sqrt{\sfr}}^\infty \md z +
\end{subequations}
for some $\epsilon_r$ tending to zero as $r\uparrow\infty$, and the bounded real scalar $u$ defined as
\begin{align}
\displaystyle u \coloneqq \frac{\chi}{\sqrt{q}}. \label{eq:cs-31}
\end{align}
Taking the limit $\sfr\uparrow\infty$, \eqref{eq:cs-30a} is valid for any bounded real value of $u$ and \eqref{eq:cs-30b} reduces to $q \approx \alpha$ for large compression rates. Noting that for this setup $q=\mse$, one concludes that $\mse^0 \approx 1$ which agrees with the results given in Fig. \ref{fig:6}. A similar approach for the $\ell_0$ norm recovery scheme in Example \ref{ex:3}, however, results in the following contradicting equations
\begin{subequations}
\begin{align}
\int_{u}^\infty z^2 \ \md z &\approx \epsilon_r \label{eq:cs-32a} \\
\int_{0}^u \ \md z &\approx \epsilon_r \label{eq:cs-32b}.
\end{align}
\end{subequations}
for the scalar $u$ defined as
\begin{align}
\displaystyle u \coloneqq \sqrt{2\frac{\chi}{q}}. \label{eq:cs-33}
\end{align}
Clearly, the set of equations in \eqref{eq:cs-32a} and \eqref{eq:cs-32b} have no solution as $\sfr \uparrow \infty$. The approximated fixed point equations explain the invalidity of the \ac{rs} predicted performance of the $\ell_0$ norm recovery scheme for large compression rates.

\begin{figure}[t]
\centering
\resizebox{1\linewidth}{!}{
\pstool[width=.35\linewidth]{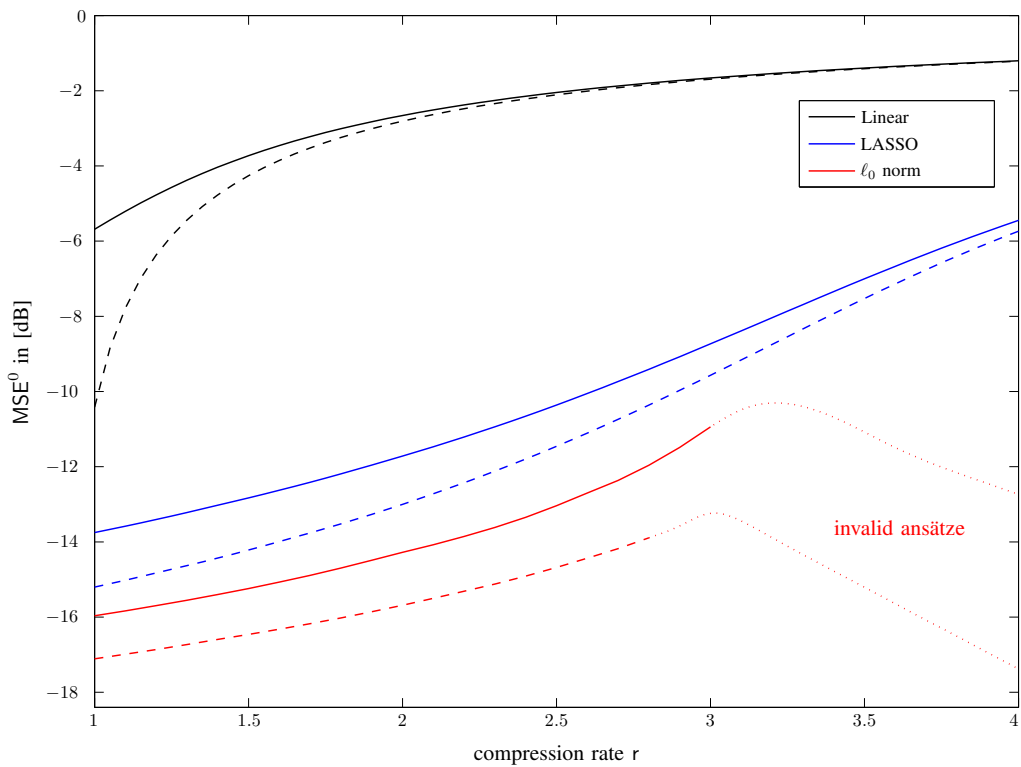}{

\psfrag{LINXXXAAA}[l][l][0.25]{Linear}
\psfrag{SINXXXAAA}[l][l][0.25]{LASSO}
\psfrag{ZINXXXAAA}[l][l][0.25]{$\ell_0$ norm}

\psfrag{invalid ansatze}[c][c][0.3]{\textcolor{red}{invalid ans\"atze}}
\psfrag{normalized mse}[c][c][0.3]{$\mse^0$ in [dB]}
%y-axis
\psfrag{0}[r][c][0.25]{$0$}
\psfrag{-2}[r][c][0.25]{$-2$}
\psfrag{-4}[r][c][0.25]{$-4$}
\psfrag{-6}[r][c][0.25]{$-6$}
\psfrag{-8}[r][c][0.25]{$-8$}
\psfrag{-10}[r][c][0.25]{$-10$}
\psfrag{-12}[r][c][0.25]{$-12$}
\psfrag{-14}[r][c][0.25]{$-14$}
\psfrag{-16}[r][c][0.25]{$-16$}
\psfrag{-18}[r][c][0.25]{$-18$}

%x-axis
\psfrag{0.5}[c][c][0.25]{$0.5$}
\psfrag{1}[c][c][0.25]{$1$}
\psfrag{1.5}[c][c][0.25]{$1.5$}
\psfrag{2}[c][c][0.25]{$2$}
\psfrag{2.5}[c][c][0.25]{$2.5$}
\psfrag{3}[c][c][0.25]{$3$}
\psfrag{3.5}[c][c][0.25]{$3.5$}
\psfrag{4}[r][c][0.25]{$4$}
\psfrag{4.5}[c][c][0.25]{$4.5$}
\psfrag{5}[c][c][0.25]{$5$}
\psfrag{5.5}[c][c][0.25]{$5.5$}
\psfrag{6}[c][c][0.25]{$6$}

\psfrag{rate}[c][c][0.3]{compression rate $\sfr$}

}}
\caption{Optimal normalized \ac{mse} versus the compression rate $\sfr$ under \ac{rs} assumption considering the linear, LASSO and $\ell_0$-norm recovery schemes. The sparsity factor and the source to noise power ratio are considered to be $\alpha=0.1$ and $10$ dB, respectively. $\mathsf{MSE}^0$ has been minimized numerically over the estimation parameter $\lambda$. The solid and dashed lines show the results for random \ac{iid} and orthogonal measurements, respectively. As $\sfr$ grows, the \ac{rs} ansatz for $\ell_0$ norm recovery starts to give invalid solutions.}
\label{fig:8}
\end{figure}

In order to further investigate the \ac{rs} ansatz, we plot the optimal normalized \ac{mse} as a function of the compression rate in Fig. \ref{fig:8}. Here, we consider the case with \ac{iid} sensing matrix when the sparsity factor is set to be $\alpha=0.1$ and source to noise power ratio to be $10$ dB. The normalized \ac{mse} is minimized numerically over the estimation parameter $\lambda$. As the figure illustrates, the \ac{mse} of the \ac{rs} ansatz starts to drop in moderate compression rates. The observation confirms the discussion on the stability of the \ac{rs} saddle point given in \cite{kabashima2009typical,yoshida2007statistical}. In fact, the unexpected drop of the \ac{rs} ansatz is caused by the limited stability region of the \ac{rs} fixed point solutions. More precisely, for a given source to noise power ratio and estimation parameter, the \ac{rs} fixed point equations have stable solutions within a certain interval of compression rates. The interval widens as $\lambda$ grows.
\begin{figure}[t]
\centering
\resizebox{1\linewidth}{!}{
\pstool[width=.35\linewidth]{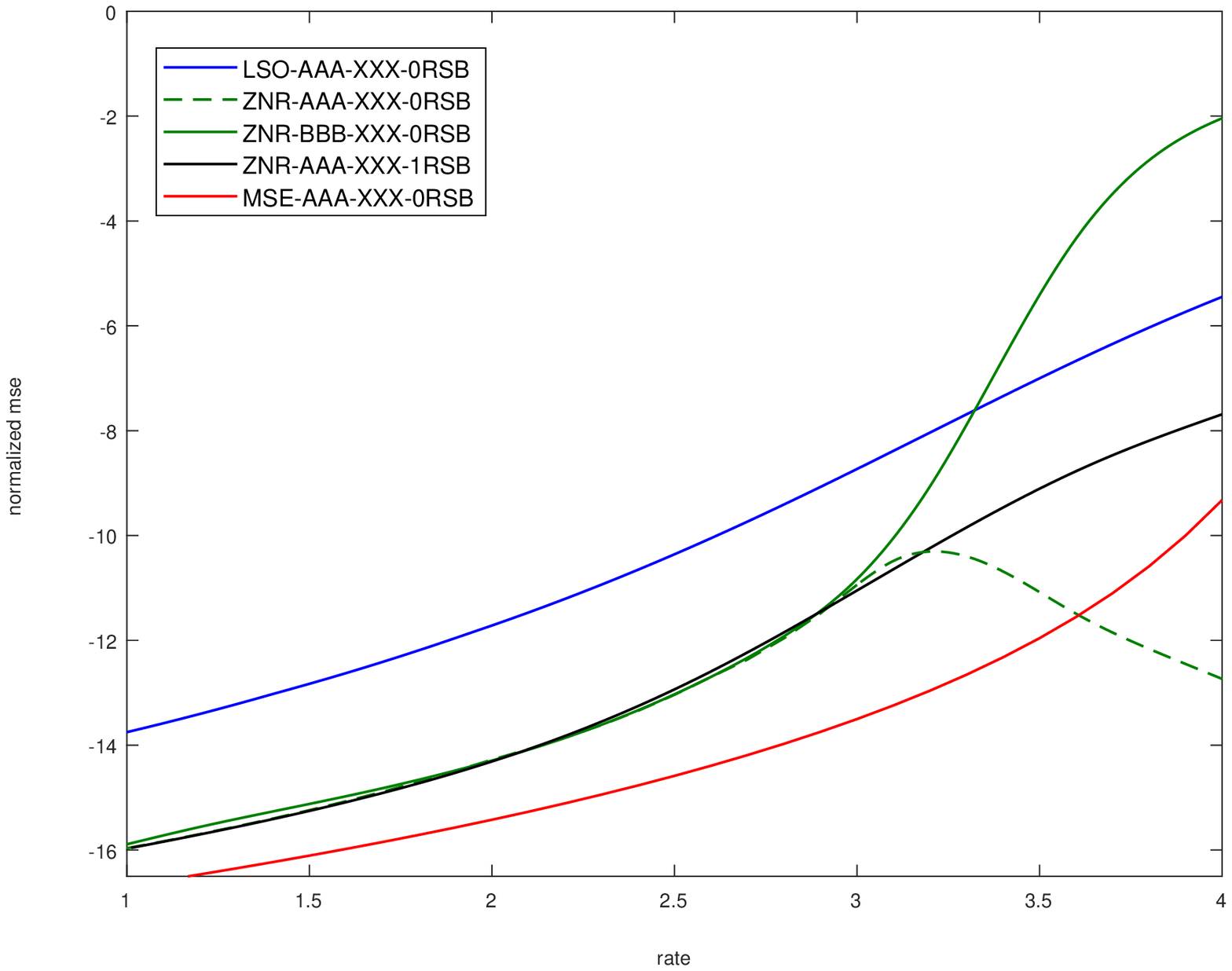}{

\psfrag{LSO-AAA-XXX-0RSB}[l][l][0.25]{LASSO-RS}
\psfrag{ZNR-AAA-XXX-0RSB}[l][l][0.25]{$\ell_0$ norm-RS}
\psfrag{ZNR-BBB-XXX-0RSB}[l][l][0.25]{$\ell_0$ norm-RS restricted}
\psfrag{ZNR-AAA-XXX-1RSB}[l][l][0.25]{$\ell_0$ norm-1RSB}
\psfrag{MSE-AAA-XXX-0RSB}[l][l][0.25]{\ac{mmse} Bound}
\psfrag{norm-mse}[c][c][0.25]{$\mse^0$ in [dB]}

%y-axis
\psfrag{0}[r][c][0.25]{$0$}
\psfrag{-2}[r][c][0.25]{$-2$}
\psfrag{-4}[r][c][0.25]{$-4$}
\psfrag{-6}[r][c][0.25]{$-6$}
\psfrag{-8}[r][c][0.25]{$-8$}
\psfrag{-10}[r][c][0.25]{$-10$}
\psfrag{-12}[r][c][0.25]{$-12$}
\psfrag{-14}[r][c][0.25]{$-14$}
\psfrag{-16}[r][c][0.25]{$-16$}

%x-axis
\psfrag{0.5}[c][c][0.25]{$0.5$}
\psfrag{1}[c][c][0.25]{$1$}
\psfrag{1.5}[c][c][0.25]{$1.5$}
\psfrag{2}[c][c][0.25]{$2$}
\psfrag{2.5}[c][c][0.25]{$2.5$}
\psfrag{3}[c][c][0.25]{$3$}
\psfrag{3.5}[c][c][0.25]{$3.5$}
\psfrag{4}[r][c][0.25]{$4$}
\psfrag{4.5}[c][c][0.25]{$4.5$}
\psfrag{5}[c][c][0.25]{$5$}
\psfrag{5.5}[c][c][0.25]{$5.5$}
\psfrag{6}[r][c][0.25]{$6$}

\psfrag{rate}[c][c][0.3]{compression rate $\sfr$}
\psfrag{normalized mse}[c][c][0.3]{$\mse^0$ in [dB]}

}}
\caption{\ac{rs} versus 1\ac{rsb} prediction of the optimal normalized \ac{mse} considering the $\ell_0$ norm recovery scheme and random \ac{iid} measurements. The sparsity factor and the source to noise power ratio are considered to be $\alpha=0.1$ and $10$ dB respectively. The green dashed line shows the \ac{rs} prediction of optimal $\mse^0$ when it is numerically minimized over $\lambda$; however, the green solid line indicates the restricted \ac{rs} predicted $\mse^0$ minimized numerically over the intervals of $\lambda$ in which the \ac{rs} ansatz is stable. The black solid line denotes the 1\ac{rsb} ansatz which deviates both the \ac{rs} and restricted \ac{rs} curves. For sake of comparison, the normalized \ac{mse} for LASSO recovery as well as the \ac{mmse} bound have been plotted. }
\label{fig:9}
\end{figure}

Fig. \ref{fig:9} compares the \ac{rs} and 1\ac{rsb} ans\"atze for $\ell_0$ norm recovery. In this figure, the optimal normalized \ac{mse} is plotted for the case with \ac{iid} sensing matrix. The sparsity factor is considered to be $\alpha=0.1$ and the noise variance is set to be $\lambda_0=0.01$. For the \ac{rs} ansatz, we have considered two cases, namely when $\mse^0$ is minimized over
\begin{inparaenum}
\item all possible estimation parameters and
\item the interval of $\lambda$'s in which the \ac{rs} ansatz is valid within\footnote{In fact, we consider $\lambda > \lambda_{\rm d}$, where $\lambda_{\rm d}$ is the point in Fig. \ref{fig:7} in which the normalized \ac{mse} curve is not differentiable.}.
\end{inparaenum}
The latter case is referred to as the ``\ac{rs} restricted'' curve. As the figure depicts, the difference between the \ac{rs} and the 1\ac{rsb} ansatz is quite small for low compression rates. The 1\ac{rsb} prediction however deviates from \ac{rs} at larger compression rates. This observation indicates that the performance analysis of the $\ell_0$ norm recovery exhibits \ac{rsb}. For sake of comparison, we also plot the curve for the LASSO recovery scheme as well as the \ac{mmse} bound. As it is observed, the normalized \ac{mse} for LASSO bounds the 1\ac{rsb} curve from above. The 1\ac{rsb} ansatz is moreover consistent with the \ac{mmse} bound which is rigorously justified in the literature \cite{barbier2017mutual,reeves2016replica}. The \ac{rs} restricted $\ell_0$ norm's curve, moreover, crosses the LASSO curve. This deviation indicates that, at large compression rates, the optimal estimation parameter $\lambda$, which minimizes the \ac{mse}, lies somewhere out of the interval in which the \ac{rs} ansatz returns a valid approximation for the \ac{mse}.

\begin{figure}[t]
\centering
\resizebox{1\linewidth}{!}{
\pstool[width=.35\linewidth]{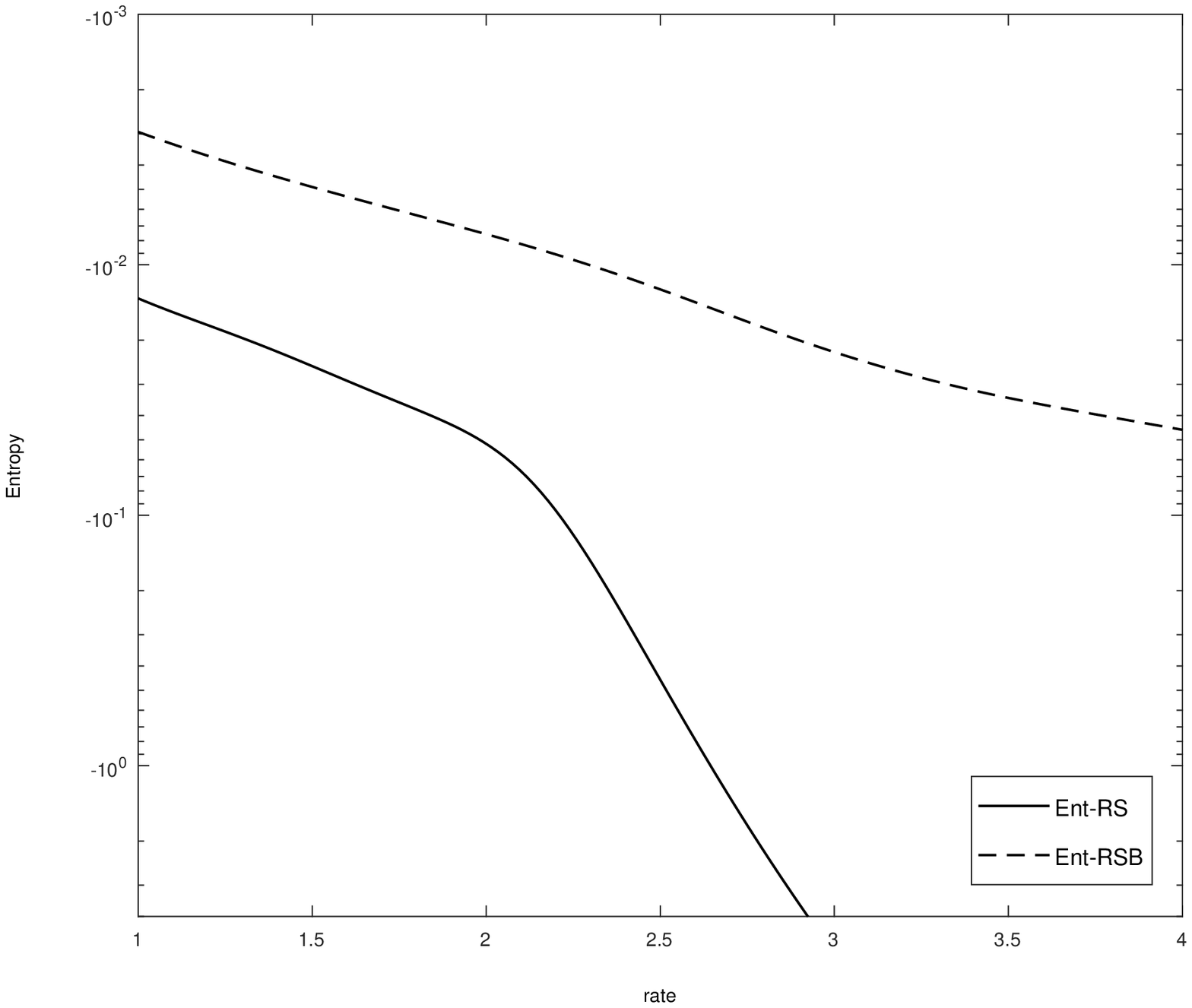}{

\psfrag{Ent-RS}[Bc][Bc][0.3]{$\rmH^0_{\mathsf{rs}}$}
\psfrag{Ent-RSB}[c][l][0.3]{$\hspace*{2mm}\rmH^{0[1]}_{\mathsf{rsb}}$}

%y-axis
\psfrag{-0.01}[c][c][0.25]{$-0.01$\hspace{3mm}}
\psfrag{-0.02}[c][c][0.25]{$-0.02$\hspace{3mm}}
\psfrag{-0.03}[c][c][0.25]{$-0.03$\hspace{3mm}}
\psfrag{-0.04}[c][c][0.25]{$-0.04$\hspace{3mm}}
\psfrag{-0.05}[c][c][0.25]{$-0.05$\hspace{3mm}}
\psfrag{-0.06}[c][c][0.25]{$-0.06$\hspace{3mm}}
\psfrag{-0.07}[c][c][0.25]{$-0.07$\hspace{3mm}}
\psfrag{-0.08}[c][c][0.25]{$-0.08$\hspace{3mm}}
\psfrag{-1}[r][c][0.21]{$-1$}
\psfrag{-2}[r][c][0.21]{$-2$}
\psfrag{-3}[r][c][0.21]{$-3$}
\psfrag{-10}[c][c][0.25]{$-10$\hspace{7mm}}

%x-axis
\psfrag{0}[r][c][0.21]{$0$\hspace*{1mm}}
\psfrag{1}[c][c][0.25]{$1$}
\psfrag{1.5}[c][c][0.25]{$1.5$}
\psfrag{2}[c][c][0.25]{$2$}
\psfrag{2.5}[c][c][0.25]{$2.5$}
\psfrag{3}[c][c][0.25]{$3$}
\psfrag{3.5}[c][c][0.25]{$3.5$}
\psfrag{4}[r][c][0.25]{$4$}
\psfrag{5}[c][c][0.25]{$5$}
\psfrag{6}[c][c][0.25]{$6$}

\psfrag{rate}[c][c][0.3]{compression rate $\sfr$}
\psfrag{Entropy}[c][c][0.3]{zero temperature entropy $\rmH^0$}

}}
\caption{$\rmH^0_{\mathsf{rs}}$ and $\rmH^{0[1]}_{\mathsf{rsb}}$ as functions of the compression rate considering the spin glass which corresponds to the $\ell_0$ norm recovery scheme. The setting is considered to be the same as in the setup investigated in Fig. \ref{fig:8} and Fig. \ref{fig:9}. The better approximation of the performance via the 1\ac{rsb} ansatz is demonstrated in this figure.}
\label{fig:10}
\end{figure}

%\begin{figure}[t]
%\centering
%\resizebox{1\linewidth}{!}{
%\pstool[width=.35\linewidth]{Figures/Fig-7}{
%
%
%%y-axis
%\psfrag{1}[r][c][0.25]{$1$}
%\psfrag{1.5}[r][r][0.25]{$1.5$}
%\psfrag{2}[r][r][0.25]{$2$}
%\psfrag{2.5}[r][r][0.25]{$2.5$}
%\psfrag{3}[r][c][0.25]{$3$}
%\psfrag{3.5}[r][c][0.25]{$3.5$}
%\psfrag{4}[r][c][0.25]{$4$}
%\psfrag{-14}[r][c][0.25]{$-14$}
%
%%x-axis
%\psfrag{5}[c][c][0.25]{$5$}
%\psfrag{10}[c][c][0.25]{$10$}
%\psfrag{15}[c][c][0.25]{$15$}
%
%\psfrag{0}[c][c][0.25]{$0$}
%\psfrag{0.01}[c][c][0.25]{$0.01$}
%\psfrag{0.02}[c][c][0.25]{$0.02$}
%\psfrag{0.03}[c][c][0.25]{$0.03$}
%\psfrag{0.04}[c][c][0.25]{$0.04$}
%\psfrag{0.05}[c][c][0.25]{$0.05$}
%\psfrag{0.06}[c][c][0.25]{$0.06$}
%\psfrag{0.07}[c][c][0.25]{$0.07$}
%\psfrag{0.08}[c][c][0.25]{$0.08$}
%\psfrag{0.09}[c][c][0.25]{$0.09$}
%\psfrag{0.1}[c][c][0.25]{$0.10$}
%
%
%\psfrag{lambda}[c][c][0.3]{$\lambda$}
%\psfrag{snr}[c][c][0.3]{$\mathsf{snr}$ in [dB]}
%\psfrag{Rbreak}[c][c][0.3]{$\sfr_{\mathsf{br}}^{\mathsf{RS}}(\mathsf{snr},\lambda)$}
%
%}}
%\caption{The break compression rate as a function of $\mathsf{snr}$ and $\lambda$ for the sparsity factor $\alpha=0.1$ and \ac{iid} sensing matrix. Within a moderate $\mathsf{snr}$ region, the break compression rate increases as $\lambda$ grows.}
%\label{fig:11}
%\end{figure}
To check the consistency of the solutions under the \ac{rs} and 1\ac{rsb} ans\"atze, we also sketch the zero temperature entropy of the corresponding spin glass as a function of the compression rate in Fig. \ref{fig:10}, considering the same setup as in Fig. \ref{fig:8} and Fig. \ref{fig:9}. The zero temperature entropy also confirms the better accuracy of the 1\ac{rsb} ansatz.

Based on the above investigations, the exact performance of the $\ell_0$ norm recovery scheme needs the \ac{rsb} ans\"atze to be accurately studied. It is however useful to see the validity region of the system parameters in which the prediction given by the \ac{rs} ansatz lies approximately on the 1\ac{rsb} curve. For this purpose, we define for the $\ell_0$ norm recovery, the ``\ac{rs} break compression rate'' $\sfr_{\mathsf{br}}^{\mathsf{RS}}(\cdot)$ as a function of system parameters to be the compression rate that the prediction via the \ac{rs} ansatz starts to deviate from 1\ac{rsb}. More precisely, for a given setting with parameters in the set $\setP$, the ``\ac{rs} $\epsilon$-validity region'' for a given $\epsilon\in \setR^+$ is defined as
\begin{align}
\setV_{\mathsf{br}}^{\mathsf{RS}} (\setP;\epsilon) \coloneqq \left\lbrace \sfr : \abs{\mse^0_{\mathsf{RS}}(\sfr,\setP)-\mse^0_{1\mathsf{RSB}}(\sfr,\setP)} < \epsilon \right\rbrace. \label{eq:cs-34}
\end{align}
where $\mse^0_{\mathsf{RS}}(\sfr,\setP)$ and $\mse^0_{1\mathsf{RSB}}(\sfr,\setP)$ indicate, respectively, the normalized \ac{mse} at the compression rate $\sfr$ for the setup specified by the set $\setP$ calculated from the \ac{rs} and 1\ac{rsb} ansatz, respectively. Consequently, the ``\ac{rs} break compression rate'' $\sfr_{\mathsf{br}}^{\mathsf{RS}}(\cdot)$ for maximum deviation $\epsilon$ is given by
\begin{align}
\sfr_{\mathsf{br}}^{\mathsf{RS}}(\setP;\epsilon) \coloneqq \max_{\setV_{\mathsf{br}}^{\mathsf{RS}} (\setP;\epsilon)} \sfr.
\end{align}
In general, the set $\setP$ includes several setup parameters such as sensing matrix and source's distribution, the sparsity factor, noise variance and estimation parameter. Considering a case with \ac{iid} sensing matrix and sparse Gaussian source with a given sparsity factor $\alpha$, the \ac{rs} break compression rate is a function of $\mathsf{snr}$ and the estimation parameter $\lambda$, i.e., $\sfr_{\mathsf{br}}^{\mathsf{RS}} (\mathsf{snr},\lambda;\epsilon)$, where we define the source to noise power ratio $\mathsf{snr}$ as
\begin{align}
\mathsf{snr} \coloneqq \frac{\E x^2}{\lambda_0} = \alpha \lambda_0^{-1}. \label{eq:cs-35}
\end{align}
In this case, the \ac{rs} $\epsilon$-validity region $\setV_{\mathsf{br}}^{\mathsf{RS}} (\mathsf{snr},\lambda;\epsilon)$ is the area enclosed by both the $\lambda$ and $\mathsf{snr}$ axes as well as the $\sfr_{\mathsf{br}}^{\mathsf{RS}} (\mathsf{snr},\lambda;\epsilon)$ surface. 

\begin{figure}[t]
\centering
\resizebox{1\linewidth}{!}{
\pstool[width=.35\linewidth]{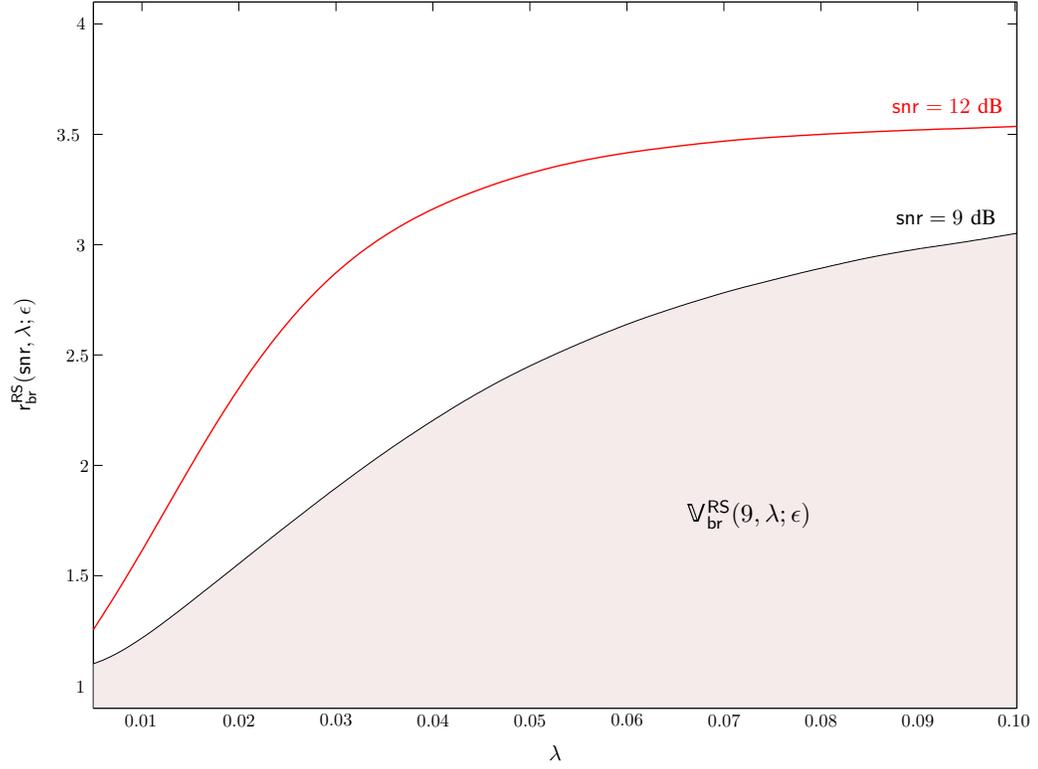}{

%y-axis
\psfrag{1}[r][c][0.25]{$1$}
\psfrag{1.5}[r][r][0.25]{$1.5$}
\psfrag{2}[r][r][0.25]{$2$}
\psfrag{2.5}[r][r][0.25]{$2.5$}
\psfrag{3}[r][c][0.25]{$3$}
\psfrag{3.5}[r][c][0.25]{$3.5$}
\psfrag{4}[r][c][0.25]{$4$}
\psfrag{-14}[r][c][0.25]{$-14$}

%x-axis
\psfrag{5}[c][c][0.25]{$5$}
\psfrag{10}[c][c][0.25]{$10$}
\psfrag{15}[c][c][0.25]{$15$}

\psfrag{0}[c][c][0.25]{$0$}
\psfrag{0.01}[c][c][0.25]{$0.01$}
\psfrag{0.02}[c][c][0.25]{$0.02$}
\psfrag{0.03}[c][c][0.25]{$0.03$}
\psfrag{0.04}[c][c][0.25]{$0.04$}
\psfrag{0.05}[c][c][0.25]{$0.05$}
\psfrag{0.06}[c][c][0.25]{$0.06$}
\psfrag{0.07}[c][c][0.25]{$0.07$}
\psfrag{0.08}[c][c][0.25]{$0.08$}
\psfrag{0.09}[c][c][0.25]{$0.09$}
\psfrag{0.1}[c][c][0.25]{$0.10$}

\psfrag{lambda}[c][c][0.3]{$\lambda$}
\psfrag{snr=12}[c][c][0.3]{\textcolor{red}{$\mathsf{snr}=12$ dB}}
\psfrag{snr=9}[c][c][0.3]{$\mathsf{snr}=9$ dB}
\psfrag{Rbreak}[c][c][0.3]{$\sfr_{\mathsf{br}}^{\mathsf{RS}}(\mathsf{snr},\lambda;\epsilon)$}
\psfrag{RS validity region}[c][c][0.35]{$\setV_{\mathsf{br}}^{\mathsf{RS}} (9,\lambda;\epsilon)$}
}}
\caption{$\sfr_{\mathsf{br}}^{\mathsf{RS}}(\mathsf{snr},\lambda;\epsilon)$ in terms of the estimation parameter $\lambda$ for $\epsilon=10^{-3}$, and $\mathsf{snr}=9$ dB and $\mathsf{snr}=12$ dB considering an \ac{iid} matrix employed for sensing. The break compression rate starts to saturate as the estimation parameter grows. As $\lambda\downarrow0$, the break compression rate takes values close to one. This observation agrees with \ac{rs} instability reported in \cite{kabashima2009typical}.}
\label{fig:12}
\end{figure}

Fig. \ref{fig:12} illustrates the intersection of the \ac{rs} $\epsilon$-validity region and the planes $\mathsf{snr}=9$ and $\mathsf{snr}=12$ for $\epsilon=10^{-3}$. The break compression rate increases with respect to both $\lambda$ and $\mathsf{snr}$ which agrees with the intuition. Moreover, $\sfr_{\mathsf{br}}^{\mathsf{RS}} (\mathsf{snr},\lambda;\epsilon)$ starts to saturate as $\lambda$ grows. Another extreme case is when the estimation parameter tends to zero in which the \ac{map} estimator reduces to
\begin{align}
\bgg(\by)= \arg \min_{\norm{\by-\mA \bv}\leq \epsilon_0} \norm{\bv}_0. \label{eq:cs-36}
\end{align}
with $\epsilon_0=\mao(\frac{1}{\sqrt{\mathsf{snr}}})$. For this case, the break compression rate converges to the minimum compression rate $\sfr=1$. This observation in the large $\mathsf{snr}$ limit agrees with the instability of the \ac{rs} ansatz reported in \cite{kabashima2009typical}.
\newpage
\subsection{Numerical Results for Finite Alphabet Sources}
Considering the finite alphabet source given in \eqref{eq:cs-source}, the boundary points for linear recovery read
\begin{align}
v_k^{\ell_2} = \pm \left(\frac{2k-1}{2} a(1 + \lams) \right)
\end{align}
for $k\in[1:\kappa]$. In the case of LASSO reconstruction, the boundary points are given by
\begin{align}
v_k^{\ell_1} = \pm \left( \frac{2k-1}{2} a + \lams \right)
\end{align}
for $k\in[1:\kappa]$. Finally, by employing $\ell_0$ norm recovery, we have $v_1^{\ell_0} = \pm \left( \frac{1}{2} a+\frac{\lams}{a} \right)$ and
\begin{figure}[t]
\centering
\resizebox{1\linewidth}{!}{
\pstool[width=.35\linewidth]{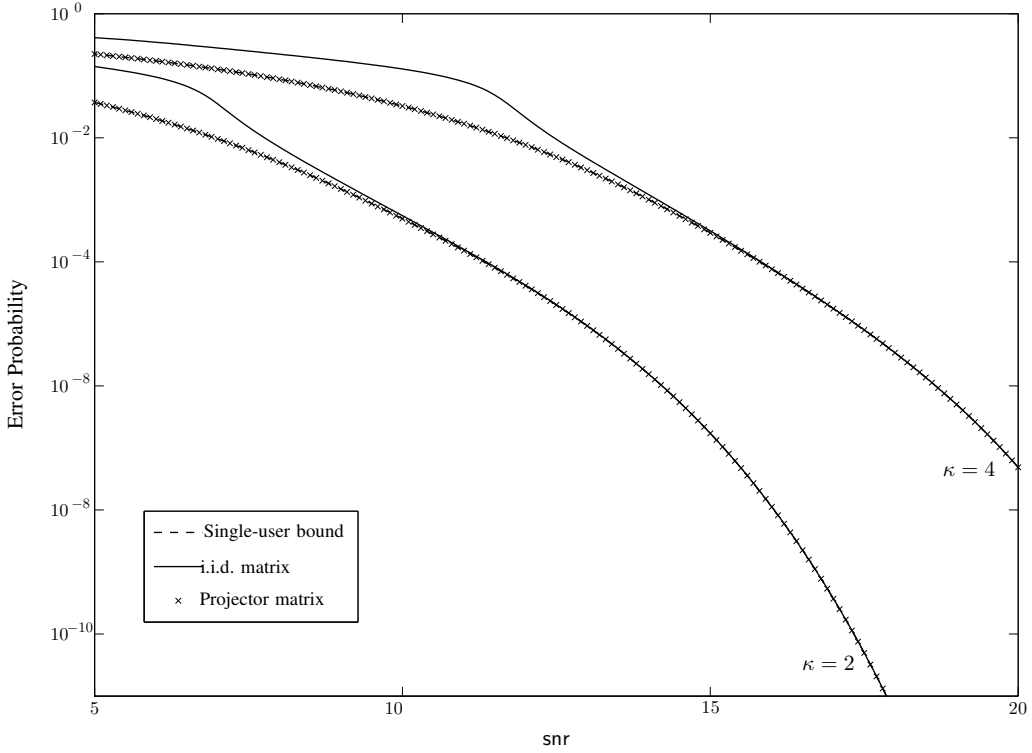}{

%y-axis

\psfrag{0}[c][c][0.2]{$0$}
\psfrag{-2}[c][c][0.2]{$-2$}
\psfrag{-4}[c][c][0.2]{$-4$}
\psfrag{-6}[c][c][0.2]{$-8$}
\psfrag{-8}[c][c][0.2]{$-8$}
\psfrag{-10}[c][c][0.2]{$-10$\hspace{1.5mm}}

%x-axis
\psfrag{5}[c][c][0.25]{$5$}
\psfrag{10}[r][c][0.25]{$10$}
\psfrag{15}[c][c][0.25]{$15$}
\psfrag{20}[c][c][0.25]{$20$}

%\psfrag{0}[c][c][0.25]{$0$}
\psfrag{0.01}[c][c][0.25]{$0.01$}
\psfrag{0.02}[c][c][0.25]{$0.02$}
\psfrag{0.03}[c][c][0.25]{$0.03$}
\psfrag{0.04}[c][c][0.25]{$0.04$}
\psfrag{0.05}[c][c][0.25]{$0.05$}
\psfrag{0.06}[c][c][0.25]{$0.06$}
\psfrag{0.07}[c][c][0.25]{$0.07$}
\psfrag{0.08}[c][c][0.25]{$0.08$}
\psfrag{0.09}[c][c][0.25]{$0.09$}
\psfrag{0.1}[c][c][0.25]{$0.10$}

\psfrag{SUB-AAA-XXX-KKK}[l][l][0.26]{\hspace*{-.85mm}Single-user bound}
\psfrag{DDD-AAA-XXX-KKK}[l][l][0.26]{\hspace*{-1.6mm}\ac{iid} matrix}
\psfrag{OOO-AAA-XXX-KKK}[l][l][0.26]{\hspace*{-1.6mm}Projector matrix}

\psfrag{snr}[c][c][0.3]{$\mathsf{snr}$}
\psfrag{K=2}[c][c][0.3]{$\kappa=2$}
\psfrag{K=4}[c][c][0.3]{$\kappa=4$}

\psfrag{Error Probability}[c][c][0.3]{Error Probability}

}}
\caption{The error probability versus $\mathsf{snr}$ when the $\ell_0$ norm recovery scheme is employed and the compression rate is set to be $\sfr=1$. By \ac{iid} measurements, the error probability deviates from the single-user bound within an interval of $\mathsf{snr}$ while random orthogonal measurements meets the single-user bound for almost every $\mathsf{snr}$ and alphabet size. Here, the sparsity factor and estimation parameter are considered to be $\alpha = 0.1$ and $\lambda = 0.1$, respectively.}
\label{fig:13}
\end{figure}
\begin{align}
v_k^{\ell_0} = \pm \frac{2k-1}{2} a 
\end{align}
for $k\in[2:\kappa]$. The source to noise power ratio $\mathsf{snr}$ is then defined as
\begin{align}
\mathsf{snr}\coloneqq \frac{\E x^2}{\lambda_0}=\frac{\alpha a^2 (\kappa+1)(2\kappa+1)}{6 \lambda_0}.
\end{align}

Fig. \ref{fig:13} shows the \ac{rs} predicted error probability of the finite alphabet system as a function of $\mathsf{snr}$ for the unit compression rate when the $\ell_0$ norm recovery scheme is employed. The cases with $\kappa=2$ and $\kappa=4$ are considered for both random \ac{iid} and orthogonal measurements. Here, the sparsity factor is set to be $\alpha=0.1$ and $\lambda=0.1$. As the figure illustrates, for either sensing matrix, the error probability meets the single-user bound in a relatively large $\mathsf{snr}$ regime. The deviation from the single-user bound in the \ac{iid} case, moreover, occurs in a larger interval of $\mathsf{snr}$ as $\kappa$ grows. In contrast to \ac{iid} measurements, the coincidence occurs for almost every $\mathsf{snr}$ and $\kappa$ when a projector sensing matrix is employed. This observation is intuitively justified due to the orthogonality of the rows in the latter case.

\begin{figure}[t]
\centering
\resizebox{1\linewidth}{!}{
\pstool[width=.35\linewidth]{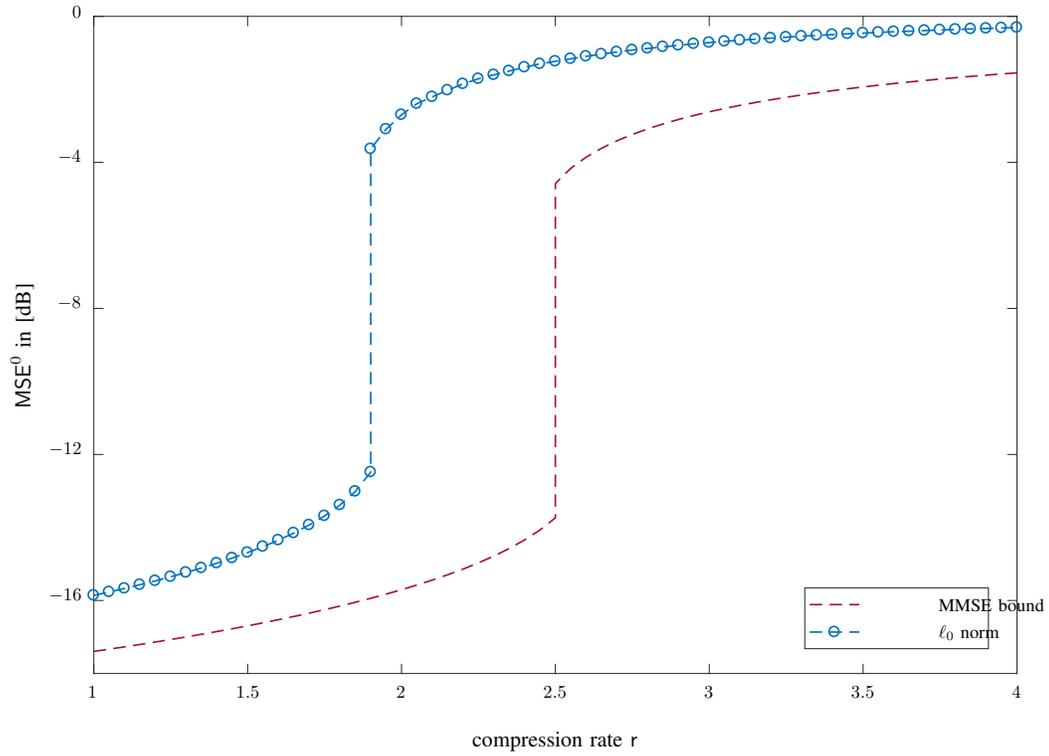}{

%y-axis

\psfrag{0}[b][c][0.25]{$0$\hspace*{2mm}}
\psfrag{-8}[b][c][0.25]{$-8$\hspace*{4mm}}
\psfrag{-4}[b][c][0.25]{$-4$\hspace*{4mm}}
\psfrag{-12}[b][c][0.25]{$-12$\hspace*{4mm}}
\psfrag{-16}[b][c][0.25]{$-16$\hspace*{4mm}}

%x-axis
\psfrag{1}[c][c][0.25]{$1$}
\psfrag{1.5}[c][c][0.25]{$1.5$}
\psfrag{2}[c][c][0.25]{$2$}
\psfrag{2.5}[c][c][0.25]{$2.5$}
\psfrag{3}[c][c][0.25]{$3$}
\psfrag{3.5}[c][c][0.25]{$3.5$}
\psfrag{4}[c][c][0.25]{$4$}
\psfrag{4.5}[c][c][0.25]{$4.5$}
\psfrag{5}[c][c][0.25]{$5$}

\psfrag{10}[r][c][0.25]{$10$}

\psfrag{MMSESNRDDDB}[c][r][0.25]{\hspace*{4mm}\ac{mmse} bound}
\psfrag{ZIMSESNRDDDB}[c][r][0.25]{\hspace*{-5mm}$\ell_0$ norm}
\psfrag{MFE-AAA-XXX-DDD}[l][l][0.25]{\ac{rs} ansatz-\ac{iid}}
\psfrag{SPO-AAA-XXX-OOO}[l][l][0.25]{FP solutions-projector}
\psfrag{MFE-AAA-XXX-OOO}[l][l][0.25]{\ac{rs} ansatz-projector}

\psfrag{iid matrix figure curve}[c][c][0.26]{\ac{iid} Matrix}
\psfrag{orthogonal matrix curve}[c][c][0.26]{Projector Matrix}

\psfrag{normMSE}[c][c][0.3]{$\mse^0$ in [dB]}

\psfrag{rate}[c][c][0.3]{compression rate $\sfr$}

\psfrag{Error Probability}[c][c][0.3]{Error Probability}

}}
\caption{Normalized \ac{mse} versus the compression rate $\sfr$ when $\mathsf{snr} =5$ dB and $a=1$. Here, $\kappa=1$ and the sparsity factor is considered to be $\alpha=0.1$. The $\ell_0$ norm recovery scheme is employed for estimation and the \ac{mse} is numerically minimized over $\lambda$ considering an \ac{iid} sensing matrix. The red curve indicates the \ac{mmse} bound on the normalized \ac{mse}. The figure demonstrates a phase transition in the normalized \ac{mse}.}
\label{fig:14}
\end{figure}

For the \ac{mse}, similar behavior as in Fig.~\ref{fig:13} is observed in Fig.~\ref{fig:14}. As the compression rate increases, considering either the \ac{mse} or error probability, a deviation from the single-user bound is observed for both random \ac{iid} and orthogonal measurements. Considering a fixed $\mathsf{snr}$, it is observed that the \ac{mse}-compression rate has a discontinuity point in which the \ac{mse} suddenly jumps from a lower value to an upper value at a certain compression rate. This phenomenon is known as the ``phase transition'' in which the macroscopic parameters suddenly change the phase within a minor variation of the setting. The compression rate in which the phase transition occurs is referred to as ``transition rate'' which increases as $\mathsf{snr}$ grows. Phase transitions were reported in the literature of communications and information theory for several problems such as turbo coding and \ac{cdma} systems \cite{tanaka2002statistical,agrawal2001turbo,tanaka2002statistical}.

Fig. \ref{fig:14} illustrates the phase transition under the \ac{rs} assumption for the case of sparse binary source, i.e., $\kappa=1$, when the source vector is sensed via a random \ac{iid} matrix and $\mathsf{snr}=5$ dB with $a=1$. For estimation, the $\ell_0$ norm recovery scheme is employed which is equivalent to LASSO in this particular example. The estimation parameter $\lambda$ is moreover optimized numerically such that the \ac{mse} is minimized. To determine each point on the curve, the \ac{rs} fixed point equations have been solved for the specific compression rate and all possible solutions have been found. Within this set of solutions, those which are physically stable have been considered and the stable solution with minimum free energy has been taken for the \ac{rs} ansatz. The normalized \ac{mse} at each compression rate has then been plotted using the \ac{rs} ansatz. As the figure shows, the normalized \ac{mse} suddenly jumps to an upper bound which converges to $\mse^0=0$ dB. For sake of comparison, we have plotted the \ac{mmse} bound as well which shows a samilar behavior at higher rates. In Fig. \ref{fig:15}, we have further plotted the normalized \ac{mse} in terms of the compression rate for the same setting with the sensing matrix replaced by a random projector matrix. The figure shows a phase transition for orthogonal measurements at a higher transition rate compared to the case with random \ac{iid} sensing. Moreover, the improvement in terms of \ac{mse} observed in Fig.~\ref{fig:15} by employing a projector matrix extends the conclusion given for the sparse Gaussian sources in \cite{vehkapera2014analysis} to the sparse finite alphabet sources as well. In order to obtain a more accurate approximation of the performance and transition point, one may investigate the free energy and entropy of the corresponding spin glass by also considering \ac{rsb} ans\"atze \cite{yoshida2007statistical}. The \ac{rs} ansatz, however, gives an accurate approximation of the \ac{mse} for the specific setup considered here.
\begin{figure}[t]
\centering
\resizebox{1\linewidth}{!}{
\pstool[width=.35\linewidth]{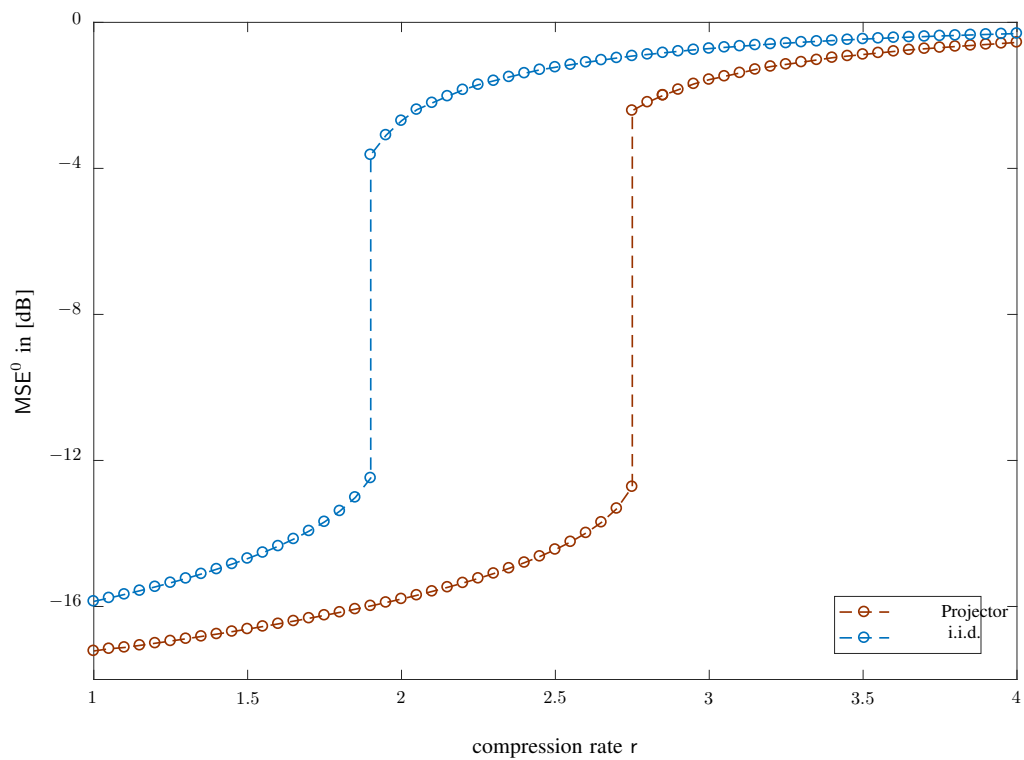}{

%y-axis

\psfrag{0}[b][c][0.25]{$0$\hspace*{2mm}}
\psfrag{-8}[b][c][0.25]{$-8$\hspace*{4mm}}
\psfrag{-4}[b][c][0.25]{$-4$\hspace*{4mm}}
\psfrag{-12}[b][c][0.25]{$-12$\hspace*{4mm}}
\psfrag{-16}[b][c][0.25]{$-16$\hspace*{4mm}}

%x-axis
\psfrag{1}[c][c][0.25]{$1$}
\psfrag{1.5}[c][c][0.25]{$1.5$}
\psfrag{2}[c][c][0.25]{$2$}
\psfrag{2.5}[c][c][0.25]{$2.5$}
\psfrag{3}[c][c][0.25]{$3$}
\psfrag{3.5}[c][c][0.25]{$3.5$}
\psfrag{4}[c][c][0.25]{$4$}
\psfrag{4.5}[c][c][0.25]{$4.5$}
\psfrag{5}[c][c][0.25]{$5$}

\psfrag{10}[r][c][0.25]{$10$}

\psfrag{MSESNR5A}[c][r][0.25]{Projector}
\psfrag{MSESNR5D}[b][r][0.25]{\hspace*{-4mm}\ac{iid}}
\psfrag{MFE-AAA-XXX-DDD}[l][l][0.25]{\ac{rs} ansatz-\ac{iid}}
\psfrag{SPO-AAA-XXX-OOO}[l][l][0.25]{FP solutions-projector}
\psfrag{MFE-AAA-XXX-OOO}[l][l][0.25]{\ac{rs} ansatz-projector}

\psfrag{iid matrix figure curve}[c][c][0.26]{\ac{iid} Matrix}
\psfrag{orthogonal matrix curve}[c][c][0.26]{Projector Matrix}

\psfrag{normMSE}[c][c][0.3]{$\mse^0$ in [dB]}

\psfrag{rate}[c][c][0.3]{compression rate $\sfr$}

\psfrag{Error Probability}[c][c][0.3]{Error Probability}

}}
\caption{$\mse^0$ in terms of the compression rate $\sfr$ for $\mathsf{snr} =5$ dB and $a=1$ considering \ac{iid} and projector sensing matrices. Here, $\alpha=0.1$ and $\kappa=1$, and the \ac{mse} is numerically minimized with respect to the estimation parameter. As the figure shows, the projector matrix enhances the reconstruction performance with respect to the both \ac{mse} and transition rate.}
\label{fig:15}
\end{figure}

\section{Conclusion}
\label{sec:conclusion}
This manuscript considered the performance of the \ac{map} estimator in the large-system limit with respect to a general distortion function. Taking a statistical mechanical approach, the replica method was employed to determine the asymptotic distortion of the system. We deviated from the earlier approaches, e.g., \cite{rangan2012asymptotic, tulino2013support,vehkapera2014analysis}, by evaluating the general replica ansatz of the corresponding spin glass. The general ansatz let us derive the \ac{rs}, as well as the $b$\ac{rsb} ansatz considering the class of rotationally invariant random matrices. The results recover the previous studies \cite{rangan2012asymptotic,tulino2013support,vehkapera2014analysis} in special cases and justifies the uniqueness of the zero temperature entropy's expression under the $b$\ac{rsb} assumption conjectured in \cite{zaidel2012vector}. The replica ansatz evaluated here led us to a more general form of the decoupling principle. In fact, invoking the general replica ansatz, it was shown that for any tuple of input-output entries, the marginal joint distribution converges to the input-output empirical distribution asymptotically. This means that in the large-system limit, the marginal joint distribution of the entries determined by any replica ansatz decouples into a set of identical joint distributions. The form of the asymptotic decoupled distribution, however, depends on the structure imposed on the ansatz. For the $b$\ac{rsb} ansatz, the vector-valued \ac{awgn} system estimated by a \ac{map} estimator decouples into single-user \ac{map}-estimated \ac{awgn} channels followed by some correlated impairment~terms which intuitively model the interference of the system and vanish by reducing the $b$\ac{rsb} assumption to \ac{rs}. The general decoupling principle justified here confirms the conjecture that decoupling is a generic property of \ac{map} estimators, since its validity relies only on the replica continuity assumption. Recent results in statistical mechanics have shown that failures in finding the exact solution via the replica method are mainly caused by the structure imposed on the ansatz, and not replica continuity \cite{talagrand2003spin,guerra2002quadratic,guerra2002thermodynamic,guerra2003broken}. The decoupling property enabled us to represent the equivalent ``replica simulator'' interpretation of the replica ansatz. In fact, the $b$\ac{rsb} fixed point equations are completely described by the statistics of the corresponding decoupled system. The $b$\ac{rsb} ansatz is therefore represented through the state evolution of a transition system which takes an initial set of ansatz parameters as the input and determines a new set of parameters via simulating the $b$\ac{rsb} decoupled system as the output.

As an example of vector-valued \ac{map}-estimated systems, we considered the noisy compressive sensing system and studied different sparse recovery schemes, including the linear, LASSO, and $\ell_0$ norm recovery schemes. The numerical investigations showed that for sparse Gaussian sources, the performance of the linear and LASSO recovery schemes are accurately approximated by the \ac{rs} ansatz within a large interval of compression rates. The \ac{rs} prediction of the $\ell_0$ norm scheme's performance, however, was observed to lack the stability within the moderate and high compression rates. As a result, the 1\ac{rsb} ansatz for this scheme was considered which deviates from the \ac{rs} ansatz significantly as the rate grows. For sparse finite alphabet sources, the \ac{rs} prediction of the error probability and \ac{mse} reported phase transitions with respect to the compression rate. The numerical results, moreover, showed a better performance of the random orthogonal measurements compared to the random \ac{iid} sensing in both the cases which were previously reported for sparse Gaussian sources in \cite{vehkapera2014analysis}. 

The current work can be pursued in several directions. As an example, the ``replica simulator'' introduced in Section \ref{sec:rep_sim} can be studied further by methods developed in the context of transition systems. The analysis may result in proposing a new framework which simplifies the evaluation of fixed point equations. Another direction is the analysis of the conditional distribution of the correlated impairment terms in the $b$\ac{rsb} decoupled system. The study can lead us to understand the necessary or sufficient conditions in which the \ac{rs} ansatz gives an accurate approximation of the system's performance. Inspired by the Sherrington-Kirkpatrick model of spin glasses, for which the full \ac{rsb} ansatz, i.e., $b\uparrow\infty$, has been proved to give a stable solution for all system parameters \cite{mezard2009information}, our conjecture is that the exact solution at all large compression rates, for the cases exhibiting \ac{rsb}, is given by a large number of breaking steps. Nevertheless, as our numerical investigations depicted, even in those cases, the \ac{rs} ansatz, or the $b$\ac{rsb} ansatz with small $b$, can give good approximations of the solution up to some moderate compression rates. The latter study, furthermore, gives us some insights about the accuracy of the approximation given by a finite number of \ac{rsb} steps. Investigating the connection between the replica simulator and message passing based algorithms is another interesting topic for future work.

\section{Acknowledgment}
The authors would like to acknowledge German Research Foundation, Deutsche Forschungsgemeinschaft (DFG), for support of this work under the Grant No. MU 3735/2-1. %Thanks also to Benjamin Zaidel, Mohammad Ali~Sedaghat, Mikko Verkaper\"a and Toshiyuki Tanaka for interesting discussions.

\clearpage
\appendices
\section{Proof of Proposition \ref{proposition:1}}
\label{app:a}
Starting from \eqref{eq:sm-8}, we define the moment function $\sfZ(m)$ as in \eqref{eq:sm-9a}. Therefore,
\begin{align}
\sfD^{\setW}(\bhx;\bx)=\lim_{\upbeta\uparrow\infty}\lim_{n\uparrow\infty}\lim_{h \downarrow 0} \lim_{m \downarrow 0} \frac{1}{m} \frac{1}{n} \frac{\partial}{\partial h} \log \sfZ(m). \label{eq:a-1}
\end{align}
Taking the expectation with respect to the noise term first, the moment function is extended as
\begin{subequations}
\begin{align}
\sfZ(m)&=\sfE_{\bx} \sfE_{\mA} \sfE_{\bz} \sum_{\{\bv_a\}}   e^{-\upbeta\sum_{a=1}^m\mae(\bv_a|\by,\mA)+h n\sum_{a=1}^m\sfD^{\setW(n)} (\bv_a;\bx)} \label{eq:a-2a}\\
&=\sfE_{\bx} \sfE_{\mA} \sfE_{\bz} \sum_{\{\bv_a\}}   e^{-\upbeta \sum_{a=1}^m\left\lbrace\frac{1}{2\lambda}(\bx-\bv_a)^\trp \mJ (\bx-\bv_a) + \frac{1}{\lambda}\bz^{\trp} \mA (\bx-\bv_a ) + \frac{1}{2\lambda}\norm{\bz}^2 + u(\bv_a)\right\rbrace+h n\sum_{a=1}^m\sfD^{\setW(n)} (\bv_a;\bx)} \label{eq:a-2b}\\
&\stackrel{\star}{=} \left(\frac{\lambda}{\lambda+m\upbeta\lambda_0}\right)^{\frac{k}{2}} \sfE_{\bx} \sfE_{\mA} \sum_{\{\bv_a\}}   e^{-\upbeta \left\lbrace \sum_{a,b=1}^m \btv_a^\trp \mJ \btv_b \zeta_{ab} +\sum_{a=1}^m  u(\bv_a) \right\rbrace +h n\sum_{a=1}^m\sfD^{\setW(n)} (\bv_a;\bx)} \label{eq:a-2c}
\end{align}
\end{subequations}
where $\btv_a=\bx-\bv_a$ for $a\in[1:m]$ is the unbiased\footnote{We call it unbiased since it expresses the deviation of the replicas from the source vector.} replica vector, $\mJ$ is the Gramian of the system matrix defined as $\mJ\coloneqq \mA^\trp \mA$ and satisfies the properties stated in Section \ref{sec:problem_formulation}, $\star$ comes from taking expectation over $\bz$, and the factor $\zeta_{ab}$ is defined as
\begin{align}
\zeta_{ab} \coloneqq \frac{1}{2\lambda} \left[ \boldsymbol{1} \{ a=b \}-  \frac{\lambda_0}{\lambda+m\upbeta \lambda_0} \upbeta \right] \label{eq:zeta}
\end{align}
with $\lambda_0$ being the true noise variance as specified in Section \ref{sec:problem_formulation}.
\begin{remark}
Considering \eqref{eq:a-2b}, one could drop the term $\frac{1}{2\lambda} \norm{\bz}^2$, since it plays no rule in the optimization problem \eqref{eq:int-2}. More precisely, one could redefine the Hamiltonian in \eqref{eq:int-7} to be
\begin{align}
\mae_{\mathrm{new}}(\bv|\by,\mA)=\mae(\bv|\by,\mA)-\frac{1}{2\lambda} \norm{\bz}^2
\end{align}
without loss of generality. In this case, the coefficient at the right hand side of \eqref{eq:a-2c} reduces to $1$, and $\zeta_{ab}$ reads
\begin{align}
\zeta_{ab}^{\mathrm{new}} \coloneqq \frac{1}{2\lambda} \left[ \boldsymbol{1} \{ a=b \}-  \frac{\lambda_0}{\lambda} \upbeta \right].
\end{align}
It is, however, clear from \eqref{eq:a-2c} and \eqref{eq:zeta} that as $m$ tends to zero, both the approaches result in a same result.
\end{remark}
Considering the expression in \eqref{eq:a-2c}, we define the random variable $\sfZ(m;\bx)$ as
\begin{align}
\sfZ(m;\bx)\coloneqq \sfE_{\mJ} \sum_{\{\bv_a\}}   e^{-\upbeta \left\lbrace \sum_{a,b=1}^m \btv_a^\trp \mJ \btv_b \zeta_{ab} +\sum_{a=1}^m  u(\bv_a) \right\rbrace +h n\sum_{a=1}^m\sfD^{\setW(n)} (\bv_a;\bx)}. \label{eq:a-3}
\end{align}
Consequently, the moment function is given by taking the expectation over $\sfZ(m;\bx)$ with respect to $\bx$ and multiplying with the scalar $\left(\frac{\lambda}{\lambda+m\upbeta\lambda_0}\right)^{\frac{k}{2}}$. However, we later on show that for almost all realizations of the source vector, $\sfZ(m;\bx)$ converges to a deterministic value, and therefore, the expectation with respect to $\bx$ can be dropped. For $\sfZ(m;\bx)$, we have
\begin{subequations}
\begin{align}
\sfZ(m;\bx) &= \E_{\mJ} \sum_{\{\bv_a\}}  e^{- \upbeta \sum_{a,b=1}^m \btv_a^{\trp} \mJ \btv_b \zeta_{ab}} \times e^{- \upbeta  \sum_{a=1}^m u(\bv_a) +h n\sum_{a=1}^m\sfD^{\setW(n)}(\bv_a;\bx)} \label{eq:a-4a}  \\
&= \sum_{\{\bv_a\}} \left[ \E_{\mJ} e^{- n \upbeta \tr{ \mJ \mG}} \right] \times e^{-\upbeta\sum_{a=1}^m u(\bv_a) +h n\sum_{a=1}^m\sfD^{\setW(n)}(\bv_a;\bx)} \label{eq:a-4b}
\end{align}
\end{subequations}
where $\mG_{n \times n}$ is defined as
\begin{align}
\mG\coloneqq \frac{1}{n} \sum_{a,b=1}^m \btv_b \btv_a^{\trp} \zeta_{a,b}. \label{eq:a-5}
\end{align}
Considering the eigendecomposition of the Gramian $\mJ=\mU \mD \mU^{\trp}$, the expectation in \eqref{eq:a-4b} can be expressed in terms of a spherical integral where the integral measure is the probability measure of $\mU$. Regarding the system setup specified in Section \ref{sec:problem_formulation}, the matrix $\mU$ is distributed over the orthogonal group $\setO_n$ with Haar probability distribution. Therefore, the corresponding spherical integral is the so-called ``Harish-Chandra'' or ``Itzykson \& Zuber'' integral. This integral has been extensively studied in the physics and mathematics literature, see for example \cite{harish1957differential}, \cite{itzykson1980planar} and \cite{guionnet2002large}. A brief discussion on the spherical integral and its closed form solution has been given in Appendix~\ref{app:f}. Invoking Theorem 1.2 and 1.7 in \cite{guionnet2005fourier}, as long as $\mathrm{rank}(\mG)=\mao(\sqrt{n})$\footnote{$\mao(\cdot)$ indicates the growth order with respect to composition, i.e., $\lim\limits_{n \uparrow \infty}[f(n)]^{-1}\mao(f(n))=\mathsf{K}<\infty$.}, the expectation in \eqref{eq:a-4b} can be written as
\begin{align}
\E_{\mJ} e^{- n \upbeta \tr{ \mJ \mG}} = e^{- n \left[\sum_{i=1}^n \int_{0}^{\upbeta \lambda^{\mG}_i} \rmR_{\mJ}(-2 \omega) \dif \omega \right] + \epsilon_n} \label{eq:a-6}
\end{align}with $\mG$ being defined in \eqref{eq:a-5}, and $\epsilon_n \downarrow 0$ as $n \uparrow \infty$. In order to employ the above result and substitute it in \eqref{eq:a-4b}, we need to check the rank condition.

\begin{lemma}
\label{lem:a-1}
Considering $\mG$ to be defined as in \eqref{eq:a-5}, the following argument holds.
\begin{align}
\mathrm{rank}(\mG)=\mao(\sqrt{n}).
\end{align}
\end{lemma}
\begin{proof}
First, we rewrite $\mG$ as
\begin{subequations}
\begin{align}
\mG&= \frac{1}{2\lambda n} \left[ \sum_{a=1}^m \btv_a \btv_a^{\trp} - \upbeta \frac{\lambda_0}{\lambda+m\upbeta \lambda_0} (\sum_{a=1}^m \btv_a) (\sum_{b=1}^m \btv_b)^{\trp} \right] \label{eq:a-7a}\\
&=\frac{1}{2\lambda n}  \mtV (\mI_m- \upbeta \frac{\lambda_0}{\lambda+m\upbeta \lambda_0}  \mathbf{1}_m) \mtV^\trp  \label{eq:a-7b}
\end{align}
\end{subequations}
where $\mtV = [ \btv_1, \ldots , \btv_m]$ is an $n \times m$ matrix with the columns being the unbiased replicas. Then, by considering \eqref{eq:a-7b}, it is obvious that $\mG$ could be, at most, of rank $m$. As Assumption \ref{asp:2} indicates, $\sfZ(m)$ analytically continues to the real axis, and the limit with respect to $m$ is taken in a right neighborhood of $0$. Therefore, for all values of $n$ there exists a constant $\mathsf{K} \in \setR^+$, such that $m \leq \mathsf{K}$. Consequently, one can write
\begin{align}
\lim_{n \uparrow \infty}\frac{\mathrm{rank}(\mG)}{\sqrt{n}} \leq \lim_{n \uparrow \infty} \frac{m}{\sqrt{n}} \leq \lim_{n \uparrow \infty} \frac{\mathsf{K}}{\sqrt{n}}=0 \label{eq:a-8}
\end{align}
which concludes that $\mathrm{rank}(\mG)=\mao(\sqrt{n})$.
\end{proof}

Lemma \ref{lem:a-1} ensures that \eqref{eq:a-6} always holds; therefore, noting the fact that $\mG$ has only $m$ non-zero eigenvalues, the expectation in the right hand side of \eqref{eq:a-4b} reduces to
\begin{align}
\E_{\mJ} e^{- n \upbeta \tr{ \mJ \mG}}  = e^{-n \mg (\mT \mQ^{\mathrm{v}})+\epsilon_n} \label{eq:a-9}
\end{align}
where the function $\mg(\cdot)$ is defined as
\begin{align}
\mg(\mM) \coloneqq  \int_{0}^{{\upbeta}} \Tr \{\mM \mathrm{R}_{\mJ}(-2\omega\mM)\} \dif \omega \label{eq:a-10}
\end{align}
for some square matrix $\mM$, $\mT$ is an $m \times m$ deterministic matrix given by
\begin{align}
\mT \coloneqq \frac{1}{2\lambda} \left[ \mI_m- \frac{\lambda_0}{\lambda+m\upbeta \lambda_0} \upbeta \mathbf{1}_m \right], \label{eq:a-11}
\end{align}
and $\mQ^{\mathrm{v}}$ is the $m \times m$ ``correlation matrix'' defined as
\begin{align}
&\mQ^{\mathrm{v}} = \frac{1}{n} \mtV^{\trp} \mtV. \label{eq:a-12}
\end{align}

\begin{remark}
Note that although $\mQ^{\mathrm{v}}$ is symmetric, $\mT\mQ^{\mathrm{v}}$ is not a symmetric matrix, in general; however, due to the symmetry of $\mG$, the eigenvalues of $\mT\mQ^{\mathrm{v}}$ are real, and therefore, the sequence of integrals over the real axis in \eqref{eq:a-6} exists for all indices.
\end{remark}

By substituting \eqref{eq:a-9} in \eqref{eq:a-4b}, $\sfZ(m;\bx)$ is given as
\begin{align}
\sfZ(m;\bx) = \sum_{ \{ \bv_a \} } e^{-n\mg (\mT \mQ^{\mathrm{v}})- \upbeta \sum_{a=1}^m u(\bv_a) + h n\sum_{a=1}^m\sfD^{\setW(n)}(\bv_a;\bx) + \epsilon_n} . \label{eq:a-13}
\end{align}
In order to determine the sum in \eqref{eq:a-13}, we follow the technique which has been employed in \cite{muller2008vector} and \cite{zaidel2012vector}. We split the space of all replicas into subshells defined by the correlation matrices in which all the vectors of replicas in each subshell have a same correlation matrix. More precisely, for a given source vector $\bx$, the subshell of the matrix $\mQ_{m \times m}$ is defined as
\begin{align}
\setS(\mQ) = \{ \bv_1, \ldots , \bv_m | (\bx-\bv_a)^{\trp} (\bx-\bv_b) =nq_{ab} \} \label{eq:a-14}
\end{align}with $q_{ab} =[\mQ]_{ab}$ denoting the entry $(a,b)$ of $\mQ$. The sum in \eqref{eq:a-13} is determined first over each subshell, and then, over all the subshells as the following.
\begin{subequations}
\begin{align}
\sfZ(m;\bx) &= \sum_{ \{ \bv_a \} } \left[\int e^{-n\mg(\mT\mQ)} \delta(\mQ^{\mathrm{v}}-\mQ) \dif \mQ\right] e^{-\upbeta\sum_{a=1}^m u(\bv_a) + h n\sum_{a=1}^m\sfD^{\setW(n)}(\bv_a;\bx) + \epsilon_n} \label{eq:a-15a}\\
&= \int e^{-n\mg (\mT\mQ)+ \epsilon_n} \left[ \sum_{ \{ \bv_a \} }  \delta(\mQ^{\mathrm{v}}-\mQ) e^{-\upbeta\sum_{a=1}^m u(\bv_a) + h n\sum_{a=1}^m\sfD^{\setW(n)}(\bv_a;\bx) } \right] \dif \mQ \label{eq:a-15b} \\
&= \int e^{-n\left[ \mg (\mT\mQ)- \mai(\mQ)\right] +\epsilon_n} \dif \mQ \label{eq:a-15c}
\end{align}
\end{subequations}
where $\dif \mQ \coloneqq \prod_{a,b=1}^{m} \dif q_{ab}$, the integral is taken over $\setR^{m\times m}$,
\begin{align}
\delta(\mQ^{\mathrm{v}}-\mQ) \coloneqq \prod_{a,b=1}^m \delta ( \btv_a^{\trp} \btv_b - nq_{ab}),  \label{eq:a-16}
\end{align}
and the term $e^{n \mai(\mQ)}$ which determines the probability weight of the subshell $\setS(\mQ)$ is defined as
\begin{align}
e^{n \mai(\mQ)} &\coloneqq \sum_{ \{ \bv_a \} }  \delta(\mQ^{\mathrm{v}}-\mQ) e^{- \upbeta \sum_{a=1}^m u(\bv_a) + h n\sum_{a=1}^m\sfD^{\setW(n)}(\bv_a;\bx)}. \label{eq:a-17}
\end{align}

\begin{remark}
One may define the subshells over the transferred correlation matrix $\mT\mQ$ instead of correlation matrix $\mQ$. In this case the subshells over $\mQ$ defined in \eqref{eq:a-14} only rotate in the $m$-dimensional space with respect to $\mT$. The rotation, however, does not have any impact on the analysis.
\end{remark}

The last step is to determine $e^{n \mai(\mQ)}$. To do so, we represent the Dirac impulse function using its inverse Laplace transform. By defining $s_{ab}$ as the complex frequency corresponding to $\delta (\btv_a^{\trp} \btv_b - nq_{ab})$,
\begin{align}
\delta ( \btv_a^{\trp} \btv_b - nq_{ab}) &= \int e^{s_{ab} (\btv_a^{\trp} \btv_b - nq_{ab})} \frac{\dif s_{ab}}{2 \pi \mathrm{j}} \label{eq:a-18}
\end{align}
where the integral is taken over the imaginary axis $\setJ=(t-\mathrm{j} \infty, t+\mathrm{j} \infty)$, for some $t \in \setR$. Consequently, by defining the frequency domain correlation matrix $\mS$ to be an $m \times m$ matrix with $[\mS]_{ab}=s_{ab}$, \eqref{eq:a-16} reads
\begin{subequations}
\begin{align}
\delta (\mQ^\vv -\mQ)&= \int e^{\sum_{a,b=1}^m s_{ab} (\btv_a^{\trp} \btv_b - nq_{ab})} \dif \mS  \label{eq:a-19a}\\
&= \int \left[ e^{-n\tr{\mS^{\trp} \mQ}} \right]  e^{\sum_{a,b=1}^m s_{ab} \btv_a^\trp \btv_b}  \dif \mS \label{eq:a-19b}
\end{align}
\end{subequations}
with $\dif \mS$ being defined as $\displaystyle \dif \mS \coloneqq \prod_{a,b=1}^{m} \frac{\dif s_{ab}}{2 \pi \mathrm{j}}$, and the integral being taken over $\setJ^{m\times m}$. Substituting in \eqref{eq:a-17}, $e^{n \mai(\mQ)}$ reduces to
\begin{subequations}
\begin{align}
e^{n \mai(\mQ)} &= \int e^{-n\tr{\mS^\trp \mQ}} \sum_{ \{ \bv_a \} } e^{\sum_{a,b=1}^m s_{ab}\btv_a^{\trp} \btv_b - \upbeta\sum_{a=1}^m u(\bv_a) +  h n\sum_{a=1}^m\sfD^{\setW(n)}(\bv_a;\bx)} \dif \mS \label{eq:a-20a} \\
&= \int e^{-n\tr{\mS^\trp \mQ}} \ e^{n \mam(\mS)}  \dif \mS \label{eq:a-20b}
\end{align}
\end{subequations}
with $\mam(\mS)$ being defined as
\begin{align}
\mam(\mS) = \frac{1}{n} \log \sum_{ \{ \bv_a \} } e^{\sum_{a,b=1}^m s_{ab}\btv_a^{\trp} \btv_b - \upbeta\sum_{a=1}^m u(\bv_a) +  h n\sum_{a=1}^m\sfD^{\setW(n)}(\bv_a;\bx)}. \label{eq:a-21}
\end{align}
Thus, $\sfZ(m;\bx)$ finally reads
\begin{align}
\sfZ(m;\bx) = \int \int e^{-n \{ \mg(\mT \mQ)+\tr{\mS^\trp \mQ}- \mam(\mS)\}+\epsilon_n}  \dif \mS \dif \mQ. \label{eq:a-22}
\end{align}

Consequently, one needs to evaluate the expectation of $\sfZ(m;\bx)$ with respect to $\bx$, in order to determine the moment function $\sfZ(m)$. However, $\sfZ(m;\bx)$ for almost all realizations of $\bx$ converges to a deterministic asymptotic as $n \uparrow \infty$, and thus, the expectation with respect to $\bx$ can be dropped. To show the latter statement, we note that the only term in \eqref{eq:a-22} which depends on $\bx$ is $\mam(\mS)$. Therefore, it is sufficient to study the convergence of $\mam(\mS)$. Lemma \ref{lem:a-2} justifies this property of $\mam(\mS)$ using the law of large numbers, and the decoupling property of the functions $u(\cdot)$ and $\sfd(\cdot;\cdot)$.

\begin{lemma}
\label{lem:a-2}
Consider the system specified in Section \ref{sec:problem_formulation}, and let Assumption \ref{asp:2} hold. Then, as $n \uparrow \infty$, $\mam(\mS)$ defined in \eqref{eq:a-21} is given by
\begin{align}
\mam(\mS)=\E \left\lbrace (1-\eta) \left[ \log \sum_{\bvv} e^{(\bxx-\bvv)^\trp \mS (\bxx-\bvv)-\upbeta u(\bvv)} \right]
+\eta \left[ \log \sum_{\bvv} e^{(\bxx-\bvv)^\trp \mS (\bxx-\bvv)-\upbeta u(\bvv) + h \eta^{-1} \sfd(\bvv; \bxx)} \right] \right\rbrace
 \label{eq:a-23}
\end{align}where $\bvv_{m\times 1} \in \setX^m$, $\bxx_{m\times 1}$ is a vector with all the elements being the random variable $x$ which is distributed with the source distribution $\rmp_x$, the expectation is taken over $\rmp_x$, and $\sfd(\bvv; \bxx)$ is defined as $\sfd(\bvv; \bxx)\coloneqq\sum_{a=1}^m \sfd(\vv_a; \mathrm{x}_a)$.
\end{lemma}
\begin{proof}
Consider the decoupling property of the functions $u(\cdot)$ and $\sfd(\cdot;\cdot)$. Define the vector $\bvv_{m\times 1}$ over the support $\setX^m$, and the coefficients $\{ \mathsf{w}_i \}$ for $i \in [1:n]$ as 
\begin{equation}
\mathsf{w}_i=
\begin{cases}
0 & \text{if}\ i\notin\setW(n) \\
\abs{\setW(n)}^{-1} & \text{if}\ i\in\setW(n).
\end{cases} \label{eq:a-24}
\end{equation}
Then, $\mam(\mS)$ reads
\begin{subequations}
\begin{align}
\mam(\mS) &= \frac{1}{n} \log \sum_{ \{ \bv_a \} } \prod_{i=1}^n e^{\sum_{a,b=1}^m  s_{ab} (x_i-v_{ai}) (x_i-v_{bi}) - \upbeta \sum_{a=1}^m  u(v_{ai}) + h n \sum_{a=1}^m \mathsf{w}_i \sfd(v_{ai}; x_{i})} \label{eq:a-25a} \\
&= \frac{1}{n} \log \prod_{i=1}^n \sum_{\bvv} e^{\sum_{a,b=1}^m  s_{ab} (x_{i}-\vv_a) ( x_{i}-\vv_b ) - \upbeta \sum_{a=1}^m  u(\vv_{a}) + h n \sum_{a=1}^m \mathsf{w}_i \sfd(\vv_{a}; x_{i})} \label{eq:a-25b} \\
&= \frac{1}{n} \left[ \sum_{i\notin \setW} \mam_0(\mS;x_{i}) + \sum_{i\in \setW} \mam_1(\mS;x_{i}) \right] \label{eq:a-25c}
\end{align}
\end{subequations}
where the functions $\mam_0(\cdot;\cdot)$ and $\mam_1(\cdot;\cdot)$ are defined as
\begin{subequations}
\begin{align}
\mam_0(\mS;x_i) &= \log \sum_{\bvv} e^{(\bxx-\bvv)^\trp \mS (\bxx-\bvv)-\upbeta u(\bvv)} \label{eq:a-26a}\\
\mam_1(\mS;x_i) &= \log \sum_{\bvv} e^{(\bxx-\bvv)^\trp \mS (\bxx-\bvv)-\upbeta u(\bvv) + h \tfrac{n}{\abs{\setW(n)}} \sfd(\bvv; \bxx)} \label{eq:a-26b}
\end{align}
\end{subequations}
where $\bxx_{m \times 1}$ is a vector with all the elements being $x_i$, and we define $\sfd(\bvv; \bxx)\coloneqq\sum_{a=1}^m \sfd(\vv_a; \mathrm{x}_a)$ for compactness. As Assumption \ref{asp:2} suggests the limits with respect to $m$ and $n$ can be exchanged in \eqref{eq:a-1}; thus, one can consider the asymptotics of $\mam(\mS)$ for a given $m$ when $n$ tends to its large limit. Regarding the fact that $\bx$ is collected from an \ac{iid} distribution, the term in the right hand side of \eqref{eq:a-25c} converges to the expectation over the distribution $\rmp_x$ due to the law of large numbers; more precisely, as $n\uparrow\infty$
\begin{subequations}
\begin{align}
\frac{1}{n} \sum_{i\notin \setW} \mam_0(\mS;x_{i}) &= \left[1-\frac{\abs{\setW(n)}}{n}\right] \frac{1}{n-\abs{\setW(n)}} \sum_{i\notin \setW} \mam_0(\mS;x_{i}) \longrightarrow (1-\eta) \ \sfE_{x} \hspace{.3mm} \mam_0(\mS;x)  \label{eq:a-27a}\\
\frac{1}{n} \sum_{i\in \setW} \mam_1(\mS;x_{i}) &= \left[\frac{\abs{\setW(n)}}{n}\right] \frac{1}{\abs{\setW(n)}} \sum_{i\in \setW} \mam_1(\mS;x_{i}) \longrightarrow \eta \ \sfE_{x} \hspace{.3mm} \mam_1(\mS;x)  \label{eq:a-27b}
\end{align}
\end{subequations}
with $\eta$ being defined as in \eqref{eq:sys-7}. Substituting \eqref{eq:a-27a} and \eqref{eq:a-27b} in \eqref{eq:a-25c}, Lemma \ref{lem:a-2} is concluded.
\end{proof}

\begin{remark}
Considering Lemma \ref{lem:a-2}, it eventually says that the probability weight $e^{n \mai(\mQ)}$ for a given correlation matrix $\mQ$ converges to a deterministic weight as $n$ tends to infinity. This statement equivalently states that for almost any given realization of the source vector, the correlation matrix converges to its expectation. In fact, considering the correlation matrix $\mQ^\vv$, as defined in \eqref{eq:a-12}, the entries are functions of $\bx$, and therefore, variate randomly due to the source distribution for a given $n$. Lemma \ref{lem:a-2}, however, indicates that, as $n\uparrow \infty$, the entries converge to some deterministic asymptotics for almost any realization of $\bx$. As an alternative approach, one could study the convergence property of the correlation matrix $\mQ^\vv$ by means of the law of large numbers first, and then, conclude Lemma \ref{lem:a-2} by rewriting $\mam(\mS)$ in proper way, and replacing it with the expectation using the fact that the probability weight $e^{n \mai(\mQ)}$ needs to converge deterministically as $n\uparrow\infty$. Nevertheless, the approach taken here seems to be more straightforward.
\end{remark}

Using Lemma \ref{lem:a-2}, we drop the expectation with respect to $\bx$ in \eqref{eq:a-2c}. Replacing in \eqref{eq:a-1}, the asymptotic distortion is found by taking the limits. As Assumption \ref{asp:2} suggests, we exchange the order of the limits and take the limit with respect to $n$ at first. Denoting that the probability measure defined with $e^{n \mai(\mQ)} \dif \mQ$ satisfy the large deviation properties \cite{dembo2009large}, we can use the saddle point approximation to evaluate the integral in \eqref{eq:a-22} which says that as $n \uparrow \infty$
\begin{align}
\sfZ(m) = \left(\frac{\lambda}{\lambda+m\upbeta\lambda_0}\right)^{\frac{k}{2}} \int \int e^{-n \{ \mg(\mT \mQ)+\tr{\mS^\trp \mQ}- \mam(\mS)\}}  \dif \mS \dif \mQ \doteq \mathsf{K}_n e^{-n \{ \mg(\mT \tilde{\mQ})+\tr{\tilde{\mS}^\trp \tilde{\mQ}}- \mam(\tilde{\mS})\}}, \label{eq:a-28}
\end{align}
where we drop $\epsilon_n$ given in \eqref{eq:a-22} regarding the fact that it vanishes in the large limit. Here, $(\tilde{\mQ},\tilde{\mS})$ is the saddle point of the integrand function's exponent, $\mathsf{K}_n$ is a bounded coefficient, and $\doteq$ indicates the asymptotic equivalency in exponential scale defined as the following.
\begin{definition}
\normalfont
The functions $a(\cdot)$ and $b(\cdot)$ defined over the non-bounded set $\setX$ are said to be asymptotically equivalent in exponential scale, if
\begin{align}
\lim_{n  \uparrow\infty} \log |\frac{a(x_n)}{b(x_n)}| =0. \label{eq:a-29}
\end{align}
for an unbounded sequence $\{x_n \in \setX \}$.
\end{definition}
As $n \uparrow \infty$, the $m$th moment can be replaced with its asymptotic equivalent in \eqref{eq:a-1}. Consequently, by substituting the equivalent term and exchanging the limits' order, we have
\begin{subequations}
\begin{align}
\sfD^{\setW}(\bhx;\bx)&=\lim_{\upbeta\uparrow\infty}\lim_{m\downarrow 0}\lim_{h \downarrow 0} \lim_{n\uparrow\infty} \frac{1}{m}  \frac{\partial}{\partial h} \left[ -\mg(\mT \tilde{\mQ})-\tr{\tilde{\mS}^\trp \tilde{\mQ}}+\mam(\tilde{\mS}) + \frac{\log \mathsf{K}_n}{n} \right] \label{eq:a-30a}\\
&\stackrel{\star}{=}\lim_{\upbeta\uparrow\infty}\lim_{m\downarrow 0}\lim_{h \downarrow 0} \frac{1}{m}  \frac{\partial}{\partial h} \mam(\tilde{\mS}) \label{eq:a-30b}\\
&= \lim_{\upbeta\uparrow\infty}\lim_{m\downarrow 0} \E \frac{\sum_{\bvv} \sfd(\bvv; \bxx) e^{(\bxx-\bvv)^\trp \tilde{\mS} (\bxx-\bvv)-\upbeta u(\bvv)}}{m \sum_{\bvv} e^{(\bxx-\bvv)^\trp \tilde{\mS} (\bxx-\bvv)-\upbeta u(\bvv)}} \label{eq:a-30c}
\end{align}
\end{subequations}
where $\star$ comes from the fact that $\mathsf{K}_n$ is bounded, and $\mg(\mT \tilde{\mQ})$ and $\tr{\tilde{\mS}^\trp \tilde{\mQ}}$ are not functions of $h$.

The saddle point $(\tilde{\mQ},\tilde{\mS})$ is found by letting the derivatives of the exponent zero. Using the standard definition $\displaystyle \left[ \frac{\partial}{\partial \mM} \right]_{ab} \coloneqq \frac{\partial}{\partial [\mM]_{ab}}$, the saddle point is given by the following fixed point equations.
\begin{subequations}
\begin{align}
\frac{\partial}{\partial \mQ} \left[ \mg (\mT \mQ) + \tr{\mS \mQ} - \mam(\mS) \right]|_{(\tilde{\mQ},\tilde{\mS})}&=0 \label{eq:a-31a} \\
\frac{\partial}{\partial \mS} \left[ \mg (\mT \mQ) + \tr{\mS \mQ} - \mam(\mS) \right]|_{(\tilde{\mQ},\tilde{\mS})}&=0. \label{eq:a-31b}
\end{align}
\end{subequations}
\eqref{eq:a-31a} reduces to 
\begin{align}
\tilde{\mS}= -\upbeta \mT \rmR_{\mJ}(-2\upbeta \mT \tilde{\mQ}), \label{eq:a-32}
\end{align}
and \eqref{eq:a-31b} results in
\begin{align}
\tilde{\mQ}= \E \frac{\sum_{\bvv}(\bxx - \bvv)(\bxx-\bvv)^{\trp} e^{(\bxx-\bvv)^{\trp} \tilde{\mS} (\bxx-\bvv)-\upbeta u(\bvv)}}{\sum_{\bvv} e^{(\bxx-\bvv)^{\trp} \tilde{\mS} (\bxx-\bvv)-\upbeta u(\bvv)}}. \label{eq:a-33}
\end{align}
By replacing \eqref{eq:a-32} in \eqref{eq:a-30c} and \eqref{eq:a-33}, the expression for the asymptotic distortion and the saddle point correlation matrix can be considered as expectations over a conditional Boltzmann-Gibbs distribution $\rmp^\upbeta_{\bvv|\bxx}$ defined as
\begin{align}
\rmp^\upbeta_{\bvv|\bxx}(\bvv|\bxx)\coloneqq \frac{e^{-\upbeta \left[(\bxx-\bvv)^{\trp} \mT \rmR_{\mJ}(-2\upbeta \mT \tilde{\mQ}) (\bxx-\bvv)+ u(\bvv)\right]}}{\sum_{\bvv} e^{-\upbeta\left[(\bxx-\bvv)^{\trp} \mT \rmR_{\mJ}(-2\upbeta \mT \tilde{\mQ}) (\bxx-\bvv)+ u(\bvv)\right]}} \label{eq:a-34}
\end{align}
which simplifies the expressions in \eqref{eq:a-30c} and \eqref{eq:a-33} to those given in Proposition \ref{proposition:1}.

In general the fixed point equation \eqref{eq:a-33} can be satisfied with several saddle points, and therefore, multiple asymptotic distortions might be found. In this case, one should note that the valid solution is the one which minimizes the free energy of the spin glasses at the zero temperature, i.e., $\upbeta\uparrow\infty$. Using the $m$th moment, the free energy of the system reads
\begin{subequations}
\begin{align}
\sfF(\upbeta) &= -\lim_{n\uparrow\infty}\lim_{h \downarrow 0} \lim_{m \downarrow 0} \frac{1}{m} \frac{1}{\upbeta} \frac{1}{n} \log \sfZ(m) \label{eq:a-35a}\\
&\stackrel{\star}{=} \lim_{m\downarrow 0}\lim_{h \downarrow 0} \frac{1}{\upbeta m}  \left[ \mg(\mT \tilde{\mQ})+\tr{\tilde{\mS}^\trp \tilde{\mQ}}-\mam(\tilde{\mS}) \right] \label{eq:a-35b} \\
&\stackrel{\dagger}{=} \lim_{m\downarrow 0} \frac{1}{m} \left[\frac{1}{\upbeta} \mg(\mT \tilde{\mQ}) - \tr{\tilde{\mQ}^\trp \mT \rmR_{\mJ}(-2\upbeta \mT \tilde{\mQ}) } - \frac{1}{\upbeta} \sfE\log \sum_{\bvv} e^{-\upbeta (\bxx-\bvv)^\trp \mT \rmR_{\mJ}(-2\upbeta \mT \tilde{\mQ})(\bxx-\bvv)-\upbeta u(\bvv)} \right]  \label{eq:a-35c}
\end{align}
\end{subequations}
where $\star$ comes from the facts that $\mathsf{K}_n$ is bounded and the limits with respect to $m$ and $n$ are supposed to exchange, and $\dagger$ is deduced from \eqref{eq:a-32} and Lemma \ref{lem:a-2}. Finally by considering the definition of $\mg(\cdot)$, Proposition \ref{proposition:1} is concluded.

\newpage
\section{Proof of Proposition \ref{proposition:3}}
\label{app:b}
Starting from Assumption \ref{asp:3}, the replica correlation matrix is
\begin{align}
\mQ=  \frac{\chi }{\upbeta} \mI_m + q  \mone_m \label{eq:b-1}
\end{align}
for some non-negative real $\chi$ and $q$. Considering Definition \ref{def:replica_spin}, the Hamiltonian of the spin glass of replicas is given by
\begin{align}
\mae^\sfR(\bvv|\bxx)= (\bxx-\bvv)^{\trp} \mT \rmR_{\mJ}(-2 \upbeta \mT \mQ) (\bxx-\bvv) + u(\bvv) \label{eq:b-2}
\end{align}
with $\mT$ being defined in \eqref{eq:rep-5}. Denoting $\mR\coloneqq\mT \rmR_{\mJ}(-2 \upbeta \mT \mQ)$, it is shown in Appendix \ref{app:e} that $\mR$ has the same structure as the correlation matrix; thus, one can write
\begin{align}
\mR= e \mI_m - \upbeta \frac{f^2}{2} \mone_m, \label{eq:b-3}
\end{align}
for some real $f$ and $e$ which are functions of $\chi$ and $q$. Denoting the eigendecomposition of $\mQ$ as $\mQ=\mV \mD^\rmQ \mV^{\trp}$, we have\footnote{Note that $\mQ$ is full-rank and symmetric.}
\begin{subequations}
\begin{align}
\mT=\mV \mD^\rmT \mV^{\trp} \label{eq:b-4a} \\
\mR=\mV \mD^\rmR \mV^{\trp} \label{eq:b-4b}
\end{align}
\end{subequations}
where $\mD^\rmQ$, $\mD^\rmT$ and $\mD^\rmR$ are the diagonal matrices of eigenvalues. Therefore, we have
\begin{align}
\mD^\rmR= \mD^\rmT \rmR_\mJ(-2\upbeta \mD^\rmT \mD^\rmQ) \label{eq:b-5}
\end{align}
which equivalently states that for $a\in[1:m]$
\begin{align}
\lambda^\mR_a= \lambda^\mT_a \rmR_\mJ(-2\upbeta \lambda^\mT_a \lambda^\mQ_a) \label{eq:b-6}
\end{align}
with $\lambda^\mR_a$, $\lambda^\mQ_a$ and $\lambda^\mT_a$ being the eigenvalue of $\mR$, $\mQ$ and $\mT$ corresponding to the $a$th column of $\mV$. The matrices $\mR$, $\mQ$ and $\mT$ have two different corresponding eigenvalues, namely $\displaystyle \left\lbrace e-\upbeta m\frac{f^2}{2}, \frac{\chi}{\upbeta} +m q, \frac{1}{2\lambda}\left[ 1-\frac{m \upbeta\lambda_0}{\lambda+m\upbeta \lambda_0} \right] \right\rbrace$ which occur with multiplicity $1$ and $\displaystyle \left\lbrace e, \frac{\chi}{\upbeta}, \frac{1}{2\lambda} \right\rbrace$ which occur with multiplicity $m-1$. Substituting in \eqref{eq:b-6} and taking the limit when $m \downarrow 0$, $e$ and $f$ are found as
\begin{subequations}
\begin{align}
e &= \frac{1}{2\lambda} \rmR_{\mJ}(- \frac{\chi}{\lambda}), \label{eq:b-7a} \\
f^2 &= \frac{1}{\lambda^2} \frac{\partial}{\partial \chi} \left\lbrace \left[ \lambda_0 \chi - \lambda q \right] \rmR_{\mJ}(- \frac{\chi}{\lambda}) \right\rbrace. \label{eq:b-7b}
\end{align}
\end{subequations}
To pursue the analysis, we rewrite the Hamiltonian using \eqref{eq:b-3}
\begin{align}
\mae^\sfR(\bvv|\bxx)= e \norm{\bxx-\bvv}^2 - \upbeta \frac{f^2}{2} \tr{(\bxx-\bvv)(\bxx-\bvv)^{\trp} \mone_m } + u(\bvv), \label{eq:b-8}
\end{align}
and therefore, the partition function $\maz^{\sfR}(\upbeta|\bxx)$ is given by
\begin{align}
\maz^\sfR(\upbeta|\bxx)= \sum_{\{\vv_a \} } e^{- \upbeta e \norm{\bxx-\bvv}^2 + \upbeta^2 \frac{f^2}{2} \tr{(\bxx-\bvv)(\bxx-\bvv)^{\trp} \mone_m } -\upbeta u(\bvv)}. \label{eq:b-9}
\end{align}
Using the Gaussian integral, we have
\begin{align}
e^{\upbeta^2 \tfrac{f^2}{2} \tr{(\bxx-\bvv)(\bxx-\bvv)^{\trp} \mone_m }} = \int e^{-\upbeta f \left[ \sum_{a=1}^m (x-\vv_a) \right] z} \md z, \label{eq:b-10}
\end{align}
and thus, the partition function reduces to
\begin{align}
\maz^\sfR(\upbeta|\bxx)= \int \left[ \sum_{v} e^{- \upbeta \left[ e (x-v)^2 + f (x-v)z + u(v)\right]} \right]^m \md z \label{eq:b-11}
\end{align}
with $v \in \setX$. The parameters of the spin glass of replicas are then determined using the partition function. Starting with the normalized free energy, it reads
\begin{align}
\sfF^\sfR(\upbeta,m)= -\frac{1}{\upbeta m} \sfE_{x} \log \int \left[ \sum_{v} e^{- \upbeta \left[ e (x-v)^2 + f (x-v)z + u(v)\right]} \right]^m \md z. \label{eq:b-12}
\end{align}
Noting that $\int \md z$ takes expectation over the Gaussian distribution, one can use the Riesz equality in \eqref{eq:sm-7} to show that when $m$ varies in a vicinity of $0$
\begin{align}
\sfF^\sfR(\upbeta,m) &= -\frac{1}{\upbeta} \sfE \int \log \sum_{v}  e^{- \upbeta \left[ e (x-v)^2 + f (x-v)z + u(v)\right]} \md z + \epsilon_m \label{eq:b-13}
\end{align}
where $\epsilon_m$ tends to $0$ as $m\downarrow 0$ and the expectation is taken over $x\sim\rmp_x$. Consequently, as $m \downarrow 0$ the normalized free energy reads
\begin{align}
\sfF^\sfR(\upbeta) = \lim_{m\downarrow 0} \sfF^\sfR(\upbeta,m) = -\frac{1}{\upbeta} \sfE \int \log \sum_{v}  e^{- \upbeta \left[ e (x-v)^2 + f (x-v)z + u(v)\right]} \md z \label{eq:b-13.1}
\end{align}

The next parameters to be specified are $\chi$ and $q$. By determining the conditional distribution $\rmp_{\bvv|\bxx}^\upbeta$ and substituting in \eqref{eq:rep-6}, the following fixed point equations are deduced
\begin{subequations}
\begin{align}
\left[ \frac{\chi}{\upbeta}+q \right] m &= \sfE_{\bxx} \sum_{\bvv} \norm{\bxx-\bvv}^2 \ \rmp_{\bvv|\bxx}^\upbeta(\bvv|\bxx), \label{eq:b-14a} \\
\left[ \frac{\chi}{\upbeta}+m q \right] m &= \sfE_{\bxx} \sum_{\bvv} \tr{(\bxx-\bvv)(\bxx-\bvv)^\trp \mone_m} \ \rmp_{\bvv|\bxx}^\upbeta(\bvv|\bxx). \label{eq:b-14b}
\end{align}
\end{subequations}
where \eqref{eq:b-14a} and \eqref{eq:b-14b} are found by taking the trace and sum over all the entries of the both sides of \eqref{eq:rep-6}, respectively. One can directly evaluate the right hand sides of \eqref{eq:b-14a} and \eqref{eq:b-14b}; however, considering \eqref{eq:b-9}, it is straightforward to show that
\begin{subequations}
\begin{align}
&\sfE_{\bxx} \sum_{\bvv} \norm{\bxx-\bvv}^2 \ \rmp_{\bvv|\bxx}^\upbeta(\bvv|\bxx) = m \frac{\partial}{\partial e} \sfF^\sfR(\upbeta,m), \label{eq:b-15a} \\
&\sfE_{\bxx} \sum_{\bvv} \tr{(\bxx-\bvv)(\bxx-\bvv)^\trp \mone_m} \ \rmp_{\bvv|\bxx}^\upbeta(\bvv|\bxx)= -\frac{m}{\upbeta f} \frac{\partial}{\partial f} \sfF^\sfR(\upbeta,m). \label{eq:b-15b}
\end{align}
\end{subequations}
After substituting and taking the limit $m \downarrow 0$, the fixed point equations finally read
\begin{subequations}
\begin{align}
\frac{\chi}{\upbeta}+q &= \E \int  \frac{ \sum_{v} (v-x)^2 e^{- \upbeta \left[ e (x-v)^2 + f (x-v)z + u(v)\right]}}{ \sum_{v}  e^{- \upbeta \left[ e (x-v)^2 + f (x-v)z + u(v)\right]}} \md z \label{eq:b-16a} \\
\chi &= \frac{1}{f} \E \int \frac{ \sum_{v} (v-x) z  e^{- \upbeta \left[ e (x-v)^2 + f (x-v)z + u(v)\right]}}{ \sum_{v}  e^{- \upbeta \left[ e (x-v)^2 + f (x-v)z + u(v)\right]} } \md z. \label{eq:b-16b}
\end{align}
\end{subequations}
with $f$ and $e$ defined in \eqref{eq:b-7a} and \eqref{eq:b-7b}.

In order to determine the replicas' average distortion defined in \eqref{eq:rep-10} regarding the distortion function $\sfd(\cdot;\cdot)$, we replace the Hamiltonian by
\begin{align}
\mae_h^\sfR(\bvv|\bxx)=\mae^\sfR(\bvv|\bxx)+h\sum_{a=1}^m \sfd(\vv_a;x) \label{eq:b-17}
\end{align}
with $\mae^\sfR(\bvv|\bxx)$ given in \eqref{eq:b-8}, and take the steps as in \eqref{eq:b-9}-\eqref{eq:b-12} to find the modified form of the normalized free energy, i.e. $\sfF^\sfR(\upbeta,h,m)$. The replicas' average distortion is then evaluated as
\begin{subequations}
\begin{align}
\sfD^{\sfR}(\upbeta,m)&=\frac{\partial}{\partial h} \sfF^\sfR(\upbeta,h,m)|_{h=0} \label{eq:b-18a}\\
&=\E \int \frac{ \sum_{v} \sfd(v;x)  e^{- \upbeta \left[ e (x-v)^2 + f (x-v)z + u(v)\right]} }{ \sum_{v}  e^{- \upbeta \left[ e (x-v)^2 + f (x-v)z + u(v)\right]} }\md z. \label{eq:b-18b}
\end{align}
\end{subequations}
which does not depend on $m$, and thus, taking the limit $m \downarrow 0$ is not needed.

The last step is to take the zero temperature limit. Using the Laplace method of summation, as $\upbeta\uparrow\infty$ the fixed point equations reduce to
\begin{subequations}
\begin{align}
q &= \E \int (\rmg-x)^2 \ \md z \label{eq:b-19a}, \\
\chi &= \frac{1}{f} \E \int (\rmg-x) z  \ \md z. \label{eq:b-19b}
\end{align}
\end{subequations}
with $\rmg$ being defined as
\begin{align}
\rmg \coloneqq \arg \min_{v} \left[ e (x-v)^2 + f (x-v)z + u(v)\right]. \label{eq:b-20}
\end{align}
Taking the same approach, the replicas' average distortion at zero temperature reads
\begin{align}
\sfD^{\setW}&= \E \int \sfd(\rmg;x) \ \md z. \label{eq:b-21}
\end{align}

In order to avoid multiple solutions, we need to find the normalized free energy of the corresponding spin glass as given in Proposition \ref{proposition:1}. In fact, the fixed point equations in \eqref{eq:b-19a} and \eqref{eq:b-19b} may have different solutions, and therefore, the several asymptotics for the distortion can be obtained. In this case, the fixed point solution which minimizes the zero temperature free energy of the system and its corresponding asymptotic distortion are taken. Substituting in Proposition \ref{proposition:1}, the free energy of the corresponding spin glass at the inverse temperature $\upbeta$ is found as
\begin{align}
\sfF(\upbeta)=\frac{1}{2\lambda} \left[ \int_0^1 \rmF^{\upbeta}(\omega) \dif \omega -  \rmF^{\upbeta} (1) \right] +\sfF^{\sfR}(\upbeta) \label{eq:b-22}
\end{align}
where the function $\rmF^{\upbeta}(\cdot)$ is defined as
\begin{align}
\rmF^{\upbeta}(\omega) = \frac{\chi}{\upbeta} \rmR_{\mJ}(-\frac{\chi}{\lambda} \omega) + \left[q-\frac{\lambda_0}{\lambda} \chi \right] \frac{\dif}{\dif \omega} \left[ \omega \rmR_{\mJ}(-\frac{\chi}{\lambda} \omega) \right]. \label{eq:b-23}
\end{align}
By taking the limit as $\upbeta \uparrow \infty$, the zero temperature free energy reads
\begin{align}
\sfF^0=\frac{1}{2\lambda} \left[ \int_0^1 \rmF^{\infty}(\omega) \dif \omega -  \rmF^{\infty} (1) \right] + \E \int e (x-\rmg)^2 + f (x-\rmg)z + u(\rmg) \ \md z \label{eq:b-24}
\end{align}
with $\rmg$ being defined as in \eqref{eq:b-20} and $\rmF^{\infty}(\omega)\coloneqq \lim_{\upbeta\uparrow\infty}\rmF^{\upbeta}(\omega)$. By defining $\lams\coloneqq \left[2e\right]^{-1}$ and $\lams_0\coloneqq \left[4e^2\right]^{-1} f^2$ Proposition \ref{proposition:3} is concluded.

\newpage
\section{Proof of Proposition \ref{proposition:5}}
\label{app:c}
We take the same approach as in Appendix \ref{app:b}. Considering the replica correlation matrix to be of the form
\begin{align}
\mQ= \frac{\chi}{\upbeta} \mI_m+ p \mI_{\frac{m \upbeta}{\mu}} \otimes \mone_{\frac{\mu}{\upbeta}} +q \mone_m , \label{eq:c-1}
\end{align}
for some non-negative real $\chi$, $p$, $q$, and $\mu$, we need to evaluate the parameters of the spin glass of replicas defined in Definition \ref{def:replica_spin}. Starting with the Hamiltonian,
\begin{align}
\mae^\sfR(\bvv|\bxx)= (\bxx-\bvv)^{\trp} \mT \rmR_{\mJ}(-2 \upbeta \mT \mQ) (\bxx-\bvv) + u(\bvv) \label{eq:c-2}
\end{align}
where $\mT$ is given in \eqref{eq:rep-5}. As discussed in Appendix \ref{app:e}, for a given $\mu$ the matrix $\mR\coloneqq\mT \rmR_{\mJ}(-2 \upbeta \mT \mQ)$ is of the following form
\begin{align}
\mR= e \mI_m - \upbeta \frac{g^2}{2} \mI_{\frac{m \upbeta}{\mu}} \otimes \mone_{\frac{\mu}{\upbeta}} - \upbeta \frac{f^2}{2} \mone_m  \label{eq:c-3}
\end{align}
where $e$, $g$ and $f$ can be found in terms of $\chi$, $p$ and $q$. Using the eigendecomposition of $\mR$, $\mQ$ and $\mT$, it is then straightforward to show that for $a\in[1:m]$
\begin{align}
\lambda^\mR_a= \lambda^\mT_a \rmR_\mJ(-2\upbeta \lambda^\mT_a \lambda^\mQ_a) \label{eq:c-4}
\end{align}
where $\lambda^\mR_a$, $\lambda^\mQ_a$ and $\lambda^\mT_a$ denote the $a$th eigenvalues of $\mR$, $\mQ$ and $\mT$, respectively. Regarding the structure of $\mQ$ and $\mR$, there are three different sets of corresponding eigenvalues for $\mR$, $\mQ$ and $\mT$, namely\\
\begin{itemize}
\item $\displaystyle \left\lbrace e-\mu \frac{g^2}{2}-m\upbeta \frac{f^2}{2}, \frac{\chi+\mu p}{\upbeta} + m q, \frac{1}{2\lambda}\left[ 1-\frac{m\upbeta \lambda_0}{\lambda+m\upbeta \lambda_0} \right] \right\rbrace$ with multiplicity $1$,
\item $\displaystyle \left\lbrace e-\mu \frac{g^2}{2}, \frac{\chi+\mu p}{\upbeta}, \frac{1}{2\lambda}\right\rbrace$ with multiplicity $\displaystyle m \upbeta\mu^{-1}-1$, and
\item $\displaystyle \left\lbrace e, \frac{\chi}{\upbeta}, \frac{1}{2\lambda} \right\rbrace$ with multiplicity $\displaystyle m-m \upbeta\mu^{-1}$.\\
\end{itemize}
Thus, by substituting in \eqref{eq:c-4} and taking the limit when $m \downarrow 0$, $e$, $g$ and $f$ are given as
\begin{subequations}
\begin{align}
e &= \frac{1}{2\lambda} \rmR_{\mJ}(- \frac{\chi}{\lambda}), \label{eq:c-5a} \\
g^2 &= \frac{1}{\lambda\mu} \left[ \rmR_{\mJ}(- \frac{\chi}{\lambda}) - \rmR_{\mJ}(- \frac{\chi+\mu p}{\lambda}) \right], \label{eq:c-5b} \\
f^2 &= \frac{1}{\lambda^2} \frac{\partial}{\partial \chi} \left\lbrace \left[ \lambda_0 (\chi+\mu p) - \lambda q \right] \rmR_{\mJ}(- \frac{\chi+\mu p}{\lambda}) \right\rbrace. \label{eq:c-5c}
\end{align}
\end{subequations}
The next step is to evaluate the partition function. Substituting \eqref{eq:c-3} in \eqref{eq:c-2}, the Hamiltonian reads
\begin{align}
\mae^\sfR(\bvv|\bxx)= e \norm{\bxx-\bvv}^2 - \upbeta \frac{g^2}{2} \tr{(\bxx-\bvv)(\bxx-\bvv)^{\trp} \mI_{\frac{m \upbeta}{\mu}} \otimes \mone_{\frac{\mu}{\upbeta}} } - \upbeta \frac{f^2}{2} \tr{(\bxx-\bvv)(\bxx-\bvv)^{\trp} \mone_m } + u(\bvv). \label{eq:c-6}
\end{align}
The partition function is then determined as in \eqref{eq:rep-8}. Substituting in \eqref{eq:rep-8} and using the equalities
\begin{subequations}
\begin{align}
e^{\tfrac{1}{2} \upbeta^2 f^2  \tr{(\bxx-\bvv)(\bxx-\bvv)^{\trp} \mone_m }} &= \int e^{-\upbeta f \left[ \sum\limits_{a=1}^m (x-\vv_a) \right] z_0} \md z_0, \label{eq:c-7a} \\
e^{ \tfrac{1}{2} \upbeta^2 g^2 \tr{(\bxx-\bvv)(\bxx-\bvv)^{\trp} \mI_{\frac{m \upbeta}{\mu}} \otimes \mone_{\frac{\mu}{\upbeta}}}} &= \prod_{k=0}^{\Xi} \int e^{-\upbeta g \left[ \sum\limits_{a=\varrho_k}^{\breve{\varrho}_k} (x-\vv_a) \right] z_1} \md z_1, \label{eq:c-7b}
\end{align}
\end{subequations}
where $\varrho_k=k \mu\upbeta^{-1}+1$, $\breve{\varrho}_k=(k+1)\mu\upbeta^{-1}$, and $\Xi=m\upbeta\mu^{-1}-1$, the partition function is found as
\begin{align}
\maz^\sfR(\upbeta;\mu|\bxx)= \int \left[ \int \left[ \sum_{v} e^{- \upbeta \left[ e (x-v)^2 + (fz_0+gz_1) (x-v) + u(v)\right]} \right]^{\frac{\mu}{\upbeta}} \md z_1 \right]^{\frac{m\upbeta}{\mu}} \md z_0 \label{eq:c-8}
\end{align}
with $v \in \setX$ where we denoted $\mu$ in the argument of the partition function to indicate that the expression is determined for a given $\mu$. The normalized free energy of the spin glass of replicas then reads
\begin{align}
\sfF^\sfR(\upbeta,m;\mu)= -\frac{1}{\upbeta m} \sfE_{x} \log \int \left[ \int \left[ \sum_{v} e^{- \upbeta \left[ e (x-v)^2 + (fz_0+gz_1) (x-v) + u(v)\right]} \right]^{\frac{\mu}{\upbeta}} \md z_1 \right]^{\frac{m\upbeta}{\mu}} \md z_0. \label{eq:c-9}
\end{align}
Using the Riesz equality and taking the limit $m \downarrow 0$, the normalized free energy reduces to
\begin{align}
\sfF^\sfR(\upbeta;\mu) = \lim_{m\downarrow 0} \sfF^\sfR(\upbeta,m;\mu) = -\frac{1}{\mu} \sfE \int \log \left\lbrace \int \left[ \sum_{v} e^{- \upbeta \left[ e (x-v)^2 + (fz_0+gz_1) (x-v) + u(v)\right]} \right]^{\frac{\mu}{\upbeta}} \md z_1 \right\rbrace \md z_0. \label{eq:c-10}
\end{align}
In order to find the fixed point equations, we use \eqref{eq:rep-6}; therefore,
\begin{subequations}
\begin{align}
\left[ \frac{\chi}{\upbeta}+q+p \right] m &= \sfE_{\bxx} \sum_{\bvv} \norm{\bxx-\bvv}^2 \ \rmp_{\bvv|\bxx}^\upbeta(\bvv|\bxx), \label{eq:c-11a} \\
\left[ \frac{\chi}{\upbeta}+\frac{\mu p}{\upbeta}+\frac{\mu q}{\upbeta} \right] m &= \sfE_{\bxx} \sum_{\bvv} \tr{(\bxx-\bvv)(\bxx-\bvv)^{\trp} \mI_{\frac{m \upbeta}{\mu}} \otimes \mone_{\frac{\mu}{\upbeta}}} \ \rmp_{\bvv|\bxx}^\upbeta(\bvv|\bxx), \label{eq:c-11b}\\
\left[ \frac{\chi}{\upbeta}+\frac{\mu p}{\upbeta}+m q \right] m&= \sfE_{\bxx} \sum_{\bvv} \tr{(\bxx-\bvv)(\bxx-\bvv)^\trp \mone_m} \ \rmp_{\bvv|\bxx}^\upbeta(\bvv|\bxx) \label{eq:c-11c}
\end{align}
\end{subequations}
where \eqref{eq:c-11a}, \eqref{eq:c-11b} and \eqref{eq:c-11c} are concluded by taking the trace, sum over the diagonal blocks and sum over all the entries of the both sides of \eqref{eq:rep-6}, respectively. To evaluate the right hand sides of \eqref{eq:c-11a}-\eqref{eq:c-11c}, we take the alternative approach and express the expectations as
\begin{subequations}
\begin{align}
&\sfE_{\bxx} \sum_{\bvv} \norm{\bxx-\bvv}^2 \ \rmp_{\bvv|\bxx}^\upbeta(\bvv|\bxx) = m \frac{\partial}{\partial e} \sfF^\sfR(\upbeta,m;\mu), \label{eq:c-12a} \\
&\sfE_{\bxx} \sum_{\bvv} \tr{(\bxx-\bvv)(\bxx-\bvv)^{\trp} \mI_{\frac{m \upbeta}{\mu}} \otimes \mone_{\frac{\mu}{\upbeta}}} \ \rmp_{\bvv|\bxx}^\upbeta(\bvv|\bxx) = -\frac{m}{\upbeta g} \frac{\partial}{\partial g} \sfF^\sfR(\upbeta,m;\mu), \label{eq:c-12b}\\
&\sfE_{\bxx} \sum_{\bvv} \tr{(\bxx-\bvv)(\bxx-\bvv)^\trp \mone_m} \ \rmp_{\bvv|\bxx}^\upbeta(\bvv|\bxx)= -\frac{m}{\upbeta f} \frac{\partial}{\partial f} \sfF^\sfR(\upbeta,m;\mu). \label{eq:c-12c}
\end{align}
\end{subequations}
Taking the derivatives and limit $m \downarrow 0$, the fixed point equations finally reduce to
\begin{subequations}
\begin{align}
\frac{\chi}{\upbeta}+q+p &= \E \int \frac{ \sum_{v} (v-x)^2 e^{- \upbeta \left[  e (x-v)^2 + (fz_0+gz_1) (x-v) + u(v)\right]} }{\sum_{v}  e^{- \upbeta \left[  e (x-v)^2 + (fz_0+gz_1) (x-v) + u(v) \right]} } \tilde{\Lambda}^\upbeta \ \md z_1 \md z_0 \label{eq:c-13a} \\
\chi+\mu p+\mu q &= \frac{1}{g} \E \int \frac{ \sum_{v} (v-x) z_1  e^{- \upbeta \left[  e (x-v)^2 + (fz_0+gz_1) (x-v) + u(v)\right]} }{ \sum_{v}  e^{- \upbeta \left[  e (x-v)^2 + (fz_0+gz_1) (x-v) + u(v)\right]} } \tilde{\Lambda}^\upbeta \ \md z_1 \md z_0, \label{eq:c-13b}\\
\chi+\mu p &= \frac{1}{f} \E \int \frac{ \sum_{v} (v-x) z_0  e^{- \upbeta \left[  e (x-v)^2 + (fz_0+gz_1) (x-v) + u(v)\right]} }{ \sum_{v}  e^{- \upbeta \left[  e (x-v)^2 + (fz_0+gz_1) (x-v) + u(v)\right]} } \tilde{\Lambda}^\upbeta \ \md z_1 \md z_0 \label{eq:c-13c}
\end{align}
\end{subequations}
with $\tilde{\Lambda}^\upbeta\coloneqq \left[ \int \Lambda^\upbeta \md z_1\right]^{-1} \Lambda^\upbeta$ and $\Lambda^\upbeta$ being defined as
\begin{align}
\Lambda^\upbeta\coloneqq  \left[ \sum_{v} e^{- \upbeta \left[  e (x-v)^2 + (fz_0+gz_1) (x-v) + u(v)\right]} \right]^{\frac{\mu}{\upbeta}}.   \label{eq:c-14}
\end{align}
The replicas' average distortion regarding the distortion function $\sfd(\cdot;\cdot)$ is further determined by modifying the Hamiltonian as
\begin{align}
\mae_h^\sfR(\bvv|\bxx)=\mae^\sfR(\bvv|\bxx)+h\sum_{a=1}^m \sfd(\vv_a;x) \label{eq:c-15}
\end{align}
with $\mae^\sfR(\bvv|\bxx)$ given in \eqref{eq:c-6}, and taking the steps as in \eqref{eq:c-6}-\eqref{eq:c-9} to find the modified form of the normalized free energy, i.e. $\sfF^\sfR(\upbeta,h,m;\mu)$. The replicas' average distortion then reads
\begin{subequations}
\begin{align}
\sfD^{\sfR}(\upbeta;\mu)&= \lim_{m\downarrow0}  \frac{\partial}{\partial h} \sfF^\sfR(\upbeta,h,m;\mu)|_{h=0} \label{eq:c-16a}\\
&= \E \int \frac{ \sum_{v} \sfd(v;x) e^{- \upbeta \left[  e (x-v)^2 + (fz_0+gz_1) (x-v) + u(v)\right]} }{\sum_{v}  e^{- \upbeta \left[  e (x-v)^2 + (fz_0+gz_1) (x-v) + u(v) \right]} } \tilde{\Lambda}^\upbeta \ \md z_1 \md z_0. \label{eq:c-16b}
\end{align}
\end{subequations}
The analysis is concluded by taking the zero temperature limit. As $\upbeta\uparrow\infty$, \eqref{eq:c-13a}-\eqref{eq:c-13c} read
\begin{subequations}
\begin{align}
q + p &= \E \int (\rmg-x)^2 \tilde{\Lambda} \md z_1 \md z_0 \label{eq:c-17a}, \\
\chi+ \mu q + \mu p &= \frac{1}{g} \E \int (\rmg-x) z _1 \tilde{\Lambda} \md z_1 \md z_0, \label{eq:c-17b} \\
\chi+ \mu p &= \frac{1}{f} \E \int (\rmg-x) z _0 \tilde{\Lambda} \md z_1 \md z_0, \label{eq:c-17c} 
\end{align}
\end{subequations}
where $\rmg$ is defined as
\begin{align}
\rmg \coloneqq \arg \min_{v} \left[ e (x-v)^2 + (fz_0+gz_1) (x-v) + u(v)\right] \label{eq:c-18}
\end{align}
and $\tilde{\Lambda}\coloneqq \left[\int \Lambda \md z_1 \right]^{-1} \Lambda$ with $\Lambda$ denoting
\begin{subequations}
\begin{align}
\Lambda &\coloneqq \lim_{\upbeta \uparrow \infty} \Lambda^\upbeta \label{eq:c-19a} \\ 
&=  e^{- \mu  \left[ e (x-\rmg)^2 + (fz_0+gz_1) (x-\rmg) + u(\rmg)\right]}.   \label{eq:c-19b}
\end{align}
\end{subequations}
Moreover, the asymptotic distortion for a given $\mu$ reads
\begin{align}
\sfD^{\setW}&= \E \int \sfd(\rmg;x) \tilde{\Lambda} \md z_1 \md z_0. \label{eq:c-20}
\end{align}

\eqref{eq:c-10} as well as \eqref{eq:c-17a}-\eqref{eq:c-20} are determined in terms of $\mu$. Moreover, for a given $\mu$, multiple solution to the fixed point equations can be found. Proposition \ref{proposition:1} suggests us to choose the solution which minimizes the free energy. Therefore, one needs to find the optimal $\mu$, and its corresponding $\chi$, $p$ and $q$, such that the free energy meets its minimum value. As the second law of thermodynamics is satisfied at any inverse temperature, we should initially search for the optimal $\mu$ considering a given $\upbeta$. We, then, find the corresponding $\chi$, $p$, and $q$ which minimize the zero temperature free energy. Using Proposition \ref{proposition:1}, the free energy at the inverse temperature $\upbeta$ for a given $\mu$ is written as
\begin{align}
\sfF(\upbeta;\mu)=\frac{1}{2\lambda} \left[ \int_0^1 \rmF^{\upbeta}(\omega;\mu) \dif \omega -  \rmF^{\upbeta} (1;\mu) \right] +\sfF^{\sfR}(\upbeta;\mu) \label{eq:c-21}
\end{align}
where the function $\rmF^{\upbeta}(\cdot;\mu)$ is defined as
\begin{align}
\rmF^{\upbeta}(\omega;\mu) = \frac{1}{\mu} \frac{\dif}{\dif \omega} \int_{\chi \omega}^{\left[\chi+\mu p\right] \omega} \rmR_{\mJ}(-\frac{t}{\lambda} ) \dif t + \frac{\chi}{\upbeta} \rmR_{\mJ}(-\frac{\chi}{\lambda} \omega) + \left[q-\lambda_0 \frac{\chi+\mu p}{\lambda} \right] \frac{\dif}{\dif \omega} \left[ \omega \rmR_{\mJ}(-\frac{\chi+\mu p}{\lambda} \omega) \right]. \label{eq:c-22}
\end{align}
To find $\mu$ at the thermal equilibrium, we let
\begin{align}
\frac{\partial}{\partial \mu} \sfF(\upbeta;\mu) = 0 \label{eq:c-23}
\end{align}
Using the equalities \eqref{eq:c-5a}-\eqref{eq:c-5c}, \eqref{eq:c-23} concludes that $\mu$ satisfies
\begin{align}
\frac{1}{2\lambda} \left[ p \rmR_{\mJ}(-\frac{\chi}{\lambda}) + q \rmR_{\mJ}(-\frac{\chi}{\lambda}) - q \rmR_{\mJ}(-\frac{\chi+\mu p}{\lambda}) \right] = \sfF^{\sfR}(\upbeta;\mu) + \frac{1}{2\lambda\mu} \int_{\chi}^{\chi+\mu p} \rmR_{\mJ}(-\frac{t}{\lambda}) \dif t \nonumber \\
+ \E \frac{1}{\upbeta} \int \log \left[ \sum_{v} e^{- \upbeta \left[  e (x-v)^2 + (fz_0+gz_1) (x-v) + u(v)\right]} \right] \tilde{\Lambda}^\upbeta \md z_1 \md z_0 \label{eq:c-24}
\end{align}
which as $\upbeta \uparrow \infty$ reduces to
\begin{align}
\frac{1}{2\lambda} \left[ p \rmR_{\mJ}(-\frac{\chi}{\lambda}) + q \rmR_{\mJ}(-\frac{\chi}{\lambda}) - q \rmR_{\mJ}(-\frac{\chi+\mu p}{\lambda}) \right] = \frac{1}{2\lambda\mu} \int_{\chi}^{\chi+\mu p} \rmR_{\mJ}(-\frac{t}{\lambda}) \dif t + \E \frac{1}{\mu} \int \log \tilde{\Lambda} \ \tilde{\Lambda} \md z_1 \md z_0. \label{eq:c-25}
\end{align}
Denoting the solution to \eqref{eq:c-25} by $\mu^\star$, the free energy of the corresponding spin glass is then given as $\sfF(\upbeta)=\sfF(\upbeta;\mu^\star)$ which at the zero temperature reads
\begin{align}
\sfF^0=\frac{1}{2\lambda} \left[ \int_0^1 \rmF^{\infty}(\omega) \dif \omega -  \rmF^{\infty} (1) \right] - \frac{1}{\mu} \E \int \log \left[ \int \Lambda \md z_1 \right] \ \md z_0 \label{eq:c-27}
\end{align}
with
\begin{align}
\rmF^{\infty}(\omega)\coloneqq \lim_{\upbeta\uparrow\infty}\rmF^{\upbeta}(\omega;\mu^\star).
\end{align}
Finally by defining $\lams\coloneqq \left[2e\right]^{-1}$, $\lams_0\coloneqq \left[4e^2\right]^{-1} f^2$ and $\lams_1\coloneqq \left[4e^2\right]^{-1} g_1^2$, Proposition \ref{proposition:5} is concluded.

\newpage
\section{Proof of Proposition \ref{proposition:7}}
\label{app:d}
The strategy here is to extend the approach in Appendix \ref{app:c} to a general number of breaking steps. Following Appendix \ref{app:e} and considering $\mQ$ as
\begin{align}
\mQ= \frac{\chi}{\upbeta} \mI_m+ \sum_{\kappa=1}^b p_\kappa \mI_{\frac{m \upbeta}{\mu_\kappa}} \otimes \mone_{\frac{\mu_\kappa}{\upbeta}} +q \mone_m , \label{eq:d-1}
\end{align}
the frequency domain correlation matrix $\mR\coloneqq\mT \rmR_{\mJ}(-2 \upbeta \mT \mQ)$ is written as
\begin{align}
\mR= e \mI_m - \upbeta \sum_{\kappa=1}^b \frac{g_\kappa^2}{2} \mI_{\frac{m \upbeta}{\mu_\kappa}} \otimes \mone_{\frac{\mu_\kappa}{\upbeta}} - \upbeta \frac{f^2}{2} \mone_m  \label{eq:d-2}
\end{align}
considering a given vector $\bmu=\left[ \mu_1, \ldots, \mu_b \right]^\trp$, such that
\begin{align}
\mu_{\kappa+1} = \vartheta_{\kappa+1} \mu_{\kappa}, \label{eq:d-2.1}
\end{align}
with $\{ \vartheta_\kappa \}$ being non-negative integers, $e$, $f$ and $\{g_\kappa\}$ are then found in terms of $\chi$, $q$ and $\{p_\kappa\}$ by letting
\begin{align}
\lambda^\mR_a= \lambda^\mT_a \rmR_\mJ(-2\upbeta \lambda^\mT_a \lambda^\mQ_a) \label{eq:d-3}
\end{align}
for $a\in[1:m]$ where $\lambda^\mR_a$, $\lambda^\mQ_a$ and $\lambda^\mT_a$ denote the $a$th corresponding eigenvalues of $\mR$, $\mQ$ and $\mT$. As long as the constraint in \eqref{eq:d-2.1} holds, $\mQ$, $\mT$ and $\mR$ have $b+2$ different sets of corresponding eigenvalues specified by\\
\begin{itemize}
\item $\displaystyle \left\lbrace e-\sum_{\kappa=1}^b \mu_\kappa \frac{g_\kappa^2}{2}-m\upbeta \frac{f^2}{2}, \frac{\chi}{\upbeta} + \sum_{\kappa=1}^b p_\kappa \frac{\mu_\kappa}{\upbeta} + m q, \frac{1}{2\lambda}\left[ 1-\frac{m\upbeta\lambda_0}{\lambda+m\upbeta\lambda_0} \right] \right\rbrace$ with multiplicity $\Theta_{b+1}(m)=1$,
\item $\displaystyle \left\lbrace e-\sum_{\kappa=1}^b \mu_\kappa \frac{g_\kappa^2}{2}, \frac{\chi}{\upbeta} + \sum_{\kappa=1}^b p_\kappa \frac{\mu_\kappa}{\upbeta}, \frac{1}{2\lambda} \right\rbrace$ with multiplicity $\Theta_b(m)=m\upbeta \mu_b^{-1} -1$,
\item $\displaystyle \left\lbrace e-\sum_{\varsigma=1}^{\kappa} \mu_\varsigma \frac{g_\varsigma^2}{2}, \frac{\chi}{\upbeta} + \sum_{\varsigma=1}^\kappa p_\varsigma \frac{\mu_\varsigma}{\upbeta}, \frac{1}{2\lambda} \right\rbrace$ with multiplicity $\Theta_\kappa(m)=m\upbeta \left( \mu_\kappa^{-1} -\mu_{\kappa+1}^{-1} \right)$ for $\kappa\in[1:b-1]$, and
\item $\displaystyle \left\lbrace e, \frac{\chi}{\upbeta}, \frac{1}{2\lambda} \right\rbrace$ with multiplicity $\Theta_0(m)=m-m\upbeta \mu_1^{-1}$.\\
\end{itemize}
Substituting in \eqref{eq:d-3} $e$, $f$ and $\{g_\kappa\}$ for $\kappa\in[1:b]$ are determined in terms of $\chi$, $q$ and $\{p_\kappa\}$ as
\begin{subequations}
\begin{align}
e &= \frac{1}{2\lambda} \rmR_{\mJ}(- \frac{\chi}{\lambda}), \label{eq:d-3.1a} \\
g_\kappa^2 &= \frac{1}{\lambda\mu_\kappa} \left[ \rmR_{\mJ}(- \frac{\tilde{\chi}_{\kappa-1}}{\lambda}) - \rmR_{\mJ}(-\frac{\tilde{\chi}_{\kappa}}{\lambda}) \right], \label{eq:d-3.1b} \\
f^2 &= \frac{1}{\lambda^2} \frac{\partial}{\partial \tilde{\chi}_{b}} \left\lbrace \left[ \lambda_0 \tilde{\chi}_{b} - \lambda q \right] \rmR_{\mJ}(- \frac{\tilde{\chi}_{b}}{\lambda}) \right\rbrace. \label{eq:d-3.1c}
\end{align}
\end{subequations}
where we define $\tilde{\chi}_0 \coloneqq \chi$ and 
\begin{align}
\tilde{\chi}_\kappa \coloneqq \chi+\sum_{\varsigma=1}^{\kappa} \mu_\varsigma p_\varsigma \label{eq:d-21}
\end{align}
for $\kappa\in[1:b]$. The Hamiltonian of the spin glass of replicas is then determined as in \eqref{eq:rep-4}. Substituting the Hamiltonian in \eqref{eq:rep-8} and using the equalities
\begin{subequations}
\begin{align}
e^{\tfrac{1}{2} \upbeta^2 f^2  \tr{(\bxx-\bvv)(\bxx-\bvv)^{\trp} \mone_m }} &= \int e^{-\upbeta f \left[ \sum_{a=1}^m (x-\vv_a) \right] z_0} \md z_0, \label{eq:d-4a} \\
e^{ \tfrac{1}{2} \upbeta^2 g_\kappa^2 \tr{(\bxx-\bvv)(\bxx-\bvv)^{\trp} \mI_{\frac{m \upbeta}{\mu_\kappa}} \otimes \mone_{\frac{\mu_\kappa}{\upbeta}}}} &= \prod_{k=0}^{\Xi_\kappa} \int e^{-\upbeta g_\kappa \left[ \sum\limits_{a=\varrho_k^\kappa}^{\breve{\varrho}_k^\kappa} (x-\vv_a) \right] z_\kappa} \md z_\kappa, \label{eq:d-4b}
\end{align}
\end{subequations}
with $\varrho_k^\kappa=k \mu_\kappa \upbeta^{-1}+1$, $\breve{\varrho}_k^\kappa=(k+1)\mu_\kappa\upbeta^{-1}$, and $\Xi_\kappa=m\upbeta\mu_\kappa^{-1}-1$, the partition function finally reads
\begin{align}
\maz^\sfR(\upbeta;\bmu|\bxx)= \int \left[ \bigwedge_{\varsigma=2}^{b} \int \left[ \int \left[ \sum_{v} e^{- \upbeta \left[ e (x-v)^2 + (fz_0+\sum\limits_{\kappa=1}^b g_\kappa z_\kappa) (x-v) + u(v)\right]} \right]^{\frac{\mu_1}{\upbeta}} \md z_1 \right]^{\frac{\mu_{\varsigma}}{\mu_{\varsigma-1}}} \md z_\varsigma \right]^{\frac{m\upbeta}{\mu_b}} \md z_0 \label{eq:d-5}
\end{align}
with $v \in \setX$ where for the sequences $\{\xi_\varsigma\}$ and $\{z_\varsigma\}$ we define
\begin{align}
\bigwedge_{\varsigma=1}^{b} \int \rmF^{\xi_\varsigma} \md z_\varsigma \coloneqq \int \left[ \cdots \int \left[ \int \rmF^{\xi_1} \md z_1 \right]^{\xi_2} \md z_2 \cdots \right]^{\xi_b} \md z_b. \label{eq:d-6}
\end{align}
Consequently, one evaluates the free energy as in \eqref{eq:rep-9} which by using the Riesz equality when $m \downarrow 0$ reduces to
\begin{align}
\sfF^\sfR(\upbeta;\bmu) = -\frac{1}{\mu_b} \sfE \int \log \left\lbrace \bigwedge_{\varsigma=1}^{b} \int \left[ \sum_{v} e^{- \upbeta \left[ e (x-v)^2 + (fz_0+\sum\limits_{\kappa=1}^b g_\kappa z_\kappa) (x-v) + u(v)\right]} \right]^{\frac{\mu_{\varsigma}}{\mu_{\varsigma-1}}} \md z_\varsigma \right\rbrace \md z_0 \label{eq:d-7}
\end{align}
where we have defined $\mu_0=\upbeta$ for sake of compactness. The fixed point equations are, moreover, found via \eqref{eq:rep-6} where we have
\begin{subequations}
\begin{align}
\left[ \frac{\chi}{\upbeta}+\sum\limits_{\kappa=1}^b p_\kappa+q \right] m &= \sfE_{\bxx} \sum_{\bvv} \norm{\bxx-\bvv}^2 \ \rmp_{\bvv|\bxx}^\upbeta(\bvv|\bxx), \label{eq:d-8a} \\
\left[ \frac{\tilde{\chi}_{\kappa-1}}{\upbeta} +\frac{\mu_\kappa}{\upbeta} \left(\sum_{\varsigma=\kappa}^b p_\varsigma + q \right) \right] m &= \sfE_{\bxx} \sum_{\bvv} \tr{(\bxx-\bvv)(\bxx-\bvv)^{\trp} \mI_{\frac{m \upbeta}{\mu_\kappa}} \otimes \mone_{\frac{\mu_\kappa}{\upbeta}}} \ \rmp_{\bvv|\bxx}^\upbeta(\bvv|\bxx), \label{eq:d-8b}\\
\left[ \frac{\tilde{\chi}_b}{\upbeta}+m q \right] m&= \sfE_{\bxx} \sum_{\bvv} \tr{(\bxx-\bvv)(\bxx-\bvv)^\trp \mone_m} \ \rmp_{\bvv|\bxx}^\upbeta(\bvv|\bxx) \label{eq:d-8c}
\end{align}
\end{subequations}
for $\kappa\in[1:b]$. \eqref{eq:d-8a} and \eqref{eq:d-8c} are found by taking trace and sum over all the entries from both sides of \eqref{eq:rep-6}, respectively. \eqref{eq:d-8b} is moreover concluded by adding up the entries over the diagonal blocks of size $\mu_\kappa\upbeta^{-1}$. The right hand sides of \eqref{eq:d-8a}-\eqref{eq:d-8c} can then be evaluated using the equalities
\begin{subequations}
\begin{align}
&\sfE_{\bxx} \sum_{\bvv} \norm{\bxx-\bvv}^2 \ \rmp_{\bvv|\bxx}^\upbeta(\bvv|\bxx) = m \frac{\partial}{\partial e} \sfF^\sfR(\upbeta,m;\bmu), \label{eq:d-9a} \\
&\sfE_{\bxx} \sum_{\bvv} \tr{(\bxx-\bvv)(\bxx-\bvv)^{\trp} \mI_{\frac{m \upbeta}{\mu_\kappa}} \otimes \mone_{\frac{\mu_\kappa}{\upbeta}}} \ \rmp_{\bvv|\bxx}^\upbeta(\bvv|\bxx) = -\frac{m}{\upbeta g_\kappa} \frac{\partial}{\partial g_\kappa} \sfF^\sfR(\upbeta,m;\bmu), \label{eq:d-9b}\\
&\sfE_{\bxx} \sum_{\bvv} \tr{(\bxx-\bvv)(\bxx-\bvv)^\trp \mone_m} \ \rmp_{\bvv|\bxx}^\upbeta(\bvv|\bxx)= -\frac{m}{\upbeta f} \frac{\partial}{\partial f} \sfF^\sfR(\upbeta,m;\bmu). \label{eq:d-9c}
\end{align}
\end{subequations}
Thus, the fixed point equations are finally concluded as
\begin{subequations}
\begin{align}
&\frac{\chi}{\upbeta}+\sum\limits_{\kappa=1}^b p_\kappa+q = \E \int \frac{ \sum_{v} (v-x)^2 \ e^{- \upbeta \left[  e (x-v)^2 + (fz_0+\sum\limits_{\kappa=1}^b g_\kappa z_\kappa) (x-v) + u(v)\right]} }{\sum_{v}  e^{- \upbeta \left[  e (x-v)^2 + (fz_0+\sum\limits_{\kappa=1}^b g_\kappa z_\kappa) (x-v) + u(v) \right]} } \prod_{\kappa=1}^b \tilde{\Lambda}^\upbeta_\kappa \ \md z_\kappa \md z_0 \label{eq:d-10a} \\
&\tilde{\chi}_{\kappa-1}+\mu_\kappa \left(\sum_{\varsigma=\kappa}^b p_\varsigma + q \right) = \frac{1}{g_\kappa}\E \int \frac{ \sum_{v} (v-x)z_\kappa \ e^{- \upbeta \left[  e (x-v)^2 + (fz_0+\sum\limits_{\kappa=1}^b g_\kappa z_\kappa) (x-v) + u(v)\right]} }{\sum_{v}  e^{- \upbeta \left[  e (x-v)^2 + (fz_0+\sum\limits_{\kappa=1}^b g_\kappa z_\kappa) (x-v) + u(v) \right]} } \times \nonumber \\ & \hspace{12cm} \times\prod_{\kappa=1}^b \tilde{\Lambda}^\upbeta_\kappa \ \md z_\kappa \md z_0 \label{eq:d-10b} \\
&\tilde{\chi}_{b} =\frac{1}{f}\E \int \frac{ \sum_{v} (v-x)z_0 \ e^{- \upbeta \left[  e (x-v)^2 + (fz_0+\sum\limits_{\kappa=1}^b g_\kappa z_\kappa) (x-v) + u(v)\right]} }{\sum_{v}  e^{- \upbeta \left[  e (x-v)^2 + (fz_0+\sum\limits_{\kappa=1}^b g_\kappa z_\kappa) (x-v) + u(v) \right]} } \prod_{\kappa=1}^b \tilde{\Lambda}^\upbeta_\kappa \ \md z_\kappa \md z_0. \label{eq:d-10c}
\end{align}
\end{subequations}
for $\kappa\in[1:b]$ in which we denote $\tilde{\Lambda}^{\upbeta}_\kappa\coloneqq \left[ \int \Lambda^\upbeta_\kappa \md z_\kappa\right]^{-1} \Lambda^\upbeta_\kappa$ with $\Lambda^\upbeta_1$ 
\begin{align}
\Lambda^{\upbeta}_{1} \coloneqq \left[ \sum_{v}  e^{- \upbeta \left[  e (x-v)^2 + (fz_0+\sum\limits_{\kappa=1}^b g_\kappa z_\kappa) (x-v) + u(v) \right]} \right]^{\tfrac{\mu_{1}}{\upbeta}} \label{eq:d-11}
\end{align}
and $\{\Lambda^\upbeta_\kappa\}$ for $\kappa\in[2:b]$ being recursively defined as
\begin{align}
\Lambda^{\upbeta}_{\kappa} \coloneqq \left[ \int \Lambda^\upbeta_{\kappa-1} \ \md z_{\kappa-1} \right]^{\tfrac{\mu_{\kappa}}{\mu_{\kappa-1}}}. \label{eq:d-12}
\end{align}

The replicas' average distortion regarding the distortion function $\sfd(\cdot;\cdot)$ is further determined using the Hamiltonian modification technique employed in Appendix \ref{app:b} and \ref{app:c}. After modifying the Hamiltonian and taking the derivatives, the average distortion at the inverse temperature $\upbeta$ is given by
\begin{align}
\sfD^{\sfR}(\upbeta;\bmu)
= \E \int \frac{ \sum_{v} \sfd(v;x) e^{- \upbeta \left[  e (x-v)^2 + (fz_0+\sum\limits_{\kappa=1}^b g_\kappa z_\kappa) (x-v) + u(v)\right]} }{\sum_{v}  e^{- \upbeta \left[  e (x-v)^2 + (fz_0+\sum\limits_{\kappa=1}^b g_\kappa z_\kappa) (x-v) + u(v) \right]} } \prod_{\kappa=1}^b \tilde{\Lambda}^\upbeta_\kappa \ \md z_\kappa \md z_0. \label{eq:d-13}
\end{align}
Finally, by taking the limit $\upbeta \uparrow \infty$, we find the asymptotic distortion as
\begin{align}
\sfD^{\setW}&= \E \int \sfd(\rmg;x) \prod_{\kappa=1}^b \tilde{\Lambda}_\kappa \ \md z_\kappa \md z_0 \label{eq:d-14}
\end{align}
where $\rmg$ is defined as
\begin{align}
\rmg \coloneqq \arg \min_{v} \left[ e (x-v)^2 + (fz_0+\sum\limits_{\kappa=1}^b g_\kappa z_\kappa) (x-v) + u(v)\right], \label{eq:d-15}
\end{align}
and $\tilde{\Lambda}_\kappa$ denotes the limiting factor $\tilde{\Lambda}_\kappa^\infty$. Considering the definition of $\tilde{\Lambda}_\kappa^\upbeta$, $\tilde{\Lambda}_\kappa$ reads $\tilde{\Lambda}_\kappa = \left[ \int \Lambda_\kappa \md z_\kappa\right]^{-1} \Lambda_\kappa$ with
\begin{align}
\Lambda_{1} \coloneqq  e^{-\mu_1  \left[  e (x-\rmg)^2 + (fz_0+\sum\limits_{\kappa=1}^b g_\kappa z_\kappa) (x-\rmg) + u(\rmg) \right]} \label{eq:d-16}
\end{align}
and $\{\Lambda_\kappa\}$ for $\kappa\in[2:b]$ 
\begin{align}
\Lambda_{\kappa} \coloneqq \left[ \int \Lambda_{\kappa-1} \ \md z_{\kappa-1} \right]^{\tfrac{\mu_{\kappa}}{\mu_{\kappa-1}}}. \label{eq:d-17}
\end{align}
Moreover, the fixed point equations reduce to
\begin{subequations}
\begin{align}
\sum_{\kappa=1}^b p_\kappa + q &= \E \int (\rmg-x)^2 \prod_{\kappa=1}^b \tilde{\Lambda}_\kappa \ \md z_\kappa \md z_0 \label{eq:d-18a}, \\
\tilde{\chi}_{\kappa-1}+\mu_\kappa \left(\sum_{\varsigma=\kappa}^b p_\varsigma + q \right) &= \frac{1}{g_\kappa} \E \int (\rmg-x) z _\kappa \prod_{\kappa=1}^b \tilde{\Lambda}_\kappa \ \md z_\kappa \md z_0 , \label{eq:d-18b} \\
\tilde{\chi}_{b} &= \frac{1}{f} \E \int (\rmg-x) z _0  \prod_{\kappa=1}^b \tilde{\Lambda}_\kappa \ \md z_\kappa \md z_0, \label{eq:d-18c} 
\end{align}
\end{subequations}
for $\kappa \in [1:b]$.

As in the 1\ac{rsb} ansatz, we set $\bmu$ to be the extreme point of the free energy at a given inverse temperature $\upbeta$, in order to satisfy the second law of thermodynamics. The solution needs to be found over the set of non-negative real vectors which satisfy the constraint in \eqref{eq:d-2.1}. The parameters of the ansatz, however, are finally taken such that the zero temperature free energy is minimized.

Using Proposition \ref{proposition:1}, the free energy of the corresponding spin glass for a given vector $\bmu$ is written as
\begin{align}
\sfF(\upbeta;\bmu)=\frac{1}{2\lambda} \left[ \int_0^1 \rmF^{\upbeta}(\omega;\bmu) \dif \omega -  \rmF^{\upbeta} (1;\bmu) \right] +\sfF^{\sfR}(\upbeta;\bmu) \label{eq:d-19}
\end{align}
where the function $\rmF^{\upbeta}(\cdot;\bmu)$ is defined as
\begin{align}
\rmF^{\upbeta}(\omega;\bmu) = \sum\limits_{\kappa=1}^b \frac{1}{\mu_\kappa} \frac{\dif}{\dif \omega} \int_{\tilde{\chi}_{\kappa-1} \omega}^{\tilde{\chi}_{\kappa} \omega} \rmR_{\mJ}(-\frac{t}{\lambda} ) \dif t + \frac{\chi}{\upbeta} \rmR_{\mJ}(-\frac{\chi}{\lambda} \omega) + \left[q- \frac{\lambda_0}{\lambda} \tilde{\chi}_b \right] \frac{\dif}{\dif \omega} \left[ \omega \rmR_{\mJ}(-\frac{\tilde{\chi}_b}{\lambda} \omega) \right].  \label{eq:d-20}
\end{align}
Therefore, the vector $\bmu^\star$, for a given $\upbeta$, is set as
\begin{align}
\bmu^\star=\arg\min_{\bmu} \sfF(\upbeta;\bmu) 
\end{align}
with $\bmu\in\setS_{\bmu}$ where $\setS_{\bmu}$ is the set of non-negative real vectors satisfying the constraint in \eqref{eq:d-2.1}. By substituting \eqref{eq:d-3.1a}-\eqref{eq:d-3.1c} in \eqref{eq:d-20}, $\bmu^\star$ reduces to
\begin{align}
\bmu^\star=\arg\min_{\bmu} \left\lbrace \frac{1}{2\lambda} \left[ \int_0^1 \rmF^{\upbeta}(\omega;\bmu) \dif \omega \right]+\sfF^{\sfR}(\upbeta;\bmu) - e \Delta(\bmu) \right\rbrace \label{eq:d-19}
\end{align}
with $\bmu\in\setS_{\bmu}$ where $\Delta(\cdot)$ reads
\begin{align}
\Delta(\bmu)\coloneqq \frac{1}{e} \left\lbrace\sum\limits_{\kappa=1}^b \frac{1}{\mu_\kappa} \left[ \tilde{e}_\kappa \tilde{\chi}_\kappa - \tilde{e}_{\kappa-1} \tilde{\chi}_{\kappa-1} \right] + \left[\frac{\tilde{e}_0 \tilde{\chi}_0}{\upbeta}+ \tilde{e}_b q - \frac{f^2}{2} \tilde{\chi}_b \right] \right\rbrace
\end{align}
with $\tilde{e}_0 \coloneqq e$ and 
\begin{align}
\tilde{e}_\kappa \coloneqq e-\sum_{\varsigma=1}^{\kappa} \mu_\varsigma \frac{g^2_\varsigma}{2} \label{eq:d-21}
\end{align}
for $\kappa\in[1:b]$. The vector $\bmu^\star$ is then determined such that it minimizes the free energy . Finally by taking the limit $\upbeta\uparrow\infty$, the zero temperature free energy is evaluated as
\begin{align}
\sfF^0=\frac{1}{2\lambda} \left[ \int_0^1 \rmF^{\infty}(\omega) \dif \omega -  \rmF^{\infty} (1) \right] - \frac{1}{\mu_b} \E \int \log \left[\int \Lambda_b \md z_b\right] \ \md z_0 \label{eq:d-22}
\end{align}
where we define 
\begin{align}
\rmF^{\infty}(\omega)\coloneqq \lim_{\upbeta\uparrow\infty}\rmF^{\upbeta}(\omega;\bmu^\star). \label{eq:d-23}
\end{align}
Denoting $\lams\coloneqq \left[2e\right]^{-1}$, $\lams_0\coloneqq \left[4e^2\right]^{-1} f^2$ and $\lams_\kappa\coloneqq \left[4e^2\right]^{-1} g_\kappa^2$ for $\kappa\in[1:b]$, and defining the sequence $\{ \zeta_\kappa \}$ such that $\zeta_0=1$ and
\begin{align}
\zeta_\kappa \coloneqq 1-\sum_{\varsigma=1}^\kappa \mu_{\varsigma} \frac{\lams_\varsigma}{\lams}
\end{align}
for $\kappa\in[1:b]$, Proposition \ref{proposition:7} is concluded.

\newpage
\section{General \ac{rsb} Frequency Domain Correlation Matrix}
\label{app:e}
Consider the spin glass of replicas defined in Definition \ref{def:replica_spin}, the Hamiltonian reads
\begin{align}
\mae^\sfR(\bvv|\bxx)= (\bxx-\bvv)^{\trp} \mR (\bxx-\bvv) + u(\bvv). \label{eq:e-1}
\end{align}
where $\mR \coloneqq \mT \rmR_{\mJ}(-2 \upbeta \mT \mQ)$ is referred to as the ``frequency domain correlation matrix''. In this appendix, we show that under the general \ac{rsb} assumption on $\mQ$, including the \ac{rs} case, the frequency domain correlation matrix has the same structure with different scalar coefficients. To show that, let the correlation matrix be of the form
\begin{align}
\mQ= q_0 \mI_m + \sum_{i=1}^{b} q_i \mI_{\frac{m}{\xi_i}} \otimes \mone_{\xi_i} + q_{b+1} \mone_m \label{eq:e-2}
\end{align}
for some integer $b$ where $q_0,q_{b+1} \neq 0$. \eqref{eq:e-2} represents the $b$\ac{rsb} as well as \ac{rs} structures by setting the coefficients correspondingly. Considering $\mT$ as defined in \eqref{eq:rep-5}, $\mT\mQ$ is then written as
\begin{align}
\mT\mQ=\frac{1}{2\lambda}\left[ \mQ-\frac{\upbeta \lambda_0}{\lambda + m\upbeta \lambda_0}\mone_m\mQ \right]. \label{eq:e-3}
\end{align}
Defining the vector $\bu_{m\times 1}$ as a vector with all entries equal to $1$, it is clear that $\bu$ is an eigenvector of $\mQ$, and therefore, by denoting the eigendecomposition of $\mQ$ as $\mV\mD^{\rmQ}\mV^\trp$, $\mone_m$ reads
\begin{align}
\mone_m=\bu\bu^\trp=\mV\mD^{1}\mV^\trp \label{eq:e-4}
\end{align}
where $\mD^1$ is a diagonal matrix in which all the diagonal entries expect the entry corresponding to the eigenvector $\bu$ are zero. Consequently, \eqref{eq:e-3} reduces to
\begin{align}
\mT\mQ= \frac{1}{2\lambda} \mV \left[ \mD^{\rmQ}-\frac{\upbeta \lambda_0}{\lambda + m\upbeta \lambda_0}\mD^{1} \mD^{\rmQ} \right] \mV^\trp \label{eq:e-5}
\end{align}
which states that $\mT\mQ$ and $\mQ$ span the same eigenspace. The eigenvalues of $\mT\mQ$ and $\mQ$ are also distributed with the same frequencies. In fact, as the eigenvalue corresponding to $\bu$ occurs with multiplicity $1$, the second term on the right hand side of \eqref{eq:e-3} does not change the distribution of eigenvalues and only modifies the eigenvalue corresponding to $\bu$. Therefore, $\mT\mQ$ can be also represented as in \eqref{eq:e-2} with different scalar coefficient.

To extend the scope of the analysis to $\mR$, we note that the function $\rmR_{\mJ}(\cdot)$ is strictly increasing for any $\rmF_{\mJ}$ different from the single mass point \ac{cdf}\footnote{In the single mass point \ac{cdf}, we have $\rmF_{\mJ}(\lambda)=\mone\{\lambda \geq \mathsf{K}\}$ for some real constant $\mathsf{K}$.}\cite{zaidel2012vector}. Consequently, the eigenvalues' distribution remains unchanged, and thus,
\begin{align}
\mR= r_0 \mI_m + \sum_{i=1}^{b} r_i \mI_{\frac{m}{\xi_i}} \otimes \mone_{\xi_i} + r_{b+1} \mone_m. \label{eq:e-6}
\end{align}
for some real $\left\lbrace r_i \right\rbrace$. In the case that $\rmF_{\mJ}$ is the single mass point \ac{cdf}, the $\rmR$-transform becomes a constant function which results in $\rmR_{\mJ}(-2 \upbeta \mT \mQ)=\mathsf{K}\mI_m$ for some constant $\mathsf{K}$. Therefore, $\mR=\sfK \mT$ which is again represented as in \eqref{eq:e-6} by setting $r_i = 0$ for $i \in [1:b]$. This concludes that $\mR$ has the same structure as $\mQ$ for any $\rmF_{\mJ}$.

\newpage
\section{Asymptotics of Spherical Integral}
\label{app:f}
Consider $\mu_{n}^{\zeta}$ to be the Haar measure on the orthogonal group $\setO_n$ for $\zeta=1$, and on the unitary group $\setU_n$ for $\zeta=2$. Let $\mG_n$ and $\mD_n$ be $n\times n$ matrices; then, the integral of the form
\begin{align}
\rmI_n^{\zeta}(\mG_n, \mD_n) \coloneqq \int e^{n \tr{\mU \mG_n \mU^{\dagger}\mD_n}} \dif \mu_{n}^{\zeta}(\mU), \label{eq:f-1}
\end{align}is known as the ``spherical integral''. This integral has been extensively studied in the mathematics literature, as well as physics where it is often called ``Harish-Chandra'' or ``Itzykson \& Zuber'' integral. In a variety of problems, such as ours, the evaluation of spherical integrals in asymptotic regime is interesting, and therefore, several investigations have been done on this asymptotics. In \cite{guionnet2002large}, the asymptotics of the integral has been investigated when the matrices $\mG_n$ and $\mD_n$ have $n$ distinct eigenvalues with converging spectrums, and under some assumptions, a closed form formula has been given; however, the final formula in \cite{guionnet2002large} is too complicated and hard to employ. In \cite{guionnet2005fourier}, the authors showed that, for a low-rank $\mG_n$, the asymptotics of the integral can be written directly in terms of the $\mathrm{R}$-transform corresponding to the asymptotic eigenvalue distribution of $\mD_n$. As long as the replica analysis is being considered, we can utilize the result from \cite{guionnet2005fourier}, since the number of replicas can be considered to be small enough.

In \cite{guionnet2005fourier}, Theorem 1.2, it is shown that when $\mG_n$ is a rank-one matrix, under the assumption that the spectrum of $\mD_n$ asymptotically converges to a deterministic \ac{cdf} $\rmF_{\mD}$ with compact and finite length support, the asymptotics of the integral can be written in terms of the $\mathrm{R}$-transform $\rmR_{\mD}(\cdot)$ as
\begin{align}
\lim_{n \uparrow \infty} \frac{1}{n} \log \rmI_n^{\zeta}(\mG_n, \mD_n)= \int_{0}^{\uptheta} \rmR_{\mD}(\frac{2 \omega}{\zeta}) \dif \omega, \label{eq:f-2}
\end{align}in which $\uptheta$ denotes the single nonzero eigenvalue of $\mG_n$. The authors further showed in Theorem 1.7 that in the case of $\mathrm{rank}(\mG_n)=\mao(\sqrt{n})$, under the same assumption as in Theorem 1.2, the spherical integral asymptotically factorizes into product of rank-one integrals, and therefore, 
\begin{align}
\lim_{n \uparrow \infty} \frac{1}{n} \log \rmI_n^{\zeta}(\mG_n, \mD_n)= \sum_{i=1}^m \int_{0}^{\uptheta_i} \rmR_{\mD}(\frac{2 \omega}{\zeta}) \dif \omega, \label{eq:f-3}
\end{align}with $\{ \uptheta_i \}$ denoting the nonzero eigenvalues of $\mG_n$ for $i\in[1:m]$, and $m=\mathrm{rank}(\mG_n)$.

In Appendix \ref{app:a}, one can employ \eqref{eq:f-3} in order to evaluate the asymptotics over the system matrix consistent to the system setup illustrated in Section \ref{sec:problem_formulation}. Moreover, by using the above discussion, the investigations in Appendix \ref{app:a} can be extended to the case of complex variables. More about the spherical integral and its asymptotics can be found in \cite{guionnet2005fourier}, and the references therein.

%acronyms
\label{list:acronyms}
\begin{acronym}
\acro{iid}[i.i.d.]{independent and identically distributed}
\acro{pmf}[PMF]{Probability Mass Function}
\acro{cdf}[CDF]{Cumulative Distribution Function}
\acro{pdf}[PDF]{Probability Density Function}
\acro{rs}[RS]{Replica Symmetry}
\acro{1rsb}[1RSB]{One-Step Replica Symmetry Breaking}
\acro{brsb}[$b$RSB]{$b$-Steps Replica Symmetry Breaking}
\acro{rsb}[RSB]{Replica Symmetry Breaking}
\acro{mse}[MSE]{Mean Square Error}
\acro{mmse}[MMSE]{Minimum Mean Square Error}
\acro{map}[MAP]{Maximum-A-Posterior}
\acro{rhs}[r.h.s.]{right hand side}
\acro{lhs}[l.h.s.]{left hand side}
\acro{wrt}[w.r.t.]{with respect to}
\acro{lln}[LLN]{Law of Large Numbers}
\acro{mpm}[MPM]{Marginal-Posterior-Mode}
\acro{mimo}[MIMO]{Multiple-Input Multiple-Output}
\acro{awgn}[AWGN]{Additive White Gaussian Noise}
\acro{cdma}[CDMA]{Code Division Multiple Access}
\acro{amp}[AMP]{Approximate Message Passing}
\end{acronym}
%\section{}
%% use section* for acknowledgment
%\section*{Acknowledgment}



\newpage
\bibliographystyle{IEEEtran}
%\begin{bibliography}{99}

\bibliography{ref}

%\end{bibliography}


%\end{IEEEbiography}
%% insert where needed to balance the two columns on the last page with
%%\newpage


%\vfill


\end{document}

%% file: commands.tex
\newcommand{\setX}{\mathbbmss{X}}
\newcommand{\setA}{\mathbbmss{A}}
\newcommand{\setP}{\mathbbmss{P}}
\newcommand{\setR}{\mathbbmss{R}}
\newcommand{\setV}{\mathbbmss{V}}
\newcommand{\setS}{\mathbbmss{S}}
\newcommand{\setW}{\mathbbmss{W}}
\newcommand{\setJ}{\mathbbmss{J}}
\newcommand{\setZ}{\mathbbmss{Z}}

\newcommand{\setO}{\mathbbmss{O}}
\newcommand{\setU}{\mathbbmss{U}}

\newcommand{\rmc}{\mathrm{c}}
\newcommand{\rmp}{\mathrm{p}}
\newcommand{\rmq}{\mathrm{q}}
\newcommand{\rmF}{\mathrm{F}}
\newcommand{\rmE}{\mathrm{E}}
\newcommand{\rmH}{\mathrm{H}}
\newcommand{\rmR}{\mathrm{R}}
\newcommand{\rmg}{\mathrm{g}}
\newcommand{\rmQ}{\mathrm{Q}}
\newcommand{\rmT}{\mathrm{T}}

\newcommand{\rmI}{\mathrm{I}}

\newcommand{\lams}{\lambda^{\mathsf{s}}}

\newcommand{\rsim}{\mathsf{Sim}^\mathsf{R}}

\newcommand{\sfF}{\mathsf{F}}
\newcommand{\sfI}{\mathsf{I}}
\newcommand{\sfT}{\mathsf{T}}
\newcommand{\sfC}{\mathsf{C}}
\newcommand{\sfM}{\mathsf{M}}
\newcommand{\sfQ}{\mathsf{Q}}
\newcommand{\sfE}{\mathsf{E}}
\newcommand{\sfD}{\mathsf{D}}

\newcommand{\sfr}{\mathsf{r}}
\newcommand{\sfR}{\mathsf{R}}
\newcommand{\sfd}{\mathsf{d}}
\newcommand{\sfZ}{\mathsf{Z}}
\newcommand{\sfL}{\mathsf{L}}
\newcommand{\sfK}{\mathsf{K}}

\newcommand{\mae}{\mathcal{E}}
\newcommand{\maz}{\mathcal{Z}}

\newcommand{\man}{\mathcal{N}}
\newcommand{\mam}{\mathcal{M}}

\newcommand{\mao}{\mathcal{O}}

\newcommand{\mg}{\mathcal{G}}
\newcommand{\mai}{\mathcal{I}}

\newcommand{\bxx}{\mathbf{x}}
\newcommand{\bss}{\mathbf{s}}
\newcommand{\bvv}{\mathbf{v}}

\newcommand{\bmu}{\boldsymbol{\mu}}
\newcommand{\bx}{{\boldsymbol{x}}}
\newcommand{\hx}{{\hat{x}}}

\newcommand{\hv}{{\hat{v}}}
\newcommand{\htt}{{\hat{t}}}
\newcommand{\vv}{\mathrm{v}}
\newcommand{\hxx}{\hat{\mathrm{x}}}
\newcommand{\xx}{\mathrm{x}}

\newcommand{\btv}{\boldsymbol{\tilde{v}}}
\newcommand{\bhx}{{\boldsymbol{\hat{x}}}}

\newcommand{\bz}{{\boldsymbol{z}}}

\newcommand{\bu}{{\boldsymbol{u}}}
\newcommand{\bv}{{\boldsymbol{v}}}

\newcommand{\bgg}{{\mathbf{g}}}
\newcommand{\baa}{{\mathbf{a}}}

\newcommand{\dif}{\mathrm{d}}

\newcommand{\by}{{\boldsymbol{y}}}

\newcommand{\trp}{\mathsf{T}}

\newcommand{\mA}{\mathbf{A}}
\newcommand{\mR}{\mathbf{R}}

\newcommand{\mI}{\mathbf{I}}
\newcommand{\mone}{\mathbf{1}}
\newcommand{\mJ}{\mathbf{J}}
\newcommand{\mG}{\mathbf{G}}
\newcommand{\mQ}{\mathbf{Q}}
\newcommand{\mS}{\mathbf{S}}
\newcommand{\mU}{\mathbf{U}}
\newcommand{\mD}{\mathbf{D}}
\newcommand{\mM}{\mathbf{M}}

\newcommand{\mV}{\mathbf{V}}
\newcommand{\mtV}{\tilde{\mathbf{V}}}
\newcommand{\mB}{\mathbf{B}}
\newcommand{\mT}{\mathbf{T}}

\newcommand{\md}{\mathrm{D}}
\newcommand{\E}{\mathsf{E}\hspace{.5mm}}

\newcommand{\mse}{\mathsf{MSE}}

\newcommand{\poster}[2]{\rmq_{\lams} ( #1 | #2 ) }

\newcommand{\norm}[1]{\lVert #1 \rVert}
\newcommand{\abs}[1]{\lvert #1 \rvert}
\newcommand{\tr}[1]{\mathrm{Tr} \{ #1 \}}

\newcommand{\Tr}{\mathrm{Tr} }

\renewenvironment{proof}{{\bfseries Proof. }}{\hfill$\blacktriangle$\vspace{2mm}}
\declaretheorem[style=definition,qed=$\blacksquare$,numberwithin=section]{definition}
\declaretheorem[style=proposition,qed=$\blacksquare$,numberwithin=section]{proposition}
\declaretheorem[style=lemma,qed=$\blacksquare$,numberwithin=section]{lemma}
\declaretheorem[style=assumption,qed=$\blacksquare$]{assumption}
\declaretheorem[style=remark,numberwithin=section]{remark}

\algnewcommand\algorithmicset{\textbf{set}}
\algnewcommand\Set{\item[\algorithmicset]}

\algnewcommand\algorithmicInitiate{\textbf{initiate}}
\algnewcommand\Initiate{\item[\algorithmicInitiate]}

\algnewcommand\algorithmicEvaluate{\textbf{evaluate}}
\algnewcommand\Evaluate{\item[\algorithmicEvaluate]}

\algnewcommand\algorithmicimplement{\textbf{implement}}
\algnewcommand\Implement{\item[\algorithmicimplement]}

\algnewcommand\algorithmicbegin{\textbf{begin}}
\algnewcommand\BEGIN{\item[\algorithmicbegin]}

\algnewcommand\algorithmicEnd{\textbf{end}}
\algnewcommand\ENDall{\item[\algorithmicEnd]}

\algnewcommand\algorithmicoutput{\textbf{output}}
\algnewcommand\Output{\item[\algorithmicoutput]}

\newcounter{bar}
\newcommand{\exmpl}[1]{\refstepcounter{bar}\label{#1} \textbf{Example \thebar. }}